\pdfoutput=1
\documentclass[11pt]{article}
\usepackage{xcolor}
\usepackage{graphicx}
\usepackage{amsmath,amsfonts,amssymb,graphics,amsthm}
\usepackage{hyperref}
\usepackage{comment}
\usepackage{tabularx}
\usepackage[protrusion=true,expansion=true]{microtype}
\usepackage{enumerate}
\usepackage{bbm}
\usepackage{mathrsfs}
\usepackage[margin=1in]{geometry}
\usepackage[shortlabels]{enumitem}
\hypersetup{
    colorlinks=false,
    linktocpage,
    }

\numberwithin{equation}{section}

\newtheorem{theorem}{Theorem}[section]
\newtheorem{corollary}[theorem]{Corollary}
\newtheorem{lemma}[theorem]{Lemma}
\newtheorem{proposition}[theorem]{Proposition}
\newtheorem{prop}[theorem]{Proposition}

\newtheorem{remark}[theorem]{Remark}
\newtheorem{definition}[theorem]{Definition}

\theoremstyle{remark}

\def\@rst #1 #2other{#1}
\newcommand\MR[1]{\relax\ifhmode\unskip\spacefactor3000 \space\fi
  \MRhref{\expandafter\@rst #1 other}{#1}}
\newcommand{\MRhref}[2]{\href{http://www.ams.org/mathscinet-getitem?mr=#1}{MR#2}}

\def\MR#1{\href{http://www.ams.org/mathscinet-getitem?mr=#1}{MR#1}}

\newcommand{\bfb}{{\mathbf b}}
\newcommand{\bft}{{\mathbf t}}

\newcommand{\C}{\mathbbm{C}}
\newcommand{\D}{\mathbbm{D}}
\newcommand{\E}{\mathbbm{E}}
\newcommand{\N}{\mathbbm{N}}

\newcommand{\R}{\mathbbm{R}}
\newcommand{\T}{\mathbbm{T}}
\renewcommand{\P}{\mathbbm{P}}
\newcommand{\bbH}{\mathbbm{H}}

\newcommand{\eps}{\varepsilon}

\newcommand{\disk}{\mathrm{disk}}

\newcommand{\sm}{\mathsf{m}}

\newcommand{\LF}{\mathrm{LF}}

\newcommand{\QD}{\mathrm{QD}}
\newcommand{\QS}{\mathrm{QS}}
\newcommand{\QA}{\mathrm{QA}}
\newcommand{\QP}{\mathrm{QP}}
\newcommand{\Ug}{\Upsilon_{\frac\gamma2}}

\newcommand{\CCLE}{{C^{\CLE}_\kappa(\lambda_1,\lambda_2,\lambda_3)}}
\newcommand{\lexp}{{\beta}}
\newcommand{\CR}{\mathrm{CR}}

\let\Re\undefined
\DeclareMathOperator{\Re}{Re}
\let\Im\undefined
\DeclareMathOperator{\Im}{Im}

\DeclareMathOperator{\SLE}{SLE}
\DeclareMathOperator{\CLE}{CLE}

\def\cX{\mathcal{X}}

\def\cS{\mathcal{S}}

\def\cM{\mathcal{M}}
\def\cL{\mathcal{L}}

\def\cJ{\mathcal{J}}
\def\cI{\mathcal{I}}

\def\cF{\mathcal{F}}
\def\cE{\mathcal{E}}
\def\cD{\mathcal{D}}
\def\cC{\mathcal{C}}

\def\alb#1\ale{\begin{align*}#1\end{align*}}
\def\allb#1\alle{\begin{align}#1\end{align}}

\newcommand{\aryb}{\begin{eqnarray*}}
\newcommand{\arye}{\end{eqnarray*}}
\def\alb#1\ale{\begin{align*}#1\end{align*}}
\newcommand{\eqb}{\begin{equation}}
\newcommand{\eqe}{\end{equation}}
\newcommand{\eqbn}{\begin{equation*}}
\newcommand{\eqen}{\end{equation*}}

\newcommand{\ol}{\overline}
\newcommand{\ul}{\underline}

\newcommand{\frk}{\mathfrak}

\newcommand{\rta}{\rightarrow}

\newcommand{\wt}{\widetilde}
\newcommand{\wh}{\widehat}

\newcommand{\bdy}{\partial}

\newcommand{\lp}{\mathrm{loop}}
\newcommand{\Cont}{\mathrm{Cont}}
\newcommand{\conf}{\mathrm{conf}}

\newcommand{\Weld}{\operatorname{Weld}}

\newcommand{\GQD}{\mathrm{GQD}}

\newcommand{\GA}{\mathrm{GA}}
\newcommand{\GP}{\mathrm{GP}}

\let\Re\undefined

\let\Im\undefined

\def\cX{\mathcal{X}}

\def\cS{\mathcal{S}}

\def\cM{\mathcal{M}}
\def\cL{\mathcal{L}}

\def\cJ{\mathcal{J}}
\def\cI{\mathcal{I}}

\def\cF{\mathcal{F}}
\def\cE{\mathcal{E}}
\def\cD{\mathcal{D}}
\def\cC{\mathcal{C}}

\def\alb#1\ale{\begin{align*}#1\end{align*}}
\def\allb#1\alle{\begin{align}#1\end{align}}

\def\alb#1\ale{\begin{align*}#1\end{align*}}

\newcommand{\Re}{\mathrm{Re}}
\newcommand{\Im}{\mathrm{Im}}

\let\originalleft\left
\let\originalright\right
\renewcommand{\left}{\mathopen{}\mathclose\bgroup\originalleft}
\renewcommand{\right}{\aftergroup\egroup\originalright}

\DeclareMathAlphabet{\mathpzc}{OT1}{pzc}{m}{it}

\begin{document}

\title{Integrability of Conformal Loop Ensemble: Imaginary DOZZ Formula and Beyond}
\author{
\begin{tabular}{c}Morris Ang\\[-3pt]\small UC San Diego \end{tabular}\; 
\begin{tabular}{c}Gefei Cai\\[-3pt]\small Peking University  \end{tabular}\; 
\begin{tabular}{c}Xin Sun\\[-3pt]\small Peking University  \end{tabular}\;
\begin{tabular}{c}Baojun Wu\\[-3pt]\small Peking University\end{tabular}
}

\date{  }

\maketitle

\begin{abstract}
The scaling limit of the probability that $n$ points are on the same cluster for 2D critical percolation 
is believed to be governed by a conformal field theory (CFT). Although this is not fully understood, Delfino and  Viti (2010) made a remarkable prediction on the exact value of a properly normalized three-point probability. 
It is expressed in terms of the imaginary DOZZ formula of Schomerus, Zamolodchikov and Kostov-Petkova, which extends the structure constants of minimal model CFTs to continuous parameters.
Later, similar conjectures were made for scaling limits of random cluster models and O$(n)$ loop models, representing 
certain three-point observables in terms of the imaginary DOZZ formula. Since 
the scaling limits of these models can be described by the conformal loop ensemble (CLE), such conjectures can be formulated as exact statements on CLE observables. In this paper, we prove Delfino and  Viti's conjecture on percolation as well as a conjecture of Ikhlef, Jacobsen and Saleur (2015) on the nesting loop statistics of CLE. Our proof is based on the coupling between CLE and Liouville quantum gravity on the sphere, and is inspired by the fact that after reparametrization, the imaginary DOZZ formula is the reciprocal of the three-point function of Liouville CFT. Recently, Nivesvivat, Jacobsen and Ribault systematically studied a CFT with a large class of CLE observables as its correlation functions, including the ones from these two conjectures.
We believe that our framework admits sufficient flexibility to exactly solve the three-point functions for CLE observables with natural geometric interpretations, including those from this CFT.   As a demonstration, we solve the case corresponding to three points lying on the same loop, where the answer is a variant of the imaginary DOZZ formula.
\end{abstract}

\setcounter{tocdepth}{1}
\tableofcontents

\section{Introduction}
Conformal field theory (CFT) and Schramm-Loewner evolution (SLE) are two successful approaches to the scaling limits of two-dimensional lattice models at criticality. In the CFT approach, the seminal work~\cite{bpz-conformal-symmetry} developed the conformal bootstrap formalism and solved a series of CFTs called minimal models, which include natural observables from the 2D Ising model.
In this approach, higher order correlation functions of a CFT are expressed in terms of the three-point correlation functions that are encoded by the so-called structure constants and data which are determined by conformal symmetry.  For minimal models, the parameters in the formula for the structure constants vary on finite sets since the theory is rational. 
Zamolodchikov found a natural extension of the formula where the parameters vary continuously~\cite{zamolodchikov-gmm}. This is the so-called imaginary DOZZ formula. It is closely related to the DOZZ formula, which is the formula for the structure constants of another important CFT --- the Liouville CFT, see Section~\ref{section:dozz}. 

The scaling limit of 2D critical percolation is expected to be closely related to CFT.
Based on this assumption, Cardy derived a formula for the crossing probabilities for percolation in  rectangles~\cite{cardy-formula}; this formula was later rigorously
proved by Smirnov~\cite{smirnov-cardy}. However, the precise CFT describing percolation is not a minimal model and remains mysterious. Since 2010, it has been noticed in physics that the imaginary DOZZ formula can express certain geometric observables in percolation and its close variants, random cluster models and O$(n)$ loop models. The first such example, found by Delfino and Viti~\cite{Delfino_2010}, is the scaling limit of the probability that three points lie in the same percolation cluster.  Later, Ikhlef, Jacobsen and Saleur gave a similar prediction for the nesting loop statistics of the O$(n)$ loop model~\cite{ijs-cle}.
In the SLE approach, the full scaling limit of a critical lattice model is described by the conformal loop ensemble (CLE), a collection of SLE-type loops.  In light of this, the aforementioned geometric observables can be defined in terms of CLE. In this paper, we confirm the Delfino-Viti conjecture for percolation and the Ikhlef-Jacobsen-Saleur conjecture for CLE, verifying that these statistics are indeed expressed by the imaginary DOZZ formula as predicted in physics.
See Theorems \ref{thm:dv} and \ref{thm:nesting}.

Our proofs rely on the coupling between CLE and Liouville quantum gravity (LQG), which enjoys two sources of exact solvability. 
First, LQG surfaces coupled with CLE inherit a rich integrable structure
from random planar maps decorated with the O$(n)$ loop model, as exhibited  in~\cite{bbck-growth-frag,ccm-perimeter-cascade,msw-non-simple,msw-cle-lqg}. Second, the field  theory describing LQG called Liouville  conformal field theory  (Liouville CFT) is integrable. 
In particular, the DOZZ formula~\cite{krv-dozz} and FZZ formula~\cite{ARS-FZZ} for its structure constants are crucial inputs to our results. 
The high-level strategy of combining inputs from both sources of solvability was first proposed in~\cite{AHS-SLE-integrability} and successfully implemented in several problems~\cite{nolin2024backboneexponenttwodimensionalpercolation,ang2024boundarytouchingprobabilitynestedpath, nonsimple-welding,ars-annuli}. In the setting of the Delfino-Viti conjecture, the main new difficulty lies in identifying the right set-up where the three-point function can be solved. This requires understanding the natural measures defined on the carpet or gasket of CLE. For the Ikhlef-Jacobsen-Saleur conjecture, the main new difficulty is to solve an LQG surface having the pair of pants topology.

Although our framework   requires case-by-case analysis, we believe it admits sufficient flexibility to solve the 
three-point function for any CLE observables having natural geometric definitions. In fact, as a warm-up for our proof of the Delfino-Viti conjecture, we first solve the analogous problem where three points are on the same percolation interface instead of on the same cluster. The answer is a variant of the imaginary DOZZ formula; see Theorem \ref{thm:mag}. This was an open question in the physics community. Recently, through private communication,  Nivesvivat, Ribault and Jacobsen informed us that they found the same formula using both the (non-rigorous) conformal bootstrap method from~\cite{Nivesvivat:2023kfp} and a numerical method based on the Temperley-Lieb algebra. We plan to investigate other  three-point functions of CFT significance in the future, such as the ones appeared in \cite{Nivesvivat:2023kfp,camia2024,baverez2024cftsleloopmeasures}. See Section~\ref{sec: outlook} for further discussion.

\subsection{The Liouville and imaginary DOZZ formulae}\label{section:dozz}

Liouville CFT was originally introduced by Polyakov~\cite{polyakov-qg1} and recently  made mathematically rigorous in~\cite{dkrv-lqg-sphere} and follow-up works~\cite{hrv-disk,drv-torus, grv-higher-genus,remy-annulus}; see Section~\ref{subsubsec:GFF} for more background.
Its three-point correlation function on the sphere depends on the three points $z_1, z_2, z_3 \in \C$ and parameters $\alpha_1, \alpha_2, \alpha_3 \in \R$. The dependence on $(z_1, z_2, z_3)$ is trivial by conformal covariance, and the dependence on $(\alpha_1, \alpha_2, \alpha_3)$ is encoded by the structure constants of Liouville CFT.
It  was first proposed in theoretical physics~\cite{do-dozz,zz-dozz} and then rigorously proved in~\cite{krv-dozz} that 
the structure constant has the following remarkable expression known as  the DOZZ formula.  
Suppose  $\gamma\in (0,2)$ is the coupling parameter for Liouville CFT and $Q=\frac{\gamma}{2} +\frac{2}{\gamma}$ is its background charge.
The DOZZ formula is defined in terms of the Barnes double gamma function $\Gamma_b$ \cite[Section 5]{rz-boundary}. Let $Q_b=b+\frac{1}{b}$ and $\Re(b)>0$, then $\Gamma_b(z)$ is the unique meromorphic function that admits the integral representation
$$
\log \Gamma_b(z)=\int_0^{\infty} \frac{\mathrm{d} t}{t}\left(\frac{\mathrm{e}^{\frac{t}{2}(Q_b-2 z)}-1}{4 \sinh \left(\frac{b t}{2}\right) \sinh \left(\frac{t}{2 b}\right)}-\frac{1}{8}(Q_b-2 z)^2 \mathrm{e}^{-t}-\frac{Q_b-2 z}{2 t}\right)\quad \textrm{when }
\operatorname{Re}(z)>0.
$$
The DOZZ formula (with cosmological constant being 1, see Proposition~\ref{lem-sph-area-law}) is given by
\begin{equation}\label{eq:DOZZ}
C^\mathrm{DOZZ}_\gamma(\alpha_1, \alpha_2, \alpha_3) = \left( \pi(\frac\gamma2)^{2-\gamma^2/2} \ell(\frac{\gamma^2}{4}) \right)^{\frac{2Q-\ol \alpha}{\gamma}}\frac{\Ug'(0) \Ug(\alpha_1)\Ug(\alpha_2)\Ug(\alpha_3)}{\Ug(\frac{\ol \alpha}2 - Q) \Ug(\frac{\ol \alpha}2 - \alpha_1)\Ug(\frac{\ol \alpha}2 - \alpha_2)\Ug(\frac{\ol \alpha}2 - \alpha_3)},
\end{equation}
where $\ol \alpha= \sum \alpha_j$, $\ell(x):=\Gamma(x)/\Gamma(1-x)$ and $\Ug(z)=\frac{1}{\Gamma_{\frac{\gamma}{2}}(z)\Gamma_{\frac{\gamma}{2}}(Q-z)}$. 

The imaginary DOZZ formula is a closely related formula proposed by~\cite{Schomerus2003,zamolodchikov-gmm,Kostov-DOZZ}:
\begin{align}\label{eq:imdozz}
C_{b}^{\rm ImDOZZ}(\wh\alpha_1,\wh\alpha_2,\wh\alpha_3)= A\Upsilon_b\left(2 b-b^{-1}+\sum_{j=1}^3 \wh\alpha_j\right)\prod_{i=1}^3\frac{ \Upsilon_b\left(\wh\alpha_1+\wh\alpha_2+\wh\alpha_3-2\wh\alpha_i+b\right) }{\left[\Upsilon_b\left( 2\wh\alpha_i+b\right) \Upsilon_b\left( 2\wh\alpha_i+2 b-b^{-1}\right)\right]^{1 / 2}}
\end{align}
where the normalization factor 
$
A=\frac{b^{b^{-2}-b^2-1}\left[\ell\left(b^2\right) \ell\left(b^{-2}-1\right)\right]^{1 / 2}}{\Upsilon_b(b)}
$  makes  $C_{b}^{\rm ImDOZZ}(\wh\alpha,\wh\alpha,0)=1$.
The formula~\eqref{eq:imdozz} can be obtained by solving Teschner's shift relations~\cite{Teschner-DOZZ} for the Liouville  DOZZ formula with the shifts being imaginary; see~\cite{Schomerus2003,zamolodchikov-gmm}. It was motivated by the problem of extending the minimal model CFT to continuously varying parameters. Indeed, when specializing to minimal model parameters, the imaginary DOZZ formula recovers the structure constants of these CFTs.
This formula has been used to express the structure constants of various CFTs with central charge less than 1. Understanding these  CFTs remains an active topic in physics; see~\cite{Witten-DOZZ,Ribault-Santachiara,Bautista2019} and references therein.
As observed  in~\cite{zamolodchikov-gmm} and~\cite{Kostov-DOZZ},
the product of the imaginary DOZZ formula and the normalized Liouville DOZZ formula is given by
\begin{equation}\label{legfactor}
\frac{C_{\gamma}^{\rm DOZZ}(\alpha_1,\alpha_2,\alpha_3)\sqrt{C_{\gamma}^{\rm DOZZ}(\gamma,\gamma,\gamma)}}{\sqrt{\prod_{i=1}^3 C_{\gamma}^{\rm DOZZ}(\alpha_i,\alpha_i,\gamma)}}   \times C_{b}^{\rm ImDOZZ}(\wh \alpha_1,\wh \alpha_2,\wh \alpha_3)=1.
\end{equation}
where $b=\frac{\gamma}{2}$ or $\frac{2}{\gamma}$ and $\wh\alpha_i=\frac{\alpha_i}{2}-b$ for $i=1,2,3$.
The motivation for this formula comes from bosonic string theory. In that context, the Liouville central charge $c_L=1+6(\frac{\gamma}{2}+\frac{2}{\gamma})^2$ and the matter central charge  $c_{m}=1-6(b-\frac{1}{b})^2$ should add up to 26, which means $b=\frac{\gamma}{2}$ or $\frac{2}{\gamma}$.  Moreover, the conformal scaling dimensions for $\alpha_i$ and $\wh \alpha_i$ adds up to 2, which is the so-called  Knizhnik-Polyakov-Zamolodchikov (KPZ) relation. Equation~\eqref{legfactor} is consistent with the probabilistic phenomenon that the SLE/LQG coupling enjoys rich solvability.

\subsection{The Delfino-Viti formula and  the Green functions of CLE clusters}\label{section: connectivity}

Consider the critical Bernoulli site percolation on the triangular lattice $\delta\T$ with mesh size $\delta$, where each site is colored open (black) or closed (white) independently with probability $\frac12$. For $n$ points $z_1, \ldots, z_n\in\C$, let $z_i^\delta$'s be their approximations on $\delta\T$. Define $P_n^\delta\left(z_1^\delta, \ldots, z_n^\delta\right)$ to be the probability that $z_1^\delta,...,z_n^\delta$ are in the same open cluster. Let $\pi_\delta=\P[0\leftrightarrow\partial B_1^\delta]$ be the probability that the origin is connected to the circle $\{|z|=1\}$. 
Then by \cite{camia2023conformal}, the limit $P_n\left(z_1, \ldots, z_n\right):=\lim _{a \rightarrow 0} \pi_\delta^{-n} P_n^\delta\left(z_1^\delta, \ldots, z_n^\delta\right)$ exists and satisfies  conformal covariance. Conformal covariance implies that for some constants $C_2,C_3$ we have $P_2(z_1,z_2)=C_2|z_1-z_2|^{-2(2-d)}$ and $P_3(z_1,z_2,z_3)=C_3|z_1-z_2|^{-(2-d)}|z_2-z_3|^{-(2-d)}|z_1-z_3|^{-(2-d)}$ with $d=\frac{91}{48}$ being the dimension of open clusters. Therefore
\begin{equation}\label{pc}
R^{\rm perc}:=\frac{P_3(z_1,z_2,z_3)}{\sqrt{P_2(z_1,z_2)P_2(z_1,z_3)P_2(z_2,z_3)}}
\end{equation} does not depend on $z_1,z_2,z_3$.
In this paper we prove the following conjecture from \cite{Delfino_2010}.
\begin{theorem}\label{thm:dv} With $b=2/\sqrt{6}$ in the imaginary DOZZ formula~\eqref{eq:imdozz}, we have
\begin{align}
   R^{\rm perc}=\sqrt{2}  C_{b}^{\rm ImDOZZ}(\frac{1}{4b}-\frac{b}{2},\frac{1}{4b}-\frac{b}{2},\frac{1}{4b}-\frac{b}{2})\approx1.02201.
\end{align}
\end{theorem}

We prove Theorem~\ref{thm:dv} using CLE$_6$. For $\kappa\in (\frac{8}{3},8)$,  CLE$_\kappa$ is a random countable collection of loops, each of which is an SLE$_\kappa$ curve~\cite{Sheffield2006ExplorationTA}. The loops are simple when $\kappa\in (\frac{8}{3},4]$  and non-simple when $\kappa\in (4,8)$. CLE satisfies conformal invariance and a spatial Markov property and is conjectured to describe the scaling limit of a wide class of 2D statistical physics models at criticality, such as the FK random cluster model and the O$(n)$ loop model. For Bernoulli site percolation on the triangular lattice, it is known that the interfaces between the open and closed clusters converge to  $\CLE_6$~\cite{camia-newman-sle6}. $\CLE_\kappa$ on a simply connected domain $D \subset \C$ was constructed by Sheffield~\cite{Sheffield2006ExplorationTA}, and full-plane  $\CLE_\kappa$ can be constructed by sending  $D \uparrow \C$. Full-plane $\CLE_\kappa$ can also be interpreted as living on the Riemann sphere $\wh \C = \C \cup \{\infty\}$, and in that setting,
the conformal invariance of full-plane  $\CLE_\kappa$ in the simple and non-simple regimes were established in~\cite{werner-sphere-cle} and \cite{gmq-cle-inversion} respectively.  

Let $\Gamma=\{\eta_i\}_{i\ge 1}$ be a full-plane $\CLE_\kappa$ with an enumeration of its loops. 
When $\kappa\in (\frac83,4]$, we say that a point $p\in \C$ is surrounded by a loop $\eta_i$ if $p$ is in the bounded connected component of $\C\setminus \eta_i$. When $\kappa\in (4,8)$ so that $\eta_i$ is non-simple, we say that $p\in \C\setminus \eta_i$ is surrounded by $\eta_i$ if $\eta_i$ has nonzero winding number with respect to $p$.
The loops in $\Gamma$ have a nested structure: for $i\neq j$, the regions surrounded by $\eta_i$ and $\eta_j$ are either disjoint, or one is a subset of the other.
Let $K_i$ be the closure of the set of points that are surrounded by $\eta_i$ but not surrounded by any loop inside $\eta_i$. We call $K_i$
the \emph{cluster} associated with $\eta_i$. In the literature, the CLE cluster is usually called the CLE carpet when $\kappa\in (\frac83,4]$ or the CLE gasket when $\kappa\in (4,8)$.

By~\cite{miller2014hausdorff,nacu-werner-carpet,ssw-radii}, the dimension of a CLE$_\kappa$ cluster is $d=2-\frac{(8-\kappa)(3 \kappa-8)}{32 \kappa}$. In \cite{miller2023existence}, Miller and Schoug constructed a $d$-dimensional measure on CLE clusters which is unique under a few assumptions including  conformal covariance. (A similar uniqueness result was also obtained by Li and the second named author~\cite{cai2022natural}.) This measure was constructed via LQG, and is conjectured to be equivalent to the $d$-dimensional Minkowski content in the spirit of~\cite{lawler-rezai-nat}.  Although the measure was constructed for the outermost cluster for CLE on the disk, the definition naturally extends to the full-plane setting. We recall the construction in Section~\ref{sec:natural-measures}. 
Let $\cM^{\rm cluster}_i$ be the Miller-Schoug measure for the cluster $K_i$  associated to $\eta_i$. 
For $n\ge 2$, we define the $n$-point Green function $G_n(z_1,z_2,..,z_n)$ of $\{\cM^{\rm cluster}_i\}_{i\ge 1}$ by requiring that for  disjoint compact sets $U_1,\cdots, U_n$ in $\C$ we have 
\begin{equation}\label{eq:greendef}
 \E[\sum_{ i\ge 1 }  \int_{U_1\times \cdots \times U_n}   \prod_{k=1}^n \cM^{\rm cluster}_i(dz_k)]  = \int_{U_1\times \cdots \times U_n}  G^{\rm cluster}_n(z_1, \cdots, z_n) \prod_{k=1}^n dz_k.   
\end{equation} 
For $\kappa=6$, the $\CLE_6$ clusters $\{K_i\}_{i\ge 1}$ describe the scaling limit of all (open and closed) macroscopic clusters of the Bernoulli site percolation on $\delta \T$.  As explained in \cite[Section 8.4]{cai2022natural}, the Miller-Schoug measure on a given cluster agrees with the scaling limit of the counting measure on the discrete cluster~\cite{gps-pivotal} up to a multiplicative constant. This gives:
\begin{proposition}\label{prop:Rperc}
 There exists a constant $\frak{C}>0$ such that for $n\ge 2$, we have $2P_n(z_1,\cdots, z_n)$ equals $\mathfrak{C}^nG^{\rm cluster}_n(z_1, \cdots, z_n)$ with $\kappa=6$.
\end{proposition}
 We provide a detailed proof of Proposition~\ref{prop:Rperc} in Appendix~\ref{percolation}. The factor $2$ comes from the fact that the definition of $P_n(z_1,\cdots, z_n)$ only involves  open clusters while for $G^{\rm cluster}_n$ we include both the open and closed ones.   By conformal covariance of the Miller-Schoug measure, we have:
\begin{lemma}\label{covariance_cluster}
We have $G_2^{\rm cluster}(z_1,z_2)=C_2^{\rm cluster} |z_1-z_2|^{-2(2-d)}$ and $G_3^{\rm cluster}(z_1,z_2,z_3)=C_3^{\rm cluster} |z_1-z_2|^{-(2-d)}|z_2-z_3|^{-(2-d)}|z_1-z_3|^{-(2-d)}$ with $d=2-\frac{(8-\kappa)(3 \kappa-8)}{32 \kappa}$ and constants $C_2^{\rm cluster}, C_3^{\rm cluster}$.
\end{lemma}
We explain Lemma~\ref{covariance_cluster} in more detail in Section~\ref{sec:natural-measures}. 
By Lemma~\ref{covariance_cluster}, we  can define
\begin{equation}\label{c}
R^{\rm cluster}(\kappa):=\frac{G^{\rm cluster}_3(z_1,z_2,z_3)}{\sqrt{G^{\rm cluster}_2(z_1,z_2)G^{\rm cluster}_2(z_1,z_3)G^{\rm cluster}_2(z_2,z_3)}}.
\end{equation}

Our next theorem expresses $R^{\rm cluster}(\kappa)$ via the imaginary DOZZ formula. 

\begin{theorem}\label{thm: connectivity}
For  $\kappa\in(\frac{8}{3},8)$ and $b=2/\sqrt{\kappa}$, 
\begin{align}\label{eq:carpet-formula}
R^{\rm cluster}(\kappa)
=C_{b}^{\rm ImDOZZ}(\frac{1}{4b}-\frac{b}{2},\frac{1}{4b}-\frac{b}{2},\frac{1}{4b}-\frac{b}{2}).
\end{align}
\end{theorem}

\begin{proof}[Proof of Theorem \ref{thm:dv} given Theorem~\ref{thm: connectivity}]
By Proposition~\ref{prop:Rperc}, the ratio $R^{\rm perc}$ from~\eqref{pc} equals  $\sqrt{2} R^{\rm cluster}(6)$. By Theorem~\ref{thm: connectivity}  we get Theorem \ref{thm:dv}.
\end{proof}
We can consider the analog of $R^{\rm perc}$ for any percolation model with $\CLE$ as its scaling limit. For
the FK-$q$ random cluster model with  $q=4\cos^2\frac{4\pi}{\kappa}\in (0,4)$, it was conjectured that the analogue of $R^{\rm perc}$ equals $\sqrt{2}R^{\rm cluster}(\kappa)$. 
If the convergence of the discrete cluster measure and $P_n^\delta(z_1^\delta,\cdots, z_n^\delta)$ to their  CLE counterparts is available, then the $\kappa\in (4,8)$ case of  Theorem~\ref{thm: connectivity}  confirms this prediction as well. For $q=2$, which is the Ising-FK model, the desired convergence was established  in \cite{camia2019conformal,camia2024conformalcovarianceconnectionprobabilities}.
As informed by the second named author of \cite{camia2024conformalcovarianceconnectionprobabilities},  the same convergence would hold for $1\le q\le 4$ once it is known that the percolation interfaces converge to CLE in the sense of Camia and Newman~\cite{camia-newman-sle6}.
For the spin cluster of the critical Ising model, whose scaling limit is described by $\CLE_3$ \cite{benoist-hongler-cle3}, it was conjectured in~\cite{Delfino:2013pma} that the analog of $R^{\rm perc}$ equals $R^{\rm cluster}(3)$. (There is no $\sqrt2$ factor in front because they consider both $+$ and $-$ spin clusters instead of a single type.) This is consistent with Theorem~\ref{thm: connectivity} for $\kappa=3$.

\subsection{The Ikhlef-Jacobsen-Saleur conjecture on the nesting loop statistics}\label{subsec:3pt-intro}

The O$(n)$ loop model is a statistical physics model where a configuration $\cL$ is a collection of disjoint loops on a graph. Given the inverse temperature $x>0$, the partition function is $Z(n)=\sum_{\cL} x^{V(\cL)}n^{N(\cL)} $, 
where $V(\cL)$ is the number of vertices occupied by all loops in $\cL$ and $N(\cL)$ is the number of loops. Let  $x_c=(2+(2-n)^{\frac{1}{2}})^{-\frac{1}{2}}$. For $n\in (0,2]$ and $x\geq x_c$, the loops in the O$(n)$ loop model on the hexagonal lattice are conjectured to converge to 
CLE$_\kappa$ where $\kappa$ satisfies  $n=-2\cos(4\pi/\kappa)$. 
The case  $x = x_c$ is called the dilute regime, and in this setting we have $\kappa \in (8/3, 4]$, whereas the case  $x > x_c$ is called the dense regime, and here $\kappa \in [4, 8)$.
The case $n=1$ and  $x=x_c$ corresponds to the spin interface of the Ising model and the convergence to $\CLE_3$ is known~\cite{benoist-hongler-cle3}. The case $n=1$ and  $x=1>x_c$ corresponds to the site percolation on the triangular lattice  hence 
the convergence to $\CLE_6$ is known~\cite{camia-newman-sle6}. 

In the O$(n)$ loop model each loop is assigned weight $n$. In~\cite{ijs-cle}, Ikhlef, Jacobsen, and Saleur considered the following modification. Fix three faces $z_1, z_2,z_3$ on the hexagonal lattice. If a loop separates $z_i$ ($i=1,2,3$) from the other two points, we assign it weight $n_i$, and we assign all other loops weight $n$. Let  $Z(n;n_1,n_2,n_3)$ be the new partition function.
It was conjectured  in~\cite{ijs-cle}  that the continuum limit of $Z(n;n_1,n_2,n_3)/Z(n)$ can be described by the imaginary DOZZ formula, where the parameters $b$ and $\wh\alpha_i$ in~\eqref{eq:imdozz}  depend on $\kappa$ and $n_i$ as follows. First, $b$ is given by
\begin{align}\label{eq:beta}
b=2/\sqrt{\kappa} \quad \textrm{where } n=-2\cos(4\pi/\kappa) \textrm{ and }  \kappa\in (8/3,8).
\end{align}
As before, we select $\kappa \in (8/3,4]$ if $x = x_c$ and $\kappa \in [4,8)$ if $x > x_c$.
Set $P_0=\frac{1}{2}(\frac{1}{b}-b )$ so that $n=2\cos(2\pi b P_0)$. Define $P_i$ by $n_i=2\cos(2\pi b P_i)$ with $2\pi b P_i\in (-\pi/2,\pi/2)$. Then $\wh\alpha_i$ is given by
\begin{equation}\label{eq:Pidef}
\wh\alpha_i=P_{0}-P_i \quad \textrm{for }  i=1,2,3.  
\end{equation}
This parametrization is from~\cite{Nivesvivat:2023kfp}, and the CFT meaning of the nest loop statistics is clarified in a prior work~\cite{Rilbault-diagonal}. See the discussion below Theorem~\ref{thm:mag}.

The existence of the scaling limit for  $Z(n;n_1,n_2,n_3)/Z(n)$ as assumed in~\cite{ijs-cle} has not been rigorously proved even when $n=1$ and $x\in \{x_c, 1\}$, as this requires a much stronger topology than known convergence results. However, the Ikhlef-Jacobsen-Saleur conjecture corresponds to the following statement on the CLE side, which we prove in Section~\ref{sec:3-pt}.

We say a (possibly non-simple) loop $\eta$ \emph{separates} $S, S' \subset \hat \C$ if $\eta \cap (S \cup S') = \emptyset$ and for every $z \in S$ and $z' \in S'$ the winding numbers of $\eta$ with respect to $z$ and $z'$ differ. This definition does not depend on the orientation of $\eta$ and is invariant under homeomorphisms of $\hat \C$.

\begin{theorem}\label{thm:nesting}
Consider full-plane  $\CLE_\kappa$ for $\kappa\in (\frac83, 8)$. Let $z_1, z_2, z_3$ be distinct points in $\C$. For $i = 1,2,3$ and a small $\eps>0$, let $X_\eps(z_i)$ be the number of loops that separate the ball $B_\eps(z_i)=\{z\in \C\mid |z-z_i|<\epsilon\}$ from the other two points. Let $n_i \in (0,2]$. For $P_{0}$ and $P_i$ defined as above~\eqref{eq:Pidef}, let $\Delta_i = 2(P_i^2-P_{0}^2)$. 
Set $\Delta_{12} = \frac12 (\Delta_1+\Delta_2 - \Delta_3)$ and likewise $\Delta_{23}, \Delta_{13}$. 
Then there is a constant $C_3(\Delta_1, \Delta_2, \Delta_3)$ such that 
\eqb\label{eq-main-limit}
\lim_{\eps \to 0} \E \left[ \prod_{i=1}^3 \eps^{-\Delta_i} (\frac{n_i}n)^{X_\eps(z_i)} \right] = C_3(\Delta_1, \Delta_2, \Delta_3) |z_1 - z_2|^{-2 \Delta_{12}} |z_1 - z_3|^{-2 \Delta_{13}} |z_2 - z_3|^{-2 \Delta_{23}}.
\eqe
For $b$ and $\wh\alpha_i$ defined in~\eqref{eq:beta} and~\eqref{eq:Pidef}, we have
\eqb \label{eq-main-idozz}
\frac{C_3 (\Delta_1, \Delta_2, \Delta_3)}{\sqrt{C_3(\Delta_1, \Delta_1, 0)C_3(\Delta_2, \Delta_2, 0)C_3(\Delta_3, \Delta_3, 0)}} = 
C_{b}^{\rm ImDOZZ}(\wh\alpha_1,\wh\alpha_2,\wh\alpha_3)
\eqe
\end{theorem}

As demonstrated in~\cite{mww-nesting,holden2022liouville}, the nesting statistics of CLE loops are closely related to the conformal radii of certain loops. To prove Theorem~\ref{thm:nesting}, the key is to obtain the joint moments of the conformal radii of $\eta_i$ viewed from $z_i$, where $\eta_i$ is the outermost loop separating $z_i$ and the other two points. See Section~\ref{sec:loop-counting} for more details. 

Note that as $n_i\to 0$ we have $\wh\alpha_i\to \frac{1}{4b}-\frac{b}{2}$ hence~\eqref{eq-main-idozz} becomes the RHS of Theorem~\ref{thm: connectivity}. This is intuitively clear as $n_1=n_2=n_3=0$ means that there is no loop separating $z_1,z_2,z_3$ hence they lie on the same CLE cluster.  However, we found it difficult to prove Theorem~\ref{thm: connectivity} directly from Theorem~\ref{thm:nesting}, partly because the definition of the Miller-Schoug measure for the CLE cluster goes through Liouville quantum gravity.

\subsection{The three-point Green function for the CLE Loops}\label{section:loop-connectivity}
Theorems~\ref{thm: connectivity} and \ref{thm:nesting} solve two instances of geometrically interesting CLE three-point functions. 
Another such example is the Green function for the natural measure on CLE loops. For $\kappa\in (\frac83,8)$, let $\Gamma$ be a full-plane CLE$_\kappa$ and $\{\eta_i\}_{i\ge 1}$ be an enumeration of all the loops. By~\cite{lawler-rezai-nat}, the $(1+\frac{\kappa}{8})$-dimensional Minkowski content defines a measure $\cM^{\rm loop}_i$ on each $\eta_i$ since it is a $\SLE_\kappa$ type loop. We can define the Green function $G^{\rm loop}_n(z_1,z_2,\cdots z_n)$ of $\{\cM^{\rm loop}_i\}_{i\ge 1}$ analogously to \eqref{eq:greendef}. 
Similarly to Lemma~\ref{covariance_cluster},
we have $G_2^\lp(z_1,z_2)=C_2^\lp |z_1-z_2|^{-2(2-d)}$ and $G_3^\lp(z_1,z_2,z_3)=C_3^\lp |z_1-z_2|^{-(2-d)}|z_2-z_3|^{-(2-d)}|z_1-z_3|^{-(2-d)}$ with $d=1+\frac{\kappa}{8}$ and some constants $C_2^\lp$ and  $C_3^\lp$.
Therefore, the ratio 
$$R^{\rm loop}(\kappa):=\frac{G_3^\lp(z_1,z_2,z_3)}{\sqrt{G_2^\lp(z_1, z_2) G_2^\lp (z_2,z_3) G_2^\lp(z_3,z_1)}}$$ 
does not depend on $z_1,z_2,z_3$. The ratio $R^{\rm loop}(\kappa)$ describes the scaling limit of the normalized probability that three edges are on the same loop in the planar O$(n)$ loop model on the hexagonal lattice. For the case $n=1$ and $x=1$ corresponding to the Bernoulli site percolation, this limit can be established by the same type of argument in~\cite{camia2023conformal}  based on~\cite{gps-pivotal}.  In physics, this was first considered by Estienne and Ikhlef~\cite{Estienne:2015sua} under the name of the three-point function for the magnetic operator of weight $(1,0)$ for the loop model, while the nesting loop statistics was called electric operators. The exact evaluation of $R^{\rm loop}(\kappa)$ was left open in~\cite{Estienne:2015sua}.
In Theorem~\ref{thm:mag}  below we express $R^{\rm loop}(\kappa)$ in terms of an extension of the imaginary DOZZ formula proposed by Nivesvivat, Jacobsen, and Rilbault~\cite{Nivesvivat:2023kfp}. Let 
$$C_{(r_1,s_1),(r_2,s_2),(r_3,s_3)}=\frac{1}{\prod_{\epsilon_1, \epsilon_2, \epsilon_3= \pm} \Gamma_b\left(\frac{b+b^{-1}}{2}+\left|\sum_{i=1}^3 \epsilon_i r_i\right| \frac{b}{2}+\left(\sum_{i=1}^3 \epsilon_i s_i\right) \frac{1}{2 b}\right)},$$
\begin{equation}\label{nor njr}
    \omega^{(b)}_{(r_1,s_1),(r_2,s_2),(r_3,s_3)}=\frac{C_{(r_1,s_1),(r_2,s_2),(r_3,s_3)}\sqrt{C_{(0,1-b^2),(0,1-b^2),(0,1-b^2)}}}{\sqrt{C_{(0,1-b^2),(r_2,s_2),(r_2,s_2)}C_{(r_3,s_3),(0,1-b^2),(r_3,s_3)}C_{(r_1,s_1),(r_1,s_1),(0,1-b^2)}}}.
\end{equation}
Using the shift equations
\begin{equation}\label{eq:shift double gamma}
\frac{\Gamma_b(x+b)}{\Gamma_b(x)}=\sqrt{2 \pi} \frac{b^{b x-\frac{1}{2}}}{\Gamma(b x)} \quad \textrm{and} \quad \frac{\Gamma_b\left(x+b^{-1}\right)}{\Gamma_b(x)}=\sqrt{2 \pi} \frac{b^{\frac{1}{2}-b^{-1} x}}{\Gamma\left(b^{-1} x\right)},
\end{equation} 
 the imaginary DOZZ formula~\eqref{eq:imdozz} can be written as 
\begin{equation}\label{eq:0s}
    C_{b}^{\rm ImDOZZ}(\wh\alpha_1,\wh\alpha_2, \wh\alpha_3)=\omega^{(b)}_{(0,2b P_1),(0,2b P_2),(0,2b P_3)}
\end{equation}
where $\wh\alpha_i=P_0(b)-P_i$ as in~\eqref{eq:Pidef}. Note that  $1-b^2$  from~\eqref{nor njr} can be written as $2b P_0$.
\begin{theorem}\label{thm:mag} 
For $\kappa\in(\frac{8}{3},8)$ and  $b=2/\sqrt{\kappa}$,
\begin{equation}\label{eq:Rformula}
    R^{\rm loop}(\kappa)=\frac{1}{\sqrt{2\cos(\pi(1-b^2))}}\omega^{(b)}_{(1,0),(1,0),(1,0)}.
\end{equation}\end{theorem}
The particular way~\eqref{eq:Rformula} of presenting the formula for $R^{\rm loop}(\kappa)$ was suggested by Nivesvivat, Jacobsen, and Rilbault in private communication.
They arrived at this formula independently from us using the CFT method from~\cite{Nivesvivat:2023kfp}, and checked it numerically using the  transfer matrix method.  
In fact, for each $\omega^{(b)}_{(r_1,s_1),(r_2,s_2),(r_3,s_3)}$ they can propose a geometric interpretation in terms of the O$(n)$ loop model. 
Each pair $(r,s)$ parameterizes an operator in the spectrum of that CFT, and $(0,s)$ corresponds to nesting loop statistics. Together with Theorem~\ref{thm: connectivity}, this gives a CFT explanation of the formulation~\eqref{eq:0s} for the imaginary DOZZ formula.
The case $(r,0)$ corresponds to $r$ loops passing through a given point, which explains the subscript $(1,0),(1,0),(1,0)$ in~\eqref{eq:Rformula}. As an intermediate step in proving Theorem~\ref{thm:mag}, we relate $\omega^{(b)}_{(1,0),(1,0),(0,s)}$ to its geometric counterpart; see  Remark~\ref{cor:mix3point}. We plan to investigate other cases in the future. See Section~\ref{sec: outlook}
for further discussion on the CFT from~\cite{Nivesvivat:2023kfp}.

\subsection{Proof method based on Liouville quantum gravity}\label{subsec:method}

Liouville quantum gravity (LQG) is a theory of random surfaces, where the geometry is governed by the exponential of the Gaussian free field (GFF) and its variants~\cite{shef-kpz,dddf-lfpp,gm-uniqueness}. An important source of  GFF variants is the Liouville CFT. On the one hand, the correlation functions of these fields are exactly solvable. In particular, the DOZZ formula gives the three-point function on the sphere~\cite{krv-dozz}.  On the other hand,  the scaling limits of natural discrete random surfaces (i.e.\ random planar maps) can be described by LQG surfaces whose fields are obtained from Liouville CFT~\cite{ahs-sphere, cercle-quantum-disk, AHS-SLE-integrability}. 
The independent coupling of CLE and LQG describes the scaling limit of random planar maps decorated by the critical FK random cluster or the O$(n)$ loop model. This insight originated from~\cite{shef-burger} and has been realized in various senses since then; see \cite{wedges} and~\cite[Section 5]{ghs-mating-survey} for references. Our general strategy for the proofs of Theorems~\ref{thm: connectivity}, \ref{thm:nesting} and \ref{thm:mag} is to express the product of a CLE three-point function on the sphere and the Liouville DOZZ formula as something solvable. We now explain how this is implemented in the three settings.

We start with Theorem~\ref{thm:mag} on the CLE loops since its proof is the easiest to explain.  
We rely on the  CLE$_\kappa$ loop intensity measure $\SLE_\kappa^\lp$ introduced by Kemppainen and Werner~\cite{werner-sphere-cle}, which describes the law of a loop chosen according to the counting measure on the set of all loops from a full-plane CLE$_\kappa$. This measure is infinite and conformally invariant. Modulo a multiplicative constant, it agrees with the SLE loop measure constructed by Zhan~\cite{zhan-loop-measures} when $\kappa\in (\frac83,8)$,  so we call a sample from $\SLE_\kappa^\lp$ an SLE loop. By definition, $G^\lp_n$ in Theorem~\ref{thm:mag} is the Green function for the Minkowski content of an SLE loop. 
Building on~\cite{shef-zipper}, it was demonstrated in~\cite{loopwelding,ACSW24a} that when the LQG parameter $\gamma\in (0,2)$ equals $\sqrt{\kappa}$ or $\sqrt{16/\kappa}$,  the SLE loop on top of an independent canonical LQG sphere has nice properties. Essentially, the two regions on the sphere separated by the loop are a pair of independent canonical LQG disks with the same LQG boundary length. Additional care is needed when $\kappa\in (4,8)$  since the loop is non-simple, but this is treated in~\cite{ACSW24a}.

One way to express the LQG boundary measure on an SLE curve is using the exponential of the field over the Minkowski content measure on the curve. By adding three points on the SLE loop and applying the Girsanov theorem, we can compute the product of $G^\lp_3$ and the Liouville DOZZ formula with appropriate parameters using the joint law of the loop length and the areas of the two regions separated by the loop.
The case of $G^\lp_2$ can be treated similarly by adding two points on the SLE loop, but we need to add a third point on one side of the loop to obtain finite quantities. 
This allows us to compute the ratio $R^{\rm loop}(\kappa)$. A crucial point in the computation is to keep track of several multiplicative constants in front of infinite measures, whose exact values were not necessary for earlier results.
These constants were determined in our companion paper~\cite{ACSW24a}.

We prove Theorems~\ref{thm: connectivity} for $R^{\rm cluster}$ in a similar manner but with additional difficulties. 
The Miller-Schoug measure and the LQG measure on the CLE cluster have the same kind of relation as the Minkowski content and LQG length for the SLE loop; see Section~\ref{sec:natural-measures}. 
This allows us to express the product of $G^{\rm cluster}_3$ and the Liouville DOZZ formula by the LQG mass of the CLE clusters together with certain lengths and areas observables as in the  $G^\lp_3$ case. At the technical level, we still need to use $\SLE_\kappa^\lp$ to express the Green function $G^{\rm cluster}_n$ in order to employ results from~\cite{loopwelding,ACSW24a}.  However, on the sphere, one CLE loop naturally corresponds to two CLE clusters that share the loop as their interface. We introduce an additional marked point on the sphere to uniquely specify a cluster.
This forces us to consider a four-pointed LQG sphere instead of three-pointed ones as in the $R^{\rm loop}$ case. 
Another challenge compared to the $R^{\rm loop}$ case lies in analyzing the LQG mass of the CLE clusters. By \cite{ccm-perimeter-cascade,msw-cle-lqg,msw-non-simple}, it can be expressed by observables of a stable L\'evy process. The final step in solving $R^{\rm cluster}$ is a fairly involved calculation concerning this L\'evy process.

For Theorem~\ref{thm:nesting} on the nested loop statistics, as mentioned below the theorem statement, the key is to solve the joint moments of the three conformal radii of $\eta_i$ viewed from $z_i$ ($i=1,2,3,$)  where $\eta_i$ is the outermost loop separating $z_i$ and the other two points. The three loops $\eta_1,\eta_2,\eta_3$ cut the LQG sphere into four surfaces. Three of them are LQG disks, each containing a point $z_i$. The last surface has the pair-of-pants topology, and we call it the quantum pair of pants. (As before, when $\kappa\in (4,8)$ the topology is more complicated.)
By a method developed by Holden and two of the authors~\cite{AHS-SLE-integrability}, after modifying the log singularity of the field around $z_1,z_2,z_3$ via the Girsanov theorem, the product of the conformal radius three-point function and the Liouville DOZZ formula can be expressed in terms of the LQG areas and boundary lengths of the quantum pair of pants and the three LQG disks with a modified log singularity near $z_i$. The LQG area and boundary length of the modified disk can be solved from the FZZ formula proved by Remy and two of the authors~\cite{ARS-FZZ}. For the quantum pair of pants, we again perform a fine analysis of the aforementioned L\'evy process.

\subsection{Outlook and perspective: towards conformal field theories for CLE}\label{sec: outlook}
There are several CFTs in the literature \cite{GKR23, Nivesvivat:2023kfp, baverez2024cftsleloopmeasures} whose correlation functions are supposed to be observables for CLE or the SLE loop. 
An important program is to rigorously establish these CFT-CLE connections, and exactly solve the correlation functions/observables. 
This is significant from both statistical physics and field theory perspectives. 
The most fundamental data to solve for these CFTs are the structure constants, namely three-point functions, since 
they prescribe the operator product expansion. These structure constants should be closely related to imaginary DOZZ formula. 
This is indeed the case for all solved instances. We expect that results and methods in our paper will play a crucial role in this program, for two reasons. 
First, the  CLE observables considered in Theorems~\ref{thm: connectivity}, \ref{thm:nesting} and \ref{thm:mag}, 
and similar ones are supposed to correspond to structure constants of these CFTs. Second, as explained in Section~\ref{subsec:method}, our method is especially suited for establishing relations between CLE observables and the  imaginary DOZZ formula. We now give an overview of these CFTs and the main mathematical challenges in this program.

\begin{itemize}
    \item A successful approach in physics to study 2D random loop models is to map them to height function models. It is assumed that at criticality the height functions converge to variants of the Gaussian free field, based on which certain loop observables can be expressed in terms of integrals over the free field. This is the Coulomb gas method; see e.g.\cite{diFrancesco:1987qf, Cardy_2006,Nivesvivat:2023kfp}. Recently  Guillarmou, Kupiainen and  Rhodes rigorously constructed a CFT called the compactified imaginary Liouville theory (CILT)~\cite{GKR23}  which is believed to capture the free field integrals in the Coulomb gas method.  
    The CLE observables from Theorems~\ref{thm:nesting} and \ref{thm:mag} are supposed to be related to the electronic and magnetic operators in the CILT, respectively. The three-point functions for CILT are not hard to compute and are indeed imaginary DOZZ type formulae. It is a major challenge to rigorously relate the CILT to CLE, putting the Coulomb gas method for loop model on the firm ground.  First, the convergence assumption is hard to prove. Second, it is not clear how to make sense of the height function mapping directly using CLE, since an intermediate object in the mapping is a complex-valued loop model that is hard to define in the continuum.

    \item     Another CFT is the one discussed  briefly below Theorem~\ref{thm:mag}. Instead of proposing a path integral, one can assume that a large class of observables in random cluster models or the O$(n)$ loop models form a CFT. Such a CFT shares some key features with the CILT: most importantly, it is non-diagnonal and logarithmic.   However, it admits correlation functions such as $n$ points lying in the same cluster or the same loop, which are beyond the Coulomb gas/CILT framework because the so-called neutrality condition in Coulomb gas is not satisfied. 
    Recently there have been several important work on the conformal bootstrap aspect of this CFT \cite{Picco:2019dkm,He:2020rfk,Nivesvivat-Ribault,Nivesvivat:2023kfp}; see~\cite[Section 5]{Ribault2024Notes} for an overview. In particular, 
    the authors from~\cite{Nivesvivat:2023kfp} made precise predictions on the conformal bootstrap expansion for several natural four-point functions, where each term in the expansion is a product of three-point functions and known functions called the conformal blocks. Despite this remarkable success, the operator product expansion in this CFT is not fully understood yet. The four-point function expansion is only determined in a term-by-term fashion via an involved numerical scheme, and  there is no predictions on $n$-point functions for $n\ge 5$. It would be extremely interesting to completely solve this CFT and prove their implications on CLE.

    \item Although the interaction between SLE and CFT has been an intensive topic 
    in probability, the study is mainly focused on boundary CFT; see the overview~\cite{peltola2019toward} and the very recent work \cite{feng2024multiplesleskappain08}.  Recently, Baverez and Jego~\cite{baverez2024cftsleloopmeasures}, and Gordina, Qian and Wang \cite{gordina2024} constructed the Virasoro representation for the $\SLE_\kappa$ loop measure for $\kappa\in (0,4]$ based on conformal restriction property. Furthermore, Baverez and Jego~\cite{baverez2024cftsleloopmeasures} used it to prove the uniqueness conjecture for the SLE loop from \cite{kontsevich2006}. These work and our work together reveal two important aspects of the bulk CFT for SLE and CLE: Virasoro representation and integrability. We expect that the fusion of these two approaches will lead to further progress. For example, Baverez and Jego~\cite{baverez2024cftsleloopmeasures} proposed an interesting CFT for the $\SLE_\kappa$ loop measure,  we plan to use the proof strategy for Theorem \ref{thm:nesting} to solve these three-point functions. A main technical challenge is to solve the quantum pair of pants in this setting, where the tool of Levy process  as in our proof of Theorem~\ref{thm:nesting} is not available.

 \item Besides the three CFTs above, there are other natural three-point functions that are of interest in 2D statistical physics.  
 One example is the scaling limit of the probability that three points lie in the same spin cluster of the $3$-Potts model. 
 The analog of $R^{\rm cluster}$ was studied numerically in~\cite{Delfino:2013pma} while the exact value was left open. In a forthcoming work by Liu, Zhuang, and the second and fourth authors,
 its exact value will be derived by applying the proof method for Theorem~\ref{thm: connectivity} to a variant of CLE.  
 Recently, Camia and Feng~\cite{camia2024} established the existence of the scaling limit of various geometric correlation functions for percolation. 
 Some of them appeared in the CFT in~\cite{Nivesvivat:2023kfp}, while others are new. For example, they considered an observable  whose scaling dimension 
is the backbone exponent recently derived in~\cite{nolin2024backboneexponenttwodimensionalpercolation}. It would be interesting to derive the three-point function for this observable and identify a CFT that contains it as a three-point function.  This would provide a CFT  interpretation of the intriguing transcendental value of the backbone exponent.    
\end{itemize}

{\bf Organization of the paper.}  In Section~\ref{sec:prelim}, we provide preliminaries on LQG and Liouville CFT. In Section~\ref{sec-loop-green-simple}, we use the conformal welding for the SLE loop to prove Theorem~\ref{thm:mag} for $\kappa\in (\frac{8}{3},4)$. We treat the $\kappa\in (4,8)$ case in Section~\ref{sec-loop-green-nonsimple}. In Section~\ref{sec:carpet green}, we prove  Theorem~\ref{thm: connectivity} modulo calculations on the stable L\'evy process,  which are carried out  in Section~\ref{levy}. 
In Section~\ref{sec-welding}, we introduce the quantum pair of pants and compute the joint law of its area and boundary lengths, based on which we prove  Theorem~\ref{thm:nesting} in Section~\ref{sec:3-pt}. For the appendix, we prove Theorem~\ref{thm:dv} in Appendix~\ref{percolation}, treat the $\kappa=4$ case of Theorems~\ref{thm: connectivity} and~\ref{thm:mag}  in Appendix~\ref{kappato4}, and gather useful integration formulas in Appendix~\ref{appendix:integration}.

\medskip

{\bf Acknowledgement}
We thank Guillaume Baverez, Federico Camia, Yu Feng, Jesper Lykke Jacobsen, Rongvoram Nivesvivat and Sylvain Ribault for explaining their works. We thank Nina Holden, Matthis Lehmkuehler, and Shengjing Xu for helpful discussions. 
M.A. was supported by the Simons Foundation as a Junior Fellow at the Simons Society of Fellows and partially supported by NSF grant DMS-1712862. G.C., X.S. and B.W.\ were supported by National Key R\&D Program of China (No.\ 2023YFA1010700). X.S. was also partially supported by the  NSF grant DMS-2027986, the NSF Career grant DMS-2046514, and a start-up grant from the University of Pennsylvania. 

\section{Preliminaries}\label{sec:prelim}
In this section we review the precise definitions of some $\gamma$-LQG surfaces and Liouville fields mentioned in~Section~\ref{subsec:method}.
For more background, we refer to~\cite{ghs-mating-survey,vargas-dozz-notes} and references  therein, as well as the preliminary sections in \cite{AHS-SLE-integrability, ARS-FZZ}.
We will often use probabilistic language in the setting of non-probability measures.
Let $M$ be a $\sigma$-finite measure on a measurable space $(\Omega, \cF)$.
Assume that the $\cF$-measurable function $X:(\Omega,\cF)\rta (E,\cE)$ takes values in $(E,\cE)$.
We call the pushforward measure $M_X = X_*M$ on $(E,\sigma(X))$ the \emph{law} of $X$, and say that $X$ is \emph{sampled} from $M_X$. We also write  $M_X[f] := \int f(x) M_X(dx)$.
The total mass of a finite measure $M$ is expressed as $|M|$, and the probability measure proportional to $M$ is denoted by $M^{\#}=|M|^{-1}M$.

\subsection{Liouville field and the DOZZ formula}\label{subsubsec:GFF}

Let $\cX$ denote either the complex plane $\C$ or the upper half plane $\bbH$. Assume $\cX$ is equipped with a smooth metric $g$ such that the metric completion of $(\cX, g)$ forms a compact Riemannian manifold. For simplicity, we will not distinguish $\cX$ from its compactification.
Define $H^1(\cX)$ as the Sobolev space where the norm is defined by $\Vert f\Vert^2=\int_\cX (|f|^2+|\nabla f|^2) d{\rm Vol}_g$ for $f\in H^1(\cX)$. The dual space of $H^1(\cX)$ is denoted by $H^{-1}(\cX)$.

We now recall the \emph{Gaussian free field} (GFF) on $\C$ or $\bbH$. Let $|z|_+:=\max\{|z|,1\}$, and define
\begin{align}\label{eq:covariance}
G_\bbH(z,w) &= -\log |z-w| - \log|z-\ol w| + 2 \log|z|_+ + 2\log |w|_+  \quad &&z,w\in \bbH,\nonumber\\
G_\C(z,w) &= -\log|z-w| + \log|z|_+ + \log|w|_+  \quad &&z,w\in \C.\nonumber
\end{align}
For $\cX\in\{\C,\bbH\}$, the Gaussian free field $h_\cX$ is a random element in $H^{-1}(\cX)$ having mean zero on $\{z\in \cX: |z|=1 \}$,  such that $(h_\cX,f)$ is a Gaussian with mean zero and variance $\int f(z) G_\cX(z,w) f(w)\,d^2z\,d^2w$ for each  $f\in H^1(\cX)$. We call $h_\C$ the \emph{whole plane GFF} and $h_\bbH$ to be the \emph{free-boundary GFF} on $\bbH$, 
and let $P_\cX$ denote the law of $h_\cX$.

The Liouville fields on $ \C$  and $\bbH$ can be defined using the Gaussian free field. We follow the treatment of \cite[Section 2.2]{AHS-SLE-integrability}.
\begin{definition}
	For $(h, \mathbf c)$ sampled from $P_\C \times [e^{-2Qc}dc]$, we call 
 $\phi =  h(z) -2Q \log |z|_+ +\mathbf c$ the \emph{Liouville field on $\C$}, and denote its law by $\LF_\C$. Similarly, for $(h, \mathbf c)$ sampled from $P_\bbH \times [e^{-Qc}dc]$, we call $\phi =  h(z) -2Q \log |z|_+ +\mathbf c$ the \emph{Liouville field on $\bbH$}, and denote its law by $\LF_\bbH$.
\end{definition}
We also define Liouville fields with   insertions.  
\begin{definition}\label{def-RV-sph}
        Let $m \geq 1$, let  $(\alpha_i, z_i) \in \R \times \C$ for $i = 1, \dots, m$, and set  
        \[C _{  \C}^{(\alpha_i,z_i)_i}=\prod_{i=1}^m |z_i|_+^{-\alpha_i(2Q -\alpha_i)} e^{\sum_{i < j} \alpha_i \alpha_j G_\C(z_i, z_j)}.\] For $(h, \mathbf c)$ sampled from $ C_\C^{(\alpha_i,z_i)_i}  P_\C \times [e^{(\sum_i \alpha_i  - 2Q)c}dc]$, the field 
        \[\phi(z) = h(z) -2Q \log |z|_+  + \sum_{i=1}^m \alpha_i G_\C(z, z_i) + \mathbf c\] is the  \emph{Liouville field on $ \C$ with insertions $(\alpha_i,z_i)_{1\le i\le m}$}. We denote its law by $\LF_{ \C}^{(\alpha_i,z_i)_i}$. 
 
	Similarly, for $(\alpha,z_0) \in \R\times \bbH$ and $(h, \mathbf c)$ sampled from $(2\Im z_0)^{-\alpha^2/2} |z_0|_+^{-2\alpha (Q-\alpha)}P_\bbH\times [e^{(\alpha-Q)c}dc]$, the field  $\phi(z) = h(z) - 2Q \log |z|_+ + \alpha G_\bbH(z, z_0) + \mathbf c$ is the \emph{Liouville field on $\bbH$ with insertion $(\alpha, z_0)$}. We denote its law by $\LF_\bbH^{(\alpha, z_0)}$. 
\end{definition}
The Liouville field with insertions is formally related to the Liouville field via $\LF_{ \C}^{(\alpha_i,z_i)_{i}} = \prod_{i=1}^{m} e^{\alpha_i \phi(z_i)}\LF_{  \C}(d\phi)$ and $\LF_\bbH^{(\alpha, z_0)}=e^{\alpha \phi(z_0)} \LF_\bbH (d\phi)$.
These can be made rigorous via regularization and renormalization, e.g.  $\LF_{ \C}^{(\alpha_i,z_i)_{i}} =\lim_{\varepsilon\to0}  \prod_{i=1}^{m} \varepsilon^{\frac{\alpha_i^2}{2}} e^{\alpha_i \phi_\varepsilon(z_i)}\LF_{  \C}(d\phi)$, where $\phi_\varepsilon$ is the average of $\phi$ on $\{w:|w-z|=\varepsilon\}$. See~\cite[Lemma 2.8]{AHS-SLE-integrability} and ~\cite[Lemma 2.2]{ARS-FZZ}.

Now we recall the quantum area and quantum length measures. Fix $\gamma\in(0,2)$ and $\cX\in \{\C, \bbH\}$, and sample  $h$ from $P_\cX$. For $\eps > 0$ and $z \in  \cX\cup \bdy \cX$, let $h_\eps(z)$ denote the average of $h$ on $\{w \in \cX \: : \: |w-z| = \eps\}$. Then the \emph{quantum area measure} $\mu_h$ on $\cX$ is the weak limit $\lim_{\eps\to0} \eps^{\gamma^2/2} e^{\gamma h_\eps(z)}d^2z$ on $\cX$; see \cite{shef-kpz,shef-wang-lqg-coord}. On $\bbH$, we similarly define the \emph{quantum boundary length measure} $\nu_h$ on $\partial \bbH$ to be $\lim_{\eps \to 0} \eps^{\gamma^2/4}e^{\frac\gamma2 h_\eps(x)} dx$.
The definitions of quantum area and boundary length can be extended straightforwardly to variants of the GFF, for instance  Liouville fields with insertions.

We now review the the 3-point structure constant of LCFT given by the DOZZ formula. We choose $(u_1, u_2, u_3) = (0,1,e^{i\pi/3})$ such that $|u_1-u_2|=|u_2-u_3|=|u_3-u_1|=1$ for convenience. 
The main result in \cite{krv-dozz} can deduce the following density for the quantum area law of $\LF_\C^{(\alpha_j, u_j)_3}$.
\begin{proposition}\label{lem-sph-area-law}\label{prop-DOZZ}
	Suppose $\alpha_1, \alpha_2, \alpha_3$ satisfy the extended Seiberg bounds
 \eqb\label{eq-ext-Seiberg}
\frac1\gamma(2Q - \sum_{j=1}^3 \alpha_j ) < \frac4{\gamma^2} \wedge \min_j \frac2\gamma(Q-\alpha_j).
\eqe
 Then the quantum area law of $\LF_\C^{(\alpha_j, u_j)_3}$ is 
	$\frac{C_\gamma^\mathrm{DOZZ}(\alpha_1,\alpha_2, \alpha_3)}{2  \Gamma(\frac1\gamma(\sum_j \alpha_j - 2Q))} a^{\frac1\gamma(\sum_j \alpha_j - 2Q) - 1} da,$
 where $C^\mathrm{DOZZ}_\gamma(\alpha_1,\alpha_2,\alpha_3)$ is given by~\eqref{eq:DOZZ}. Moreover, under the stronger condition that $\alpha_1, \alpha_2, \alpha_3$ satisfy the \emph{Seiberg bounds}
	\eqb\label{eq-seiberg}
	\sum_{i=1}^3 \alpha_i > 2Q, \qquad \textrm{and}\qquad  \alpha_i < Q\text{ for }i=1,2,3,
	\eqe
 for $\mu>0$, $\LF_\C^{(\alpha_i, z_i)_i}[ e^{-\mu \mu_\phi(\C)}]$ is finite and equals $\frac12 \mu^{-\frac1\gamma(\sum_j \alpha_j - 2Q)} C^\mathrm{DOZZ}_\gamma(\alpha_1,\alpha_2,\alpha_3)$.
\end{proposition}
\begin{proof}
 Note that  our choice of $(u_1,u_2,u_3)$ satisfies $C^{(\alpha_i, u_i)_i}_\C = 1$. Thus, if we sample  $(h, \mathbf c) \sim P_\C\times [e^{(\sum_j \alpha_j - 2Q)c}\, dc]$ and let $\phi_0(z) = h(z) - 2Q \log |z|_+ + \sum_{i=1}^3 \alpha_i G_\C(z, u_i)$, then $\phi_0 + \mathbf c$ is the Liouville field with insertions. By the change of variables $y = e^{\gamma c} \mu_{\phi_0}(\C)$, for $I \subset (0,\infty)$ we have 
	\[\int_\R \E[ 1_{e^{\gamma c} \mu_{\phi_0}(\C) \in I}] e^{(\sum_j \alpha_j - 2Q)c}\, dc = \E\left[\int_I \left(\frac y {\mu_{\phi_0}(\C)}\right)^{\frac1\gamma(\sum_j\alpha_j - 2Q)} \frac1{\gamma y} \, dy\right].\]
	Therefore the quantum area law of $\LF_\C^{(\alpha_j, u_j)_3}$ is $\frac1\gamma \E[\mu_{\phi_0}(\C)^{\frac1\gamma(2Q - \sum_{j=1}^3 \alpha_j)}] a^{\frac1\gamma(\sum_j \alpha_j - 2Q) - 1} da$. The first claim is then immediate from $\E[\mu_{\phi_0}(\C)^{\frac1\gamma(2Q - \sum_{j=1}^3 \alpha_j)}] = \frac{\gamma C_\gamma^\mathrm{DOZZ}(\alpha_1,\alpha_2, \alpha_3)}{2 \Gamma(\frac1\gamma(\sum_j \alpha_j - 2Q))}$ \cite[Theorem 2.4]{krv-dozz}, and the second claim follows from the first using $\int_0^\infty e^{-\mu x}x^{\beta-1}dx=\mu^{-\beta}\Gamma(\beta)$ for $\beta>0$.
\end{proof}

\subsection{Adding insertions to the Liouville field}
We collect two facts on Liouville fields that follow from  the Girsanov theorem. They allow us to understand the resulting field after adding marked points according a certain quantum measure.

\begin{lemma}[{\cite[Lemma 2.31]{AHS-SLE-integrability}}]\label{lem-add-gamma}
 We have $\mu_\phi(d u) \LF_\C^{(\alpha_i, z_i)_i}(d\phi) = \LF_\C^{(\gamma, u), (\alpha_i, z_i)_i}(d\phi) \, d^2u$.
\end{lemma}
Let $d \in [0,2]$.  
Suppose $\sigma(dz)$ is a Radon measure on $\C$ satisfying 
\[ \int_{\C^2} \frac{\sigma(dx) \sigma(dy)}{|x-y|^{d - \eps}} < \infty \quad \text{ for all } \eps >0.\] 
Then for $\alpha < \sqrt{2d}$ and $h$ sampled from $P_\C$, the \emph{Gaussian multiplicative chaos (GMC)} measure 
\eqb\label{eq-gmc}
\mathrm{GMC}_{\sigma, \alpha, h}(du):= \lim_{\eps \to 0} \eps^{\alpha^2/2} e^{\alpha h_\eps(u)}\sigma(du)
\eqe
exists; see e.g.\ \cite{berestycki-gmt-elementary} for details. 
\begin{lemma}\label{lem-add-insertion}
We have $\mathrm{GMC}_{\sigma, \alpha, \phi}(du) \LF_\C^{(\alpha_i, z_i)_i}(d\phi) = \LF_\C^{(\alpha, u), (\alpha_i, z_i)_i}(d\phi) \sigma(du)$. Similarly, for $n \geq 1$, 
   \[(\prod_{j=1}^n \mathrm{GMC}_{\sigma, \alpha, \phi}(du_j))\LF_\C^{(\alpha_i, z_i)_i}(d\phi) = \LF_\C^{(\alpha, u_j)_j, (\alpha_i, z_i)_i}(d\phi) \prod_{j=1}^n \sigma(du_j). \]
\end{lemma}
The proof of Lemma~\ref{lem-add-insertion} is essentially identical to that of \cite[Lemma 2.31]{AHS-SLE-integrability} hence we omit the detail.
We will use it in the proof of Theorem~\ref{thm:mag}, where the GMC measure is  the quantum length of an SLE curve. See Lemmas~\ref{lem-benoist} and~\ref{lem-benoist-ns}. We also use a variant of Lemma~\ref{lem-add-insertion} for  the quantum natural measure of a CLE carpet/gasket; see Lemma~\ref{resampling}. 

\subsection{Quantum sphere}\label{subsub:quantum-surface}
We first recall the definition of quantum surfaces in LQG.
Let $\gamma\in(0,2)$ and $Q=\frac{\gamma}{2}+\frac{2}{\gamma}$, and consider pairs $(D,h)$ where $D\subset\C$ is a planar domain and $h$ is a distribution on $D$. For two pairs $(D,h)$ and $(\wt D,\wt h)$, we say $(D, h) \sim_\gamma (\wt D, \wt h)$ if there exists a conformal map $\psi$ from  $\wt D$ to $D$ such that 
\eqb\label{eq-QS}
\wt h = h \circ \psi + Q \log |\psi'|. 
\eqe
This gives an equivalence relation on the space of pairs $(D,h)$. A \emph{quantum surface} is an equivalence class $(D,h)/{\sim_\gamma}$, and we call a choice of representative an \emph{embedding} of the quantum surface. The quantum area and boundary length measures of a quantum surface do not depend on the choice of embedding. That is, if $h$ is a GFF on $D$, then~\eqref{eq-QS} implies $\psi_* \mu_{\wt h} = \mu_h$ and $\psi_* \nu_{\wt h} = \nu_h$  \cite{shef-kpz,shef-wang-lqg-coord}.

We now define quantum surfaces decorated by  curves and marked points. 
Consider tuples $(D,h,(\eta_i)_{i\in \cI},(z_j)_{j \in \cJ})$ where $D\subset\C$ is a domain, $h$ is a distribution on $D$, $(\eta_i)_{i \in \cI}$ is a collection of curves on $\ol D$, and $\{z_j\}_{j \in \cJ}$ is a collection of points in $\ol D$. We similarly say
\[(D,h,(\eta_i)_{i\in \cI},(z_j)_{j \in \cJ})\sim_\gamma(\wt D,\wt h,(\wt \eta_i)_{i\in \cI},(\wt z_j)_{j \in \cJ})) \]
if there is a conformal map $\psi:\wt D\to D$ satisfying~\eqref{eq-QS} and $\psi\circ\wt \eta_i=\eta_i$ for all $i\in\cI$ and $\psi(\wt z_j)=z_j$ for all $j \in \cJ$. An equivalence class is called a \emph{decorated quantum surface}, and a representative is called an embedding of the decorated quantum surface.

The quantum sphere is the most canonical quantum surface with the sphere topology. Let $\QS_n$ be the law of the quantum sphere with $n$ marked points, and $\QS = \QS_0$. We will need $\QS, \QS_2$ and $\QS_3$.
Originally, $\QS_2$ was  defined first as in \cite{wedges}, and $\QS$ and $\QS_2$ were defined by adding or removing marked points from $\QS_2$. For our purpose it is convenient to use the following definition. The equivalence with the original definition is explained in \cite[Remark 2.30]{AHS-SLE-integrability}.
\begin{definition}\label{def-QS-2}
Let $(u_1, u_2, u_3) = (0, 1, e^{i\pi/3})$,  and sample $\phi$ from $\frac{\pi \gamma}{2(Q-\gamma)^2}\LF_\C^{(\gamma, u_1),(\gamma, u_2),(\gamma, u_3)}$. We call the decorated quantum surface  $(\C, \phi,u_1,u_2,u_3)/{\sim_\gamma}$ the \emph{quantum sphere with three marked points}, and denote its law by $\QS_3$.
\end{definition}
\begin{definition}\label{def-QS}
For $(\C,h,u_1,u_2,u_3)/{\sim_\gamma}$ sampled from $\mu_h(\C)^{-1} \QS_3$,
let $\QS_2$ be the law of the decorated quantum surface $(\C, h,u_1,u_2)/{\sim_\gamma}$. 	
For $(\C,h,u_1,u_2,u_3)/{\sim_\gamma}$ sampled from $\mu_h(\C)^{-3}\QS_3$, let $\QS$ be the law of the quantum surface $(\C,h)/{\sim_\gamma}$.
\end{definition}

We now recall the \emph{uniform embedding} of $\QS$.
Let $\conf(\hat\C)$ be the group of conformal automorphisms of the Riemann sphere $\hat\C = \C \cup\{\infty\}$. 
Let $\mathbf m_{\wh \C}$ be the left and right invariant Haar measure on $\conf(\hat\C)$, which is unique up to multiplicative constant and can be described as follows.
\begin{lemma}[{\cite[Lemma 2.28]{AHS-SLE-integrability}}]\label{Haar-density}
Suppose $\frak f$ is a conformal automorphism on $\C$ sampled from $\mathbf m_{\wh \C}$. Then for any fixed three points $z_1,z_2,z_3\in\hat\C$, the law of $(\frak f(z_1),\frak f(z_2),\frak f(z_3))$ equals $C|(p-q)(q-r)(r-p)|^{-2}\,d^2p\,d^2q\,d^2r$ for some constant $C\in(0,\infty)$. 
\end{lemma}
By the definition of quantum surfaces, $\QS$ can be viewed as an infinite measure on $H^{-1}(\C)/\conf(\hat\C)$. 
For a quantum surface $S\in H^{-1}(\C)/\conf(\hat\C)$, let $h \in H^{-1}(\C)$ be a representative of $S$, and let $\mathbf m_{\wh \C}\ltimes S$ denote the pushforward  of $\mathbf m_{\wh \C}$ under the map sending $\psi \in \conf(\hat\C)$ to $h\circ\psi+Q\log|\psi'| \in H^{-1}(\C)$. Note that $\mathbf m_{\wh \C}\ltimes S$ does not depend on the choice of $h$ since $\mathbf m_{\wh \C}$ is invariant. Finally, set  $\mathbf m_{\wh \C}\ltimes\QS:=\int_{H^{-1}(\C)/\conf(\hat\C)}  [\mathbf m_{\wh \C}\ltimes S] \QS(dS)$.

\begin{theorem}[{\cite[Theorem 1.2, Proposition 2.32]{AHS-SLE-integrability}}]\label{thm:embed-qs}
Let $\mathbf m_{\wh \C}$ be such that the constant $C$ in Lemma \ref{Haar-density} equals $1$. Then $\mathbf m_{\wh \C}\ltimes\QS=\frac{\pi\gamma}{2(Q-\gamma)^2}\LF_{\C}$.
\end{theorem}

\subsection{Quantum disks}
\label{subsec-QD}

The quantum disk is the most canonical quantum surface with the disk topology. Let $\QD_{m,n}$ be the law of the quantum disk with $m$ marked bulk points and $n$ marked boundary points. We will need $\QD$, $\QD_{1,0}$, $\QD_{1,1}$ $\QD_{0,2}$, and $\QD_{0,3}$.
Historically, $\QD_{0,2}$ was defined first in \cite{wedges}, and other $\QD_{m,n}$ were defined in terms of $\QD_{0,2}$ by adding and removing points. For our purpose it is convenient use
the equivalent definition of $\QD_{1,0}$ using the Liouville field \cite[Theorem 3.4]{ARS-FZZ}, then define $\QD$, $\QD_{1,1}$ and $\QD_{0,3}$ in terms of $\QD_{1,0}$.

Let $\{\LF_\bbH^{(\alpha, i)}(\ell): \ell >0 \}$ be the disintegration of $\LF_\bbH^{(\alpha, i)}$  over the quantum boundary length. That is, each measure $\LF_\bbH^{(\alpha, i)}(\ell)$ is supported on the set of fields having quantum boundary length $\ell$, and $\LF_\bbH^{(\alpha, i)} = \int_0^\infty \LF_\bbH^{(\alpha, i)}(\ell)\, d\ell$. 
See \cite[Lemma 4.3]{ARS-FZZ} for an explicit construction of $\LF_\bbH^{(\alpha, i)}(\ell)$.

\begin{definition}\label{def-QD-alpha}
For $\ell>0$,
we let $\cM_{1,0}^\disk(\alpha;\ell )$ be the law of the quantum surface $(\bbH, \phi, i)/{\sim_\gamma}$ 
with $\phi$ sampled from $\LF_\bbH^{(\alpha, i)} (\ell)$. Define $\QD_{1,0}(\ell) := \frac\gamma{2\pi (Q-\gamma)^2}\cM_{1,0}^\disk(\gamma;\ell )$.
\end{definition}
\begin{definition}\label{def-QD}\label{def-QD0203}
Let $\ell>0$. Sample $(D,h,z)/{\sim_\gamma}$ from the weighted measure  $\mu_h(D)^{-1}\QD_{1,0}(\ell)$, and let $\QD$ be the law of $(D,h)/{\sim_\gamma}$. For $(D,h,z)/{\sim_\gamma}$ sampled from $\ell\QD_{1,0}(\ell)$, sample $p\in \bdy D$ from the probability measure proportional to quantum boundary length measure. Let $\QD_{1,1}(\ell)$ be the law of $(D,h,z, p)/{\sim_\gamma}$. Let $n\in\N$. For $(D,h)/{\sim_\gamma}$ sampled from $\ell^n\QD(\ell)$, independently sample $(p_i)_{1\le i\le n}\in\partial D$ from the probability measure proportional to quantum boundary measure. Let $\QD_{0,n}(\ell)$ be the law of $(D,h,(p_i)_{1\le i\le n})/\sim_\gamma$.
\end{definition}

One can fix the embedding of $\cM^{\disk}_{1,0}(\alpha;\ell)$ as follows:
\begin{proposition}[{\cite[Proposition 2.21]{ACSW24a}}]\label{lem:har}
	For $\alpha > \frac\gamma2$ and $\ell > 0$, let $(D,h,z)$ be an embedding of a sample from $\cM^{\disk}_{1,0}(\alpha;\ell)$. Given $(D,h,z)$, let $p$ be a point sampled from the harmonic measure on $\bdy D$ viewed from $z$, then the law of  $(D,h,z,p)/{\sim_\gamma}$ equals that of $(\bbH, X, i,0)/{\sim_\gamma}$ where $X$ is sampled from $\LF_{\bbH}^{(\alpha,i)}(\ell)$.
\end{proposition}
Now we record results on the boundary length and area distribution of quantum disks. 

\begin{proposition}[{\cite{remy-fb-formula}}]\label{prop-remy-U}\label{lem-M1-bdy-law}
For $\alpha > \frac{\gamma}{2}$,  the total mass of $\cM_{1,0}^\disk(\alpha;\ell )$ is 
	\eqb\label{eq:U0-explicit}
	1_{\ell>0} \frac2\gamma 2^{-\frac{\alpha^2}2} \ol U(\alpha)\ell^{\frac2\gamma(\alpha-Q)-1} \, d\ell, \text{ where }\ol U(\alpha) = \left( \frac{2^{-\frac{\gamma\alpha}2} 2\pi}{\Gamma(1-\frac{\gamma^2}4)} \right)^{\frac2\gamma(Q-\alpha)} 
	\Gamma( \frac{\gamma\alpha}2-\frac{\gamma^2}4).
	\eqe
\end{proposition}

Taking $\alpha = \gamma$ in Proposition~\ref{lem-M1-bdy-law}  gives the boundary length distribution for $\QD_{1,0}$.
\begin{lemma} [{\cite[Lemma 3.2]{ARS-FZZ}}] \label{rem-QD-length}\label{lem-QD-bdy-law}
The law of the quantum boundary length of $\QD$ is
$$R_\gamma 1_{\ell>0} \ell^{-\frac4{\gamma^2}-2} \, d\ell,\text{ where } R_\gamma=\frac{(2 \pi)^{\frac{4}{\gamma^2}-1}}{\left(1-\frac{\gamma^2}{4}\right) \Gamma\left(1-\frac{\gamma^2}{4}\right)^{\frac{4}{\gamma^2}}}.$$
\end{lemma}

Let $ K_\nu(x)$ be the modified Bessel function of the second kind, namely
\begin{align}\label{eq-Kv}
    K_\nu(x) := \int_0^\infty e^{-x \cosh t} \cosh(\nu t) \, dt \text{ for } x > 0 \text{ and } \nu \in \R.
\end{align}Define 
\[
\ol K_\nu(x)=\frac{2^{1-\nu}}{\Gamma(\nu)}x^\nu K_\nu(x)\quad \textrm{for }x>0 \quad\textrm{and}\quad \ol K_\nu(0)=\lim_{x\to 0^+}\ol K_\nu(x)=1.
\]

\begin{theorem}[{\cite[Theorem 1.2,  Proposition 4.19]{ARS-FZZ}}]\label{thm-FZZ}\label{prop-fzz}
	For $\alpha \in (\frac\gamma2, Q)$ and $\ell>0$, the law of quantum area $A$ of a sample from $\cM_{1,0}^\disk(\alpha; 1)^\#$ is the inverse gamma distribution with shape $\frac{2}{\gamma}(Q-\alpha)$ and scale $\frac{\ell^2}{4\sin\frac{\pi\gamma^2}{4}}$. 
Furthermore, with  $\ol U(\alpha)$ in \eqref{eq:U0-explicit}, for $\mu>0$ we have
	\[\cM_{1,0}^\disk(\alpha; \ell)[e^{-\mu A}] = \frac2\gamma 2^{-\frac{\alpha^2}2} \ol U(\alpha) \ell^{-1} \frac2{\Gamma(\frac2\gamma(Q-\alpha))} \left(\frac12 \sqrt{\frac{\mu}{\sin(\pi\gamma^2/4)}}  \right)^{\frac2\gamma(Q-\alpha)}K_{\frac2\gamma(Q-\alpha)} \left(\ell\sqrt{\frac{\mu}{\sin(\pi\gamma^2/4)}}  \right). \]
In particular, \[\cM_{1,0}^\disk(\alpha; \ell)^{\#}[e^{-\mu A}] = \ol K_{\frac2\gamma(Q-\alpha)} \left(\ell\sqrt{\frac{\mu}{\sin(\pi\gamma^2/4)}}  \right), \]
 \[\cM_{1,0}^\disk(\alpha; \ell)^{\#}[Ae^{-\mu A}] = \frac2{\mu\Gamma(\frac2\gamma(Q-\alpha))} \left(\frac\ell2 \sqrt{\frac{\mu}{\sin(\pi\gamma^2/4)}}  \right)^{\frac2\gamma(Q-\alpha)+1}K_{\frac2\gamma(Q-\alpha)-1} \left(\ell\sqrt{\frac{\mu}{\sin(\pi\gamma^2/4)}}  \right). \]
\end{theorem}

Putting $\alpha = \gamma$ and using Definitions~\ref{def-QD-alpha} and~\ref{def-QD} yields the area law of $\QD(\ell)^\#$: 
\begin{proposition}\label{prop-QD-area}
	The law of the quantum area of a sample from $\QD(\ell)^\#$ is the inverse gamma distribution with shape $\frac4{\gamma^2}$ and scale $\frac{\ell^2}{4 \sin \frac{\pi\gamma^2}4}$.
 In particular, we have $\mathrm{QD}^{\#}(\ell)\left[e^{-\mu A}\right]=\ol K_{\frac{4}{\gamma^2}}\left(\ell \sqrt{\frac{\mu}{\sin \left(\pi \gamma^2 / 4\right)}}\right)$.
\end{proposition}
\noindent Proposition~\ref{prop-QD-area} also follows from \cite[Theorem 1.2]{ag-disk} and \cite[Theorem 1.3]{ARS-FZZ}.

\section{SLE loop Green function: the simple case}\label{sec-loop-green-simple}

Recall the two-point and three-point Green functions  $G_2^\lp(z_1, z_2)$ and $G_3^\lp(z_1, z_2, z_3)$ of the SLE loop defined in Section~\ref{section:loop-connectivity}.  In this section, we calculate $R^{\rm loop}(\kappa)=\frac{G_3^\lp(z_1,z_2,z_3)}{\sqrt{G_2^\lp(z_1, z_2) G_2^\lp (z_2,z_3) G_2^\lp(z_3,z_1)}}$ for the simple case $\kappa\in (\frac{8}{3},4)$. 
The non-simple case $\kappa\in (4,8)$ will be treated in the next section. 

In Section~\ref{subsec:weld-disk}, we recall that the SLE loop measure can be obtained from the conformal welding of quantum disks.
In Section~\ref{section:G3} we deduce a conformal welding identity where three points lie on the SLE loop, from which $G_3^\lp$ naturally arises. Evaluating an LQG observable for this identity then gives $G_3^\lp$ in terms of the DOZZ formula from Liouville CFT. The case for $G^\lp_2$ is more technical due to the non-integrability of the analogous observable; we treat it in Sections~\ref{section:G2} and~\ref{section:reweight}. The proof of Theorem~\ref{thm:mag} is given at the end of  Section~\ref{section:reweight}.

\subsection{Conformal welding of quantum disks and SLE loop measure}\label{subsec:weld-disk}

In this section we review the conformal welding result established in~\cite[Theorem 1.3]{AHS-SLE-integrability}.
We first recall  the notion of \emph{conformal welding}. 
For concreteness,  suppose $S_1$ and $S_2$ are two oriented Riemann surfaces, each 
conformally equivalent to a planar domain whose boundary consists of finitely many disjoint circles.
For $i=1,2$, suppose $B_i$ is a boundary component of $S_i$ and $\nu_i$ is a finite length measure on $B_i$ with the same total length.  
Given an oriented Riemann surface $S$ and  a simple loop $\eta$ on $S$ with a length measure $\nu$, 
we call $(S,\eta,\nu)$ a \emph{conformal welding} of $(S_1,\nu_1)$ and $(S_2,\nu_2)$ if
the two  connected components of $S\setminus \eta$ with their  orientations inherited from $S$ are  conformally equivalent to $S_1$ and $S_2$, and moreover, 
both $\nu_1$ and $\nu_2$ agree with $\nu$.

We now introduce the notion of uniform conformal welding.
Suppose  $(S_1,\nu_1)$ and $(S_2,\nu_2)$ from the previous paragraph are such that for each $p_1\in B_1$ and $p_2\in B_2$, modulo conformal automorphism 
there exists a unique conformal welding identifying $p_1$ and $p_2$. Now let $\mathbf p_1\in B_1$ and $\mathbf p_2\in B_1$ be independently sampled from the probability measures proportional to $\nu_1$ and  $\nu_2$, respectively. We call the conformal welding of $(S_1,\nu_1)$ and $(S_2,\nu_2)$ with $\mathbf p_1$ identified with $\mathbf p_2$ their  \emph{uniform conformal welding}.

Fix $\kappa\in (0,4)$ and $\gamma=\sqrt{\kappa}\in (0,2)$. 
Recall the quantum sphere and disk  measures $\QS$ and $\QD(\ell)$  from Sections~\ref{subsub:quantum-surface} and~\ref{subsec-QD}. 
For $\ell > 0$, let $\cD_1$ and $\cD_2$ be quantum surfaces sampled from $\QD(\ell) \times \QD(\ell)$.
By Sheffield's work~\cite{shef-zipper}, 
viewed as oriented Riemann surfaces with the quantum  length measure, almost surely 
the conformal welding of $\cD_1$ and $\cD_2$ is  unique after specifying a boundary point on each surface. 
We write  $\mathrm{Weld}(\QD(\ell),\QD(\ell))$ to denote the law of the loop-decorated quantum surface obtained from the uniform conformal welding of  $\cD_1$ and $\cD_2$.

We also recall the SLE loop measure, which is a canonical loop variant of SLE. For $\kappa\in(\frac{8}{3},8)$, suppose $\Gamma$ is a whole-plane $\CLE_\kappa$. Let  $\SLE_\kappa^\lp$ be the law of a loop from $\Gamma$ chosen from counting  measure, i.e., for nonnegative measurable $F$  we define $\int F(\eta)\SLE_\kappa^\lp(d\eta) := \E[ \sum_{\eta \in \Gamma} F(\eta)]$. Up to constants, $\SLE_\kappa^\lp$ agrees with Zhan's SLE loop measure \cite{zhan-loop-measures}; this was implicitly obtained \cite{werner-sphere-cle} for $\kappa \in (8/3, 4]$ (and possibly could be extended to the full range of $\kappa$), and a different proof via LQG was given in \cite[Theorem 1.1]{ACSW24a}. 

According to \cite[Theorem 1.3]{AHS-SLE-integrability}, the conformal welding of two quantum disks gives a quantum sphere decorated with an independent SLE loop. Uniform embedding (as in Theorem~\ref{thm:embed-qs}) then yields the Liouville field decorated by an independent SLE loop, with the multiplicative constant computed in \cite{ACSW24a}. The result is the following.
\begin{proposition}[{\cite[Corollary 9.5]{ACSW24a}}]\label{prop-tabula-rasa}
For $\kappa\in(\frac{8}{3},4)$ and $\gamma = \sqrt\kappa$, we have
	\[\mathbf m_{\wh \C} \ltimes \left( \int_0^\infty \ell \mathrm{Weld}(\QD(\ell), \QD(\ell)) \, d\ell \right) = K \LF_\C \times \SLE_\kappa^\lp\]
 where $ K =\frac{\gamma}{ 8} \frac{\Gamma(\frac{\gamma^2}4) \Gamma(1-\frac{\gamma^2}4)}{(Q-\gamma)^4} \tan(\pi(\frac4{\gamma^2}-1))$. 
\end{proposition}

\subsection{The calculation of \texorpdfstring{$G_3^\lp$}{g}}\label{section:G3}
In this section we compute $G_3^\lp$ (Lemma~\ref{proposition-solve-G3}). A key input is the fact that the quantum length of an SLE curve on an independent LQG field (defined as the quantum boundary length in the domain with the SLE curve removed) is, up to multiplicative constant, the $\frac\gamma2$-GMC measure with respect to Minkowski content as defined in~\eqref{eq-gmc}.

\begin{lemma}[{\cite{benoist-lqg-chaos}}]\label{lem-benoist}
Suppose $\kappa \in (8/3, 4)$ and $\gamma = \sqrt\kappa$. There exists a constant $C = C(\gamma)$ such that the following holds. Sample  $(\phi, \eta) \sim \LF_\C \times \SLE_\kappa^\lp$, then the quantum length measure of $\eta$ with respect to $\phi$ equals $C \cdot  \mathrm{GMC}_{\mathrm{Cont}_\eta, \frac\gamma2, \phi}$. 
\end{lemma}
\begin{proof}
    \cite[Section 3.2]{benoist-lqg-chaos} proves the claim for all $\kappa \in (0,4)$ and $\gamma = \sqrt\kappa$ in the setting where $\phi$ is a variant of the GFF on $\bbH$ and $\eta$ is an independent chordal  SLE$_\kappa$ curve. The version stated here follows by local absolute continuity. 
\end{proof}

Recall $\QD_{0,n}$ from Definition~\ref{def-QD0203}.
 Let $\mathrm{Weld}(\QD(\ell), \QD_{0,n}(\ell))$ be the law of the uniform conformal welding of a sample from $\QD(\ell) \times \QD_{0,n}(\ell)$. Note that $|\QD_{0,n}(\ell)|=\ell^n|\QD(\ell)|$.

\begin{proposition}\label{prop-G3-identity}
	With $K$ and $C$ the constants as in Proposition~\ref{prop-tabula-rasa} and Lemma~\ref{lem-benoist}, we have
\[
	\mathbf m_{\wh \C} \ltimes \left( \int_0^\infty \ell \mathrm{Weld}(\QD(\ell), \QD_{0,3}(\ell)) \, d\ell \right) 
	= K C^3 \LF_\C^{(\frac\gamma2, z_j)_3}(d\phi)  \, \prod_{j=1}^3 \Cont_\eta(d z_j)\,\SLE_\kappa^\lp(d\eta).
\]
\end{proposition}
\begin{proof}
	Take Proposition~\ref{prop-tabula-rasa}, weight by the cubed quantum length of the loop, and sample three points on the loop from the probability measure proportional to quantum length measure. The left hand side of Proposition~\ref{prop-tabula-rasa} becomes the left hand side of our desired identity, and by Lemma~\ref{lem-add-insertion} with $n=3$ applied to the GMC measure $\mathrm{GMC}_{\mathrm{Cont}_\eta, \frac\gamma2, \phi}$ in Lemma~\ref{lem-benoist}, the right hand side of Proposition~\ref{prop-tabula-rasa} becomes the right hand side of our desired identity.
\end{proof}
We will extract the value of $G_3^\lp$ from Proposition~\ref{prop-G3-identity}, by disintegrating on $(z_1,z_2,z_3)$ then taking the expected value of $A e^{-A}$ where $A$ is the quantum area. We need the following integration identity on modified Bessel function from \eqref{eq-Kv} whose proof can be found Appendix \ref{appendix:integration}.

\begin{lemma}\label{lem-int-KK}
	Suppose $c>0$ and real parameters $a,a'$ satisfy $|a-a'|<2$ and $|a+a'| < 2$. Then 
	\[\int_0^\infty \ell K_a(c \ell) K_{a'}(c\ell) \, d\ell = \frac1{2c^2} \frac{\frac\pi2(a-a')}{\sin (\frac\pi2 (a-a'))} \frac{\frac\pi2(a+a')}{\sin (\frac\pi2(a+a'))}. \]
\end{lemma}

\begin{lemma}\label{lem-3-observable}
	Writing $A$ for the random quantum area, we have
	\[\left( \int_0^\infty \ell^4 \mathrm{Weld}(\QD(\ell), \QD(\ell)) \, d\ell \right)[ Ae^{-A}] = \frac{\pi^2R_\gamma^2}{4\Gamma(\frac4{\gamma^2})^2}  (4 \sin \frac{\pi\gamma^2}4)^{\frac12 - \frac4{\gamma^2}} \frac{\frac8{\gamma^2}-1}{\sin (\frac\pi2(\frac8{\gamma^2}-1))}.\]
\end{lemma}
\begin{proof}
	Writing $A_1, A_2$ for the quantum areas of the first and second quantum disks, since $(A_1+A_2)e^{-A_1-A_2} = A_1 e^{-A_1} \cdot e^{-A_2} + e^{-A_1} \cdot A_2 e^{-A_2}$, the desired quantity can be rewritten as 
	\[ \int_0^\infty \ell^4 \left( \QD(\ell)[A_1 e^{-A_1}] \QD(\ell)[e^{-A_2}] + \QD(\ell)[ e^{-A_1}] \QD(\ell)[A_2e^{-A_2}] \right) \, d\ell.\]
	By symmetry and Lemma~\ref{lem-QD-bdy-law} 
	this equals $2R_\gamma^2 \int_0^\infty \ell^{-\frac8{\gamma^2}} \QD(\ell)^\#[A_1 e^{-A_1}] \QD(\ell)^\#[e^{-A_2}] \, d\ell$, and by Proposition~\ref{prop-QD-area}, this is equal to 
	\[2R_\gamma^2\int_0^\infty \ell^{-\frac8{\gamma^2}} \cdot \frac4{\Gamma(\frac4{\gamma^2})^2}\left( \frac{\ell^2}{4 \sin \frac{\pi \gamma^2}4}\right)^{\frac4{\gamma^2}+\frac12}K_{\frac4{\gamma^2}}\left(\frac{\ell}{\sqrt{\sin \frac{\pi\gamma^2}4}}\right) 
	K_{\frac4{\gamma^2} - 1}\left(\frac{\ell}{\sqrt{\sin \frac{\pi\gamma^2}4}}\right)
	\, d\ell.\]
	Combining with Lemma \ref{lem-int-KK} completes the proof. 
\end{proof}

Plugging Lemma \ref{lem-3-observable} into Proposition \ref{prop-tabula-rasa} we can now solve the 3-point Green function for SLE loop measure in the following proposition.
\begin{proposition}\label{proposition-solve-G3}
	With $C$ the constant of Lemma~\ref{lem-benoist}, we have
	\[G_3^\lp(0,1,e^{i\pi/3}) = \frac{2 \pi (\frac2\gamma)^3  (Q-\gamma)^2}{ C^3 C^\mathrm{DOZZ}_\gamma(\frac\gamma2,\frac\gamma2,\frac\gamma2)}  \frac{\Gamma(1-\frac4{\gamma^2})}{\Gamma(\frac4{\gamma^2}) }
	\left( \frac{\pi \Gamma(\frac{\gamma^2}4)}{\Gamma(1-\frac{\gamma^2}4)}\right)^{\frac4{\gamma^2} - \frac32} \times \frac{1}{\Gamma(1-\frac{\gamma^2}{4})^3}
	.\]
\end{proposition}
\begin{proof}
With the choice of $\mathbf m_{\wh \C}$ as in Theorem~\ref{thm:embed-qs}, the density of three marked points under $\mathbf m_{\hat \C}$ is $|(p-q)(q-r)(r-p)|^2 dp \, dq \, dr$. Since $|\QD_{0,3}(\ell)|=\ell^3|\QD(\ell)|$, Proposition~\ref{prop-G3-identity} then gives

	\alb
	&\left( \int_0^\infty \ell^4 \mathrm{Weld}(\QD(\ell), \QD(\ell)) \, d\ell \right)[ Ae^{-A}] |(z_1-z_2)(z_2-z_3)(z_3-z_1)|^2 \prod_{j=1}^3 dz_j \\
	&= KC^3 \LF_\C^{(\frac\gamma2, z_j)_3}[Ae^{-A}] G_3^\lp(z_1,z_2,z_3) \prod_{j=1}^3 dz_j.
	\ale
	Disintegrating and using the triple $(u_1, u_2, u_3) = (0, 1, e^{i\pi/3})$ thus yields 
	\[\left( \int_0^\infty \ell^4 \mathrm{Weld}(\QD(\ell), \QD(\ell)) \, d\ell  \right)[ Ae^{-A}] = 
 KC^3 \LF_\C^{(\frac\gamma2, u_j)_3}[Ae^{-A}] G_3^\lp(u_1,u_2,u_3) . \]
 Lemma~\ref{lem-sph-area-law} gives
$\LF_\C^{(\frac\gamma2, u_j)_3}[Ae^{-A}] = \frac1{2\gamma}(\frac{3\gamma}2 - 2Q){C_\gamma^\mathrm{DOZZ}(\frac\gamma2,\frac\gamma2, \frac\gamma2)}$.
By Lemma~\ref{lem-3-observable} we get
 \alb
 G_3^\lp(0,1,e^{i\pi/3}) &= \frac{\pi^2R_\gamma^2}{4\Gamma(\frac4{\gamma^2})^2}  (4 \sin \frac{\pi\gamma^2}4)^{-(\frac8{\gamma^2}-1)/2} \frac{\frac8{\gamma^2}-1}{\sin (\frac\pi2(\frac8{\gamma^2}-1))}\cdot \frac{2\gamma}{KC^3 (\frac{3\gamma}2 - 2Q){C_\gamma^\mathrm{DOZZ}(\frac\gamma2,\frac\gamma2, \frac\gamma2)}} \\
  &= -\frac{2^{1-\frac8{\gamma^2}}\pi^2}{ C^3 C^\mathrm{DOZZ}_\gamma(\frac\gamma2,\frac\gamma2,\frac\gamma2)}  \frac{(\sin \frac{\pi\gamma^2}4)^{\frac12 - \frac4{\gamma^2}}}{\Gamma(\frac4{\gamma^2})^2 \cos(\pi(\frac4{\gamma^2}-1))} \frac{R_\gamma^2}K.
  \ale
Recall the expressions of $R_\gamma$ and $K$ from Lemma~\ref{lem-QD-bdy-law} and Proposition~\ref{prop-tabula-rasa}, respectively. We have
\begin{equation}\label{eq:k-r}
\frac{R_\gamma^2}{K} = \frac{2^{\frac8{\gamma^2}} (\frac2\gamma)^3 \pi^{\frac8{\gamma^2}-2} (Q-\gamma)^2}{\Gamma(1-\frac{\gamma^2}4)^{\frac8{\gamma^2} +1} \Gamma(\frac{\gamma^2}4)} \cot(\pi(\frac4{\gamma^2}-1)).
\end{equation}
  The result follows by combining the previous two equations and cancelling terms.
\end{proof}

\subsection{The calculation of \texorpdfstring{$G_2^\lp$}{g}}\label{section:G2}

In this section, we compute the 2-point Green function $G_2^\lp$ and conclude the proof of Theorem~\ref{thm:mag} for $\kappa\in (8/3,4)$, modulo an ingredient (Proposition~\ref{lem-G2-alpha}) that we supply in Section~\ref{section:reweight}. Similarly to the 3-point case, using Proposition~\ref{prop-tabula-rasa} we first obtain a conformal welding of quantum disks with two marked boundary points. In order to fix the conformal embedding of the resulting decorated quantum surface, we need a third marked point, which we will sample from the quantum area measure.
\begin{proposition}\label{prop-G2-identity}
With $K$ and $C$ the constants of Proposition~\ref{prop-tabula-rasa} and Lemma~\ref{lem-benoist}, we have
	\[
	\mathbf m_{\wh \C} \ltimes \left( \int_0^\infty 2\ell \mathrm{Weld}(\QD_{1,0}(\ell), \QD_{0,2}(\ell)) \, d\ell \right) 
	= K C^2 \LF_\C^{ (\frac\gamma2, z_j)_2, (\gamma, z_3)}(d\phi)  \, d^2 z_3 \, \prod_{j=1}^2 \Cont_\eta(d z_j)\,\SLE_\kappa^\lp(d\eta).
	\]
\end{proposition}
\begin{proof}
The proof is identical to that of Proposition~\ref{prop-G3-identity} except we add two marked points.
\end{proof}

	Disintegrating Proposition~\ref{prop-G2-identity} on the location of the marked points $(z_1, z_2, z_3)$ gives the following.
	\begin{lemma}\label{lem-G2-gamma}
		Let $(u_1, u_2, u_3) = (0,1,e^{i\pi/3})$. If we embed a sample from $\int_0^\infty 2\ell \mathrm{Weld}(\QD_{1,0}(\ell), \QD_{0,2}(\ell)) \, d\ell$ in $(\C, u_1, u_2, u_3)$, then the law of the (field, curve) pair is 
		\[K C^2 G_2^\lp(u_1,u_2) \LF_\C^{ (\frac\gamma2, u_j)_2, (\gamma, u_3)}(d\phi)\sm^{u_1,u_2}(d\eta),\]
		where $\sm^{u_1,u_2}$ is a probability measure on the set of curves passing through $u_1, u_2$. 
	\end{lemma}

As in the 3-point case, we want to take the expectation of $Ae^{-A}$ where $A$ is the quantum area. However, now the three marked points have weights $(\frac{\gamma}{2},\frac{\gamma}{2},\gamma)$, which do not satisfy the extended Seiberg bound~\eqref{eq-ext-Seiberg}. We will use a variant of Lemma~\ref{lem-G2-gamma} where $\gamma$ has been perturbed (Proposition \ref{lem-G2-alpha}) so the insertions can satisfy the extended Seiberg bound. The proof of Proposition \ref{lem-G2-alpha} is postponed  to Section \ref{section:reweight}.
	
\begin{proposition}\label{lem-G2-alpha}
Let $\alpha \in \R$. If we embed a sample from $\int_0^\infty 2\ell \mathrm{Weld}(\cM_{1,0}^\disk(\alpha; \ell), \QD_{0,2}(\ell)) \, d\ell$ in $(\C, u_1, u_2, u_3)$, then the law of the (field, curve) pair is 
\[\frac{2\pi}\gamma(Q-\gamma)^2 K C^2 G_2^\lp(u_1,u_2) \LF_\C^{ (\frac\gamma2, u_j)_2, (\alpha, u_3)}(d\phi)\sm^{u_1,u_2,u_3}_\alpha(d\eta),\]	
 where $\sm^{u_1,u_2,u_3}_\alpha$ is the measure defined via $\frac{d\sm^{u_1,u_2,u_3}_\alpha}{d\sm^{u_1,u_2}}(\eta) = (\CR(\eta, u_3)/2)^{-\frac{\alpha^2}2 +Q\alpha-2}$, where $\CR(\eta, z)$ is the conformal radius of the connected component of $\C\backslash\eta$ containing $z$, viewed from $z$.
\end{proposition}

To obtain $G_2^\lp(0,1)$, in the setting of Proposition~\ref{lem-G2-alpha} we compute the expected value of $Ae^{-A}$ in two different ways (Lemmas~\ref{lem-alpha-obs} and~\ref{lem-lim-LF}). Recall that $|\QD_{0,2}(\ell)|=\ell^2|\QD(\ell)|$.

\begin{lemma}\label{lem-alpha-obs}
	Let $\alpha \in (\frac\gamma2, \gamma)$, then $\left(\int_0^\infty 2\ell^3 \mathrm{Weld}(\cM_{1,0}^\disk(\alpha; \ell), \QD(\ell)) \, d\ell\right)[Ae^{-A}]$ equals
	\alb
	\frac2\gamma 2^{-\alpha^2/2} \ol U(\alpha) R_\gamma  \frac{(4\sin\frac{\pi\gamma^2}4)^{\frac1\gamma(\alpha-2Q)+1}}{\Gamma(\frac2\gamma(Q-\alpha))\Gamma(\frac4{\gamma^2})}  
	\frac{\pi }{\sin(\pi\frac\alpha\gamma)}
	\frac{\pi (\frac4{\gamma^2} - \frac\alpha\gamma)}{\sin(\pi (\frac4{\gamma^2} - \frac\alpha\gamma))}.
	\ale
\end{lemma}
\begin{proof}
	Lemmas~\ref{lem-M1-bdy-law} and~\ref{lem-QD-bdy-law} give
	\[\int_0^\infty 2\ell^3 \mathrm{Weld}(\cM_{1,0}^\disk(\alpha; \ell), \QD(\ell))\,d\ell =  \frac2\gamma 2^{1-\alpha^2/2} \ol U(\alpha) R_\gamma \int_0^\infty \ell^{\frac2\gamma\alpha -\frac8{\gamma^2} -1} \mathrm{Weld}(\cM_{1,0}^\disk(\alpha; \ell)^\#, \QD(\ell)^\#) \, d\ell. \]
 Writing $A_1$ and $A_2$ to denote the quantum areas of samples from $\cM_{1,0}^\disk(\alpha; \ell)^\#$ and $\QD(\ell)^\#$, we have $Ae^{-A} = A_1 e^{-A_1} \cdot e^{-A_2} + e^{-A_1} \cdot A_2 e^{-A_2}$. For the first summand, we have  
	\alb
&\int_0^\infty \ell^{\frac2\gamma \alpha - \frac{8}{\gamma^2} -1}\cM_{1,0}^\disk(\alpha; \ell)^\# [A_1 e^{-A_1}] \QD(\ell)^\#[e^{-A_2}] \, d\ell  \\
&= \int_0^\infty \ell^{\frac2\gamma \alpha - \frac{8}{\gamma^2} -1}\cdot \frac4{\Gamma(\frac2\gamma(Q-\alpha))\Gamma(\frac4{\gamma^2})} (\frac{\ell^2}{4 \sin \frac{\pi\gamma^2}4})^{(\frac2\gamma(Q-\alpha)+1 + \frac4{\gamma^2})/2} K_{\frac2\gamma(Q-\alpha)-1}(\frac{\ell}{\sqrt{\sin \frac{\pi\gamma^2}4}})  K_{\frac4{\gamma^2}}(\frac{\ell}{\sqrt{\sin \frac{\pi\gamma^2}4}})  \, d\ell  \\
&= \frac{(4\sin\frac{\pi\gamma^2}4)^{\frac1\gamma(\alpha-2Q)+1}}{2\Gamma(\frac2\gamma(Q-\alpha))\Gamma(\frac4{\gamma^2})}  
\frac{\pi\frac\alpha\gamma }{\sin(\pi\frac\alpha\gamma)}
\frac{\pi (\frac4{\gamma^2} - \frac\alpha\gamma)}{\sin(\pi (\frac4{\gamma^2} - \frac\alpha\gamma))}.
	\ale
	Here, the first equality uses Propositions~\ref{prop-fzz} and~\ref{prop-QD-area} to identify the area laws of $\cM_{1,0}^\disk(\alpha; \ell)^\#$ and $\QD(\ell)^\#$ to compute the expectations, and the second equality uses the integral identity  Lemma~\ref{lem-int-KK}. Similarly,
$$ \int_0^\infty \ell^{\frac2\gamma \alpha - \frac{8}{\gamma^2} -1}\cM_{1,0}^\disk(\alpha; \ell)^\# [e^{-A_1}] \QD(\ell)^\#[A_2e^{-A_2}] \, d\ell= \frac{(4\sin\frac{\pi\gamma^2}4)^{\frac1\gamma(\alpha-2Q)+1}}{2\Gamma(\frac2\gamma(Q-\alpha))\Gamma(\frac4{\gamma^2})}  
	\frac{\pi(1-\frac\alpha\gamma) }{\sin(\pi(1-\frac\alpha\gamma))}
	\frac{\pi (\frac4{\gamma^2} - \frac\alpha\gamma)}{\sin(\pi (\frac4{\gamma^2} - \frac\alpha\gamma))}.$$
    
	Summing these gives
	\[\int_0^\infty \ell^{\frac2\gamma\alpha -\frac8{\gamma^2} -1} \mathrm{Weld}(\cM_{1,0}^\disk(\alpha; \ell)^\#, \QD(\ell)^\#) \, d\ell [Ae^{-A}] = \frac{(4\sin\frac{\pi\gamma^2}4)^{\frac1\gamma(\alpha-2Q)+1}}{2\Gamma(\frac2\gamma(Q-\alpha))\Gamma(\frac4{\gamma^2})}  
	\frac{\pi }{\sin(\pi\frac\alpha\gamma)}
	\frac{\pi (\frac4{\gamma^2} - \frac\alpha\gamma)}{\sin(\pi (\frac4{\gamma^2} - \frac\alpha\gamma))},\]
which combined with the first indented equation yields the result. 	
\end{proof}

\begin{lemma}\label{lem-lim-LF}
	We have 
	\[\lim_{\alpha \uparrow \gamma} (\gamma - \alpha)\LF_\C^{(\frac\gamma2, u_1), (\frac\gamma2, u_2), (\alpha, u_3)}[Ae^{-A}] = -\frac8{\gamma^3} (Q-\gamma)\pi^{\frac4{\gamma^2}-1} \frac{\Gamma(1- \frac4{\gamma^2})}{\Gamma(\frac4{\gamma^2})} (\frac{\Gamma(\frac{\gamma^2}4)}{\Gamma(1-\frac{\gamma^2}4)})^{\frac4{\gamma^2}}  .\]
\end{lemma}
\begin{proof}
	By the first claim of Lemma~\ref{lem-sph-area-law} and 
	$\int_0^\infty a^{\frac1\gamma(\gamma +\alpha - 2Q)} e^{-a} \, da = \Gamma(\frac1\gamma(\gamma + \alpha - 2Q) + 1)$,
 we have
	\alb
	\LF_\C^{(\frac\gamma2, u_1), (\frac\gamma2, u_2), (\alpha, u_3)}[Ae^{-A}] &= \frac{1}{2 \gamma}(\gamma +\alpha - 2Q){C_\gamma^\mathrm{DOZZ}(\frac\gamma2,\frac\gamma2, \alpha)} \\
	&=  \frac{1}{2 \gamma}(\gamma +\alpha - 2Q) (\pi  (\frac\gamma2)^{2-\gamma^2/2}\frac{\Gamma(\frac{\gamma^2}4)}{\Gamma(1-\frac{\gamma^2}4)})^{\frac{2Q-\gamma-\alpha}{\gamma}} \frac{\Ug'(0) \Ug(\frac\gamma2)^2\Ug(\alpha)}{\Ug(\frac{\gamma+\alpha-2Q}2) \Ug(\frac{\gamma-\alpha}2) \Ug(\frac\alpha2)^2}.
	\ale
	Note that $\lim_{\alpha \uparrow \gamma} \frac{\gamma - \alpha}{\Ug(\frac{\gamma - \alpha}2)} = \frac2{\Ug'(0)}$, so 
	\[\lim_{\alpha \uparrow \gamma} (\gamma - \alpha)\LF_\C^{(\frac\gamma2, u_1), (\frac\gamma2, u_2), (\alpha, u_3)}[Ae^{-A}]  =\frac{1}{\gamma} 2(\gamma - Q) (\pi  (\frac\gamma2)^{2-\gamma^2/2}\frac{\Gamma(\frac{\gamma^2}4)}{\Gamma(1-\frac{\gamma^2}4)})^{\frac2\gamma(Q-\gamma)} \frac{\Ug(\gamma)}{\Ug(\gamma - Q) }.\]
	The shift equations~\eqref{shiftUpsilon} for $\Ug$ yield the following which conclude the proof:
	\[\frac{\Ug(\gamma)}{\Ug(\frac\gamma2)} = \frac{\Gamma(\frac{\gamma^2}4)}{\Gamma(1-\frac{\gamma^2}4)} \left(\frac\gamma2\right)^{1-\frac{\gamma^2}2}  \quad  \textrm{and}\quad \frac{\Ug(\frac\gamma2)}{\Ug(\gamma-Q)} = \frac{\Gamma(1- \frac4{\gamma^2})}{\Gamma(\frac4{\gamma^2})}\left(\frac\gamma2\right)^{1-\frac8{\gamma^2}}. \qedhere  \]
\end{proof}

Combining Lemma \ref{lem-alpha-obs} and \ref{lem-lim-LF} with Proposition \ref{lem-G2-alpha}, we can finally solve $G_2$.

\begin{proposition}\label{proposition-solve-G2}
	With $C$ the constant in Lemma~\ref{lem-benoist}, we have
	\[G_2^\lp(0,1) =2 \frac{(Q-\gamma)^2 \cot(\pi(\frac4{\gamma^2}-1))}{C^2 \pi \Gamma(\frac{\gamma^2}4) \Gamma(1-\frac{\gamma^2}4)}. \]
\end{proposition}
\begin{proof}
	Let $(u_1, u_2, u_3) = (0,1,e^{i\pi/3})$ and let $\alpha \in (\frac\gamma2, \gamma)$. 
	By Proposition~\ref{lem-G2-alpha} we have, viewing both sides as measures on decorated quantum surfaces, 
	\[ \int_0^\infty 2\ell\mathrm{Weld}(\cM_{1,0}^\disk(\alpha; \ell), \QD_{0,2}(\ell)) \, d\ell = \frac{2\pi}\gamma(Q-\gamma)^2 K C^2 G_2^\lp(u_1,u_2) \LF_\C^{ (\frac\gamma2, u_j)_2, (\alpha, u_3)}(d\phi)\sm^{u_1,u_2,u_3}_\alpha(d\eta).\]
	We evaluate the observable $Ae^{-A}$ on both sides. By Lemma~\ref{lem-alpha-obs} we get
 \begin{gather}\label{eq-solve-G2}
 \begin{aligned}
	&\frac2\gamma 2^{-\alpha^2/2} \ol U(\alpha) R_\gamma \cdot \frac{(4\sin\frac{\pi\gamma^2}4)^{\frac1\gamma(\alpha-2Q)+1}}{\Gamma(\frac2\gamma(Q-\alpha))\Gamma(\frac4{\gamma^2})}  
	\frac{\pi }{\sin(\pi\frac\alpha\gamma)}
	\frac{\pi (\frac4{\gamma^2} - \frac\alpha\gamma)}{\sin(\pi (\frac4{\gamma^2} - \frac\alpha\gamma))} \\
	=& \frac{2\pi}\gamma(Q-\gamma)^2 K C^2 G_2^\lp(u_1,u_2) \LF_\C^{ (\frac\gamma2, u_j)_2, (\alpha, u_3)}(d\phi)[Ae^{-A}] |\sm^{u_1,u_2,u_3}_\alpha| .
 \end{aligned}
 \end{gather}
Note that this implies $|\sm^{u_1,u_2,u_3}_{\alpha}|  = \mathsf m^{u_1,u_2}[ (\frac12\CR(\eta, u_3))^{-\frac{\alpha^2}2 + Q\alpha -2}]<\infty$.  The exponent $-\frac{\alpha^2}2 + Q\alpha -2 < 0$ is negative and increasing on $\alpha \in (\frac\gamma2, \gamma)$, so for any $\frac\gamma2 < \alpha_0 < \alpha < \gamma$ and $x > 0$ we have $x^{-\frac{\alpha^2}2 + Q\alpha -2} < \max(1, x^{-\frac{\alpha_0^2}2 + Q\alpha_0 -2})$. The dominated convergence theorem thus implies
\[\lim_{\alpha \uparrow \gamma} |\sm^{u_1,u_2,u_3}_\alpha| = \lim_{\alpha \uparrow \gamma}  \mathsf m^{u_1,u_2} [(\frac12\CR(\eta, u_3))^{-\frac{\alpha^2}2 + Q\alpha -2}] =   \mathsf m^{u_1,u_2} [\lim_{\alpha \uparrow \gamma} (\frac12\CR(\eta, u_3))^{-\frac{\alpha^2}2 + Q\alpha -2}] = 1,\]
where the dominating function is $\max(1, (\frac12\CR(\eta, u_3))^{-\frac{\alpha_0^2}2 + Q\alpha_0 -2})$, and we use  $-\frac{\gamma^2}2 + Q\gamma - 2 = 0$.

	Multiply both sides of~\eqref{eq-solve-G2} by $(\gamma - \alpha)$ and take the limit as $\alpha \uparrow \gamma$. By Lemma~\ref{lem-lim-LF} and $\lim_{\alpha \uparrow \gamma} |\sm^{u_1,u_2,u_3}_\alpha| = 1$, this gives
	
	\alb
	&\frac2\gamma 2^{-\gamma^2/2} \ol U(\gamma) R_\gamma \cdot \frac{(4\sin\frac{\pi\gamma^2}4)^{1-\frac4{\gamma^2}}}{\Gamma(\frac4{\gamma^2}-1)\Gamma(\frac4{\gamma^2})}  
	\cdot \gamma\cdot
	\frac{\pi (\frac4{\gamma^2} - 1)}{\sin(\pi (\frac4{\gamma^2} - 1))} \\
	=& -\frac{2\pi}\gamma(Q-\gamma)^2 K C^2 G_2^\lp(u_1,u_2)\frac8{\gamma^3} (Q-\gamma)\pi^{\frac4{\gamma^2}-1} \frac{\Gamma(1- \frac4{\gamma^2})}{\Gamma(\frac4{\gamma^2})} (\frac{\Gamma(\frac{\gamma^2}4)}{\Gamma(1-\frac{\gamma^2}4)})^{\frac4{\gamma^2}} . \ale
 The result follows by expanding  $\ol U(\gamma)$, $R_\gamma$ and $K$ (recall Proposition \ref{prop-remy-U}), Lemma \ref{rem-QD-length} and Proposition~\ref{prop-tabula-rasa} and simplifying.
\end{proof}

\begin{proof}[{Proof of Theorem~\ref{thm:mag}} for $\kappa \in (\frac83, 4{]}$] For $\kappa\in(\frac{8}{3},4)$, Proposition~\ref{proposition-solve-G3} gives a formula for $G_3^\lp(0,1,e^{i\pi/3})$, and Proposition~\ref{proposition-solve-G2} gives a formula for $G_2^\lp(0,1)$.
By plugging in these formulas, we get
\begin{equation}\label{eq:mag-exp}
\frac{G_3^\lp(0,1,e^{i\pi/3})}{G_2^\mathrm{loop}(0,1)^{3/2}} = \frac{2^{5/2}\pi \gamma^{-3}}{C_\gamma^\mathrm{DOZZ}(\frac\gamma2, \frac\gamma2, \frac\gamma2)}\frac{\Gamma(1-\frac4{\gamma^2})}{\Gamma(\frac4{\gamma^2})}\frac{\tan(\pi(\frac4{\gamma^2}-1))^{3/2}}{Q-\gamma} \left(\frac{\pi \Gamma(\frac{\gamma^2}4)}{\Gamma(1-\frac{\gamma^2}4)}\right)^{\frac4{\gamma^2}}.
\end{equation}
By shift equation~\eqref{eq:shift double gamma}, the above formula equals $\frac{1}{\sqrt{-2\cos(\pi\beta^2)}}\omega^{(\beta)}_{(1,0),(1,0),(1,0)}$. By covariance of Green function this equals $R^\lp(\kappa)$. 
The case $\kappa=4$ follows from taking the limit as $\kappa \nearrow 4$; see Lemmas~\ref{lem:Green-cont loop} and~\ref{lem:cont4} in Appendix~\ref{kappato4}.
\end{proof}

\subsection{Proof of Proposition \ref{lem-G2-alpha}: the reweighting argument}\label{section:reweight}
Our proof follows the strategy from \cite{AHS-SLE-integrability}. Similar arguments will be also used in Section~\ref{sec:3-pt} for the proof of Theorem~\ref{thm:nesting}. We give a detailed proof here and will be brief there. We start with the $\alpha = \gamma$ case and then perform a reweighting argument.
\begin{lemma}\label{lem:base}
Proposition \ref{lem-G2-alpha} holds for the case $\alpha=\gamma$.
\end{lemma}
\begin{proof}
This is Lemma~\ref{lem-G2-gamma} after identifying $\cM_{1,0}^\disk(\gamma; \ell) = \frac{2\pi}\gamma(Q-\gamma)^2 \QD_{1,0}(\ell)$ (Definition~\ref{def-QD-alpha}).
\end{proof}
We record two ``reweighting'' lemmas that change the $\gamma$-insertion of a Liouville field to a generic $\alpha$-insertion. Applying them to the setting of Lemma~\ref{lem:base} yields Proposition \ref{lem-G2-alpha} for general $\alpha$.
\begin{lemma}[{\cite[Lemma 4.6]{ARS-FZZ}}]\label{lem-disk-reweight}
For any $\ell>0,\eps\in (0,1)$ and for any nonnegative  measurable function $f$ of $X$ that depends only on $X|_{\bbH\setminus B_\eps(i)}$,  we have 
	\begin{equation*}
		\int f(X|_{\bbH\setminus B_\eps(i)}) \times   \eps^{\frac12(\alpha^2 - \gamma^2)} e^{(\alpha - \gamma)X_\eps(i)}  \,d   {\LF_\bbH^{(\gamma, i)}(\ell)}= \int f(X|_{\bbH\setminus B_\eps(i)}) \, d\LF_{\bbH}^{(\alpha, i)}(\ell),
	\end{equation*}
	where $X$ is a sample in $H^{-1}(\bbH)$ and $X_\eps(i)$ means the average of $X$ on the boundary of the ball $B_\eps(i)=\{z: |z-i|< \eps\}$. 
\end{lemma}

\begin{lemma}\label{lem:reweight}
 Let $n \geq 0$, let   $(\alpha_i, z_i) \in \R \times \C$ for $i = 1, \dots, n$ and $(\alpha, z) \in \R \times \C$ satisfy that $z,z_1, \dots, z_n$ are pairwise distinct. Let $\eta$ be a simple loop not passing through $z$, and let $D_\eta$ be the connected component of $\wh \C \backslash \eta$ containing $z$. Let $p$ be a point on $\eta$ and let $\psi: \bbH \to D_\eta$ be the conformal map with $\psi(i) = z$ and $\psi(0) = p$. For $\eps \in (0,\frac14)$, let $\C_{\eta, p, \eps} = \C \backslash \psi(B_\eps(i))$. For $\phi$ sampled from $\LF_\C^{(\alpha_i, z_i)_i, (\gamma, z)}$, let $X = \phi \circ \psi + Q \log |\psi'|$, so that $(\bbH, X, i, 0)/{\sim_\gamma} = (D_\eta, \phi, z, p)/{\sim_\gamma}$. Then for any nonnegative measurable function $f$ of $\phi$ that depends only on $\phi|_{\C_{\eta, p, \eps}}$, we have 
 \[\int f(\phi)\times \eps^{\frac12(\alpha^2-\gamma^2)} e^{(\alpha - \gamma)X_\eps(i)} d\LF_\C^{(\alpha_i, z_i)_i, (\gamma, z)} = \int f(\phi) \left( \frac12\CR(\eta, z)\right)^{-\frac{\alpha^2}2 + Q\alpha-2} d\LF_\C^{(\alpha_i, z_i)_i, (\alpha, z)}, \]
 where $\CR(\eta, z)$ is the conformal radius of the connected component of $\C\backslash\eta$ containing $z$, viewed from $z$.
\end{lemma}
\begin{proof}
The case of $i=1$ with $(\alpha_i,z_i)=(\gamma,1)$ and $z=0$ is proved in \cite[Lemma 8.9]{ACSW24a}. The same proof also works for the general case. \qedhere 	 
\end{proof}

\begin{proof}[Proof of Proposition \ref{lem-G2-alpha}]
By Lemma~\ref{lem:base}, when embedding a sample from $\int_0^\infty 2\ell \mathrm{Weld}(\cM_{1,0}^\disk(\gamma; \ell), \QD_{0,2}(\ell)) \, d\ell$ in $(\C, u_1, u_2, u_3)$, the resulting law of the (field, curve) pair is 
\[\frac{2\pi}\gamma(Q-\gamma)^2 K C^2 G_2^\lp(u_1,u_2) \LF_\C^{ (\frac\gamma2, u_j)_2, (\alpha, u_3)}(d\phi)\sm^{u_1,u_2}(d\eta),\]	
 where $\sm^{u_1,u_2}$ is a probability measure on simple loops passing through $u_1,u_2$, defined in Lemma~\ref{lem-G2-gamma}.
	Recall the notation in Lemma~\ref{lem:reweight}.
	Then for $\eps\in (0,\frac14)$, for any nonnegative  measurable function $f$ of $\phi|_{\C_{\eta,u_1,\eps}}$, and any nonnegative measure function $g$ of $\eta$, we have
	\begin{align}
		&\int f(\phi|_{\C_{\eta,u_1,\eps}})g(\eta) \eps^{\frac12(\alpha^2 - \gamma^2)} e^{(\alpha - \gamma)X_\eps(i)} \LF_\C^{ (\frac\gamma2, u_j)_2, (\gamma, u_3)}(d\phi)\sm^{u_1,u_2}(d\eta)\nonumber\\
		&= \int_0^\infty \left(\int f(\phi|_{\C_{\eta,u_1,\eps}})g(\eta) \eps^{\frac12(\alpha^2 - \gamma^2)} e^{(\alpha - \gamma)X_\eps(i)}  \ell\LF_\bbH^{(\gamma,i)} (\ell) \times\QD_2(\ell)\right)d\ell. \label{eq:G2-gamma-eps}
	\end{align} 
	Recall that \(\sm^{u_1,u_2,u_3}_{\alpha} =   \left(\frac12\CR(\eta, u_3)\right)^{-\frac{\alpha^2}2 + Q \alpha  - 2}  \sm\). By Lemma~\ref{lem:reweight}, the left side of~\eqref{eq:G2-gamma-eps} equals 
	\begin{equation}\label{eq:G2-compare1}
		\int f g  \left(\frac12\CR(\eta,u_3))\right)^{-\frac{\alpha^2}2 + Q \alpha  - 2} \LF_\C^{ (\frac\gamma2, u_j)_2, (\gamma, u_3)}(d\phi)\sm^{u_1,u_2}(d\eta)\nonumber
		=  \int f g \LF_\C^{(\frac\gamma2, u_j)_2, (\alpha, u_3)} (d\phi)\sm^{u_1,u_2,u_3}_{\alpha}\,(d\eta). 
	\end{equation}
	Here we write $f=f(\phi|_{\C_{\eta,u_1,\eps}})$ and $g=g(\eta)$ to ease the notation. 
	
	By Lemma~\ref{lem-disk-reweight}, the right side of~\eqref{eq:G2-gamma-eps} equals 
\begin{equation}\label{eq:G2-compare2}
	\int_0^\infty \left(\int fg \ell\LF_\bbH^{(\alpha,i)} (\ell) \times \QD_2(\ell)\right)d\ell.
\end{equation}
 
Since $\eps,f,g$  are arbitrary, comparing~\eqref{eq:G2-compare1} and~\eqref{eq:G2-compare2}  we obtain that embedding a sample from $\int_0^\infty 2\ell \mathrm{Weld}(\cM_{1,0}^\disk(\alpha; \ell), \QD_{0,2}(\ell)) \, d\ell$ in $(\C, u_1, u_2, u_3)$ gives a (field, curve) pair with law 
\[\frac{2\pi}\gamma(Q-\gamma)^2 K C^2 G_2^\lp(u_1,u_2) \LF_\C^{ (\frac\gamma2, u_j)_2, (\alpha, u_3)}(d\phi)\sm^{u_1,u_2,u_3}_\alpha(d\eta).\]
The proposition then follows.
\end{proof}

\section{SLE loop Green function: the non-simple case}\label{sec-loop-green-nonsimple}

In this section, we prove the non-simple case of Theorem~\ref{thm:mag}. The main new ingredient is the notion of generalized quantum disk, which we  review in Section \ref{section:GQS}. This notion is necessary whenever we deal with non-simple CLE. Other types of generalized surfaces will be used in the proof of Theorem~\ref{thm:nesting}.
In Section~\ref{section:ns-calculation} we obtain $G_2^\lp$ and $G_3^\lp$ using arguments parallel to those in Section~\ref{sec-loop-green-simple}.  Throughout this section we set $\kappa'\in(4,8)$ and $\gamma=4/\sqrt{\kappa'}\in(\sqrt2,2)$.

\subsection{Generalized quantum disk}\label{section:GQS}

In this section, we review the notion of generalized quantum surface and basic properties of generalized quantum disks; see \cite{wedges,msw-non-simple,holden2022liouville,nonsimple-welding} for further background. 

Consider a pair $(D,h)$ where $D$ is a closed set (not necessarily homeomorphic to a closed disk) such that each component of its interior together with its prime-end boundary is homeomorphic to the closed disk, and $h$ is only defined as a distribution on each of these components. We define the equivalence relation $\sim_\gamma$ such that $(D,h)\sim_\gamma(D',h')$ if there is a homeomorphism $g:D\to D'$ that is conformal on each component of the interior of $D$, and $h'=h\circ g^{-1}+Q\log|(g^{-1})'|$. A {\it generalized quantum surface} $S$ is an equivalence class of pairs $(D,h)$ under $\sim_\gamma$, and we call a choice of representative $(D, h)$ an embedding of a generalized quantum surface. Generalized quantum surfaces decorated by curves and marked points can be defined analogously.

We first introduce the notion of {\it loop-tree}. Fix $\gamma\in (\sqrt{2},2)$, and let $(X_s)_{s\geq 0}$ be a  c\`adl\`ag function with only positive jumps. Then on the graph $\{(s,X_s)\}$ we define $(s,X_s)\sim (t,X_t)$ when $X_s=X_t$ and the horizontal segment connecting $(s,X_s)$ and $(t,X_t)$ is below the graph of $X|_{[s,t]}$. We also set $(t,X_t)\sim(t,X_{t-})$ if $t$ is a jump time of $X$. We call the quotient $\mathcal{T}$ of the graph $\{(s,X_s)\}$ under the above equivalence relation $\gamma$ to be the \emph{loop-tree} corresponding to $X$.

The \emph{forested line} is defined by independently gluing quantum disks on each loop of the loop-tree.
\begin{definition}\label{def:forested-line}
    For $\gamma\in (\sqrt{2},2)$, let $(X_s)_{s\geq 0}$ be a stable L\'evy process of index $\frac{4}{\gamma^2}$ with non-negative jumps starting from $0$. For $t>0$, let $Y_t=\inf\{s>0:X_s\leq -t\}$. Fix the root $o=(0,0)$, and define the forested line of length $t$ as follows: on each loop of length $L$ of the loop-tree corresponding to $X$, independently sample a quantum disk from $\QD(L)^\#$ and then topologically identify its boundary with the loop (with a uniformly-chosen rotation). For two points on a forested line, we define the {\rm generalized quantum length} between them to be the length of the corresponding time interval for $(X_s)$.
\end{definition}

We can also {\it foresting} a topological loop with length $t$ as follows. Suppose $\mathcal{L}_t$ is a forested line of length $t$, and $p$ is chosen uniformly by the probability measure $\frac{1}{L}1_{[0,L]}dx$. Then glue $\mathcal{L}_t$ on the loop, with two endpoints identified with $p$, and forget $p$. The forested disks can then be obtained by foresting the boundaries of quantum disks.

\begin{definition}\label{def:forested-disk}
For $\alpha > \frac\gamma2$, the forested disk with one bulk insertion $\cM_{1,0}^\mathrm{f.d.}(\alpha)$ is defined as the generalized quantum surface obtained by foresting the boundary of a sample from $\cM_{1,0}^\mathrm{disk}(\alpha)$. In particular, for the case $\alpha=\gamma$, we define the generalized quantum disk with one bulk marked point as $\GQD_{1,0}:=\frac{\gamma}{2\pi(Q-\gamma)^2}\cM_{1,0}^\mathrm{f.d.}(\gamma)$.
\end{definition}

We define $\cM_{1,0}^\mathrm{f.d.}(\alpha;\ell)$ (resp. $\GQD_{1,0}(\ell)$) to be the disintegration of the generalized boundary lengths such that $\cM_{1,0}^\mathrm{f.d.}(\alpha)=\int_0^\infty \cM_{1,0}^\mathrm{f.d.}(\alpha;\ell)d\ell$ (resp. $\GQD_{1,0}=\int_0^\infty\GQD_{1,0}(\ell)d\ell$).
We also define the law $\GQD$ by weighting the inverse of its quantum area on the law $\GQD_{1,0}$ and forgetting the bulk insertion point, and $\GQD_{0,1}$ by weighting the boundary length on the law $\GQD$ and sampling a boundary point according to the probability measure proportional to the generalized quantum length; similarly for $\GQD(\ell)$ and $\GQD_{0,1}(\ell)$.
\begin{remark}
The definition of $\GQD$ often has different expressions in different literature, e.g. \cite{msw-non-simple, holden2022liouville, nonsimple-welding}. As explained in \cite[Remark 3.12]{nonsimple-welding} and \cite[Proposition 2.34]{ACSW24a}, these definitions are all equivalent to our definition of $\GQD$ up to a multiplicative constant. The choice of multiplicative constant will not affect our results, since it will cancel out in our computations.
\end{remark}

The following lemmas give the law of the generalized boundary length under $\GQD$, and  the law of the quantum area of a generalized quantum disk 
with given generalized boundary length.
\begin{lemma}[{\cite[Lemma 2.28]{ACSW24a}}]\label{lem:gqd-bdy-length-law}
$|\GQD(\ell)|=R_\gamma'\cdot \ell^{-\frac{\gamma^2}{4}-2}d\ell$ for some constant $R_\gamma'\in(0,\infty)$.
\end{lemma}
\begin{prop}[{\cite[Theorem 1.8]{holden2022liouville}}]\label{prop:ns-area-law}
Recall $K$ and $\ol K$ from \eqref{eq-Kv}. Denote the quantum area to be $A$ and let $M'=2\left(\frac{\mu}{4\sin\frac{\pi\gamma^2}{4}}\right)^{\frac{\kappa'}{8}}$ for $\mu>0$.
Under $\GQD(\ell)^\#$ we have
\[ \GQD(\ell)^\#[e^{-\mu A}]=\ol K_{4/\kappa'}(M'\ell)\quad \textrm{and} \quad \GQD(\ell)^\#[Ae^{-\mu A}]=\frac{\kappa'}{2\mu\Gamma(\frac{4}{\kappa'})}\left(\frac{M'l}{2}\right)^{4/\kappa'+1}K_{1-\frac{4}{\kappa'}}(M'\ell).\]
In particular, we have
\[\GQD(\ell)^\#[A]=\frac{\kappa'}{16\sin\frac{\pi\gamma^2}{4}}\ell^{8/\kappa'}\frac{\Gamma(1-\frac{4}{\kappa'})}{\Gamma(\frac{4}{\kappa'})}\quad \textrm{and} \quad  |\GQD_{1,0}(\ell)|=R_\gamma'\frac{\kappa'}{16\sin\frac{\pi\gamma^2}{4}}\frac{\Gamma(1-\frac{4}{\kappa'})}{\Gamma(\frac{4}{\kappa'})}\ell^{\frac{4}{\kappa'}-2},\]
and $\GQD_{1,0}(\ell)^\#[e^{-\mu A}]=\ol K_{1-4/\kappa'}(M'\ell)$.
\end{prop}

 For two independent generalized quantum surfaces $\mathcal{D}_1,\mathcal{D}_2$ which have distinguished boundary arcs of the same generalized boundary length, by \cite{wedges}, there exists a conformal welding measurable with respect to $(\mathcal{D}_1,\mathcal{D}_2)$ where the welding interface is locally absolutely continuous with respect to $\SLE_{\kappa'}$. In this paper, we always consider this notion of conformal welding for generalized quantum surfaces, and the notion of uniform conformal welding  $\Weld(\cdot,\cdot)$ can be naturally extended to this setting. We now recall the counterpart of Proposition \ref{prop-tabula-rasa}, which says the (uniform) conformal welding of two independent generalized quantum disks, under the uniform embedding, gives the whole-plane Liouville field decorated by a $\SLE_{\kappa'}$ loop (see Section~\ref{subsec:weld-disk} for its definition). 
\begin{prop}[{\cite[Corollary 9.8]{ACSW24a}}]\label{weldgqd}
For $\kappa'\in(4,8)$, we have
\begin{equation*}
\mathbf m_{\wh \C}\ltimes\left(\int_0^\infty \ell\Weld(\GQD(\ell),\GQD(\ell))d\ell\right)=K'\LF_\C\times\SLE_{\kappa'}^{\mathrm{loop}}
\end{equation*}
where $K'$ is expressed in terms of $R_\gamma'$ from Lemma~\ref{lem:gqd-bdy-length-law}  by
\begin{equation}\label{wc}
\frac{K'}{{R_\gamma'}^2}=-\left(\frac{2}{\gamma}\right)^5\pi^{-\frac{8}{\gamma^2}+2}2^{-\frac{8}{\gamma^2}}\frac{\Gamma(\frac{\kappa'}{4}-1)}{\Gamma(2-\frac{\kappa'}{4})}\Gamma(1-\frac{\gamma^2}{4})^{\frac{8}{\gamma^2}+2}\tan\pi\left(\frac{4}{\kappa'}-1\right).
\end{equation}
\end{prop}

\subsection{Proof of Theorem~\ref{thm:mag} in non-simple case}\label{section:ns-calculation}
In this section, we compute $G_2^\lp$ and $G_3^\lp$ for non-simple CLE, and hence prove Theorem \ref{thm:mag} for the non-simple case. The arguments are straightforward modifications of those in Section~\ref{sec-loop-green-simple} but with  quantum disks replaced by generalized quantum disks (and quantum length replaced by generalized boundary length), so we will be brief overall and highlight the difference.
We first solve for $G_3^\lp$, following the argument in Section \ref{section:G3}. The following lemma shows that the $\frac{2}{\gamma}$-GMC measure over the $(1+\frac{\kappa'}{8})$-dimensional Minkowski content of a $\SLE_{\kappa'}$-type curve equals its generalized quantum length up to a multiplicative constant, which is the non-simple counterpart of Lemma \ref{lem-benoist}.

\begin{lemma}\label{lem-benoist-ns}
Suppose $\kappa' \in (4,8)$. There exists a constant $C' = C'(\gamma)$ such that the following holds. Sample  $(\phi, \eta) \sim \LF_\C \times \SLE_{\kappa'}^\lp$, then the generalized quantum length measure of $\eta$ with respect to $\phi$ equals $C' \cdot  \mathrm{GMC}_{\mathrm{Cont}_\eta, \frac2\gamma, \phi}$. 
\end{lemma}

\begin{proof}
This is stated in Section 5 of \cite{benoist-lqg-chaos}. Its proof is same as the $\kappa \in (0,4)$ case, via the ergodicity argument used in \cite{benoist-lqg-chaos}.
\end{proof}

Now we state the counterpart of Proposition \ref{prop-G3-identity}. Let $\GQD_{0,n}(\ell)$ be the law of a sample from $\ell^n\GQD(\ell)$ with $n$ points independently sampled from the probability measure proportionate to the generalized boundary length measure. Let $\mathrm{Weld}(\GQD(\ell), \GQD_{0,n}(\ell))$ denote the law of the uniform conformal welding of a sample from $\GQD(\ell) \times \GQD_{0,n}(\ell)$.

\begin{proposition}\label{prop-G3-identity-ns}
	With $K'$ and $C'$ the constants of Proposition~\ref{weldgqd} and Lemma~\ref{lem-benoist-ns}, we have
\[
	\mathbf m_{\wh \C} \ltimes \left( \int_0^\infty \ell \mathrm{Weld}(\GQD(\ell), \GQD_{0,3}(\ell)) \, d\ell \right) 
	= K' C'^3 \LF_\C^{(\frac2\gamma, z_j)_3}(d\phi)  \, \prod_{j=1}^3 \Cont_\eta(d z_j)\,\SLE_{\kappa'}^\lp(d\eta).
\]
\end{proposition}
\begin{proof}
    The proof is identical to that of Lemma~\ref{prop-G3-identity-ns}, with Proposition~\ref{prop-tabula-rasa} replaced by Proposition~\ref{weldgqd} and Lemma~\ref{lem-benoist} replaced by Lemma~\ref{lem-benoist-ns}. 
\end{proof}

Following the argument of Proposition \ref{proposition-solve-G3}, we compute $G_3^\lp$ by disintegrating over $z_1,z_2,z_3$ and taking the expectation of $Ae^{-A}$ where $A$ is the quantum area. 
\begin{proposition}\label{prop:G3-solve}
     Let $R_\gamma'$ be  the constant in Lemma \ref{lem:gqd-bdy-length-law} and $M'=2\left(\frac{1}{4\sin\frac{\pi\gamma^2}{4}}\right)^{\frac{\kappa'}{8}}$. Then  $$G_3^{\rm loop}\left(0,1, e^{i \pi / 3}\right)=\frac{R_\gamma'^2}{K'C'^3}M'^{\frac{\gamma^2}{2}-1}2^{-1-\frac{\gamma^2}{2}}\frac{\kappa'}{\Gamma(\frac{4}{\kappa'})^2}\frac{\frac{\pi^2}{4}(\frac{\gamma^2}{2}-1)}{\sin(\frac{\pi}{2}(\frac{\gamma^2}{2}-1))}\frac{1}{(\frac{1}{\gamma^2}-\frac{1}{2})C_\gamma^{\mathrm {DOZZ}}\left(\frac{2}{\gamma}, \frac{2}{\gamma}, \frac{2}{\gamma}\right)}.$$
\end{proposition}
\begin{proof}
Since $|\GQD_{0,3}(\ell)|=\ell^3|\GQD(\ell)|$, by Proposition \ref{prop:ns-area-law} and using Lemma \ref{lem-int-KK}, we can compute
\begin{equation*}
\left(\int_0^\infty \ell^4 \mathrm{Weld}(\GQD(\ell), \GQD(\ell))d\ell\right)[Ae^{-A}]=R_\gamma'^2M'^{\frac{\gamma^2}{2}-1}2^{-1-\frac{\gamma^2}{2}}\frac{\kappa'}{\Gamma(\frac{4}{\kappa'})^2}\frac{\frac{\pi}{2}(\frac{\gamma^2}{2}-1)}{\sin(\frac{\pi}{2}(\frac{\gamma^2}{2}-1))}\frac{\pi}{2}.
\end{equation*}
Let $(u_1, u_2, u_3) = (0,1,e^{i\pi/3})$. Disintegrating over $z_1,z_2,z_3$ as in the proof of Proposition \ref{proposition-solve-G3} gives
\begin{equation*}
\left(\int_0^\infty \ell^4 \mathrm{Weld}(\GQD(\ell), \GQD(\ell))d\ell\right)[Ae^{-A}]=K' C'^3 \LF_\C^{(\frac2\gamma, u_j)_3}[Ae^{-A}]G_3^{\lp}(u_1,u_2,u_3).
\end{equation*}
By Lemma \ref{lem-sph-area-law} we know $\LF_\C^{(\frac2\gamma, u_j)_3}[Ae^{-A}]=(\frac{1}{\gamma^2}-\frac{1}{2})C_\gamma^{\text {DOZZ }}\left(\frac{2}{\gamma}, \frac{2}{\gamma}, \frac{2}{\gamma}\right)$. The result follows.
\end{proof}

Next, we solve for $G_2^\lp$. As before, we start with a conformal welding result.
\begin{proposition}\label{prop: weld-gqd}
    If we embed a sample from $\int_{0}^\infty 2\ell\Weld(\GQD_1(\ell),\GQD_{0,2}(\ell))d\ell$ in $(\C,u_1,u_2,u_3)$, then its law is given by
    $$ K'C'^2 G_2^{\rm loop}(u_1,u_2)\LF_{\C}^{(\frac{2}{\gamma},u_j)_2,(\gamma,u_3)}(d\phi)\sm^{u_1,u_2}(d\eta)  $$
    where $\sm^{u_1,u_2}$ is a probability measure on the space of loops passing through $u_1$ and $u_2$.
\end{proposition}
\begin{proof}
    The proof is identical to that of Proposition~\ref{prop-G2-identity}, with Proposition~\ref{prop-tabula-rasa} replaced by Proposition~\ref{weldgqd} and Lemma~\ref{lem-benoist} replaced by Lemma~\ref{lem-benoist-ns}. 
\end{proof}

Fortunately, the insertions $(\frac2\gamma, \frac2\gamma, \gamma)$ satisfy the extended Seiberg bounds~\eqref{eq-ext-Seiberg}, so to compute  $G_2^\lp$ we will not need to employ a reweighting argument as in Section~\ref{section:G2}. 

\begin{proposition}\label{prop:G2-solve}
\begin{equation}\label{eq:G2-solve}
        G_2^{\rm loop }\left(0,1 \right)=\frac{{R_\gamma'}^2}{K'C'^2}\frac{4\pi}{\gamma^3\sin(\frac{\pi\gamma^2}{4})^2\Gamma(\frac{\gamma^2}{4})^2}.
    \end{equation}
\end{proposition}
\begin{proof}
The result also follows by choosing the observable $Ae^{-A}$ in Proposition~\ref{prop: weld-gqd}, where $A$ is the quantum area. Since $|\GQD_{0,2}(\ell)|=\ell^2|\GQD(\ell)|$, by Proposition \ref{prop:ns-area-law} and Lemma \ref{lem-int-KK} we have
    \begin{align*}
        \Big(\int_{0}^\infty 2\ell^3\Weld(\GQD_{1,0}(\ell),\GQD(\ell))d\ell [Ae^{-A}]\Big)= \frac{\kappa'^2}{64\sin\frac{\pi\gamma^2}{4}}\frac{\Gamma(1-\frac{4}{\kappa'})}{\Gamma(\frac{4}{\kappa'})}{R_\gamma'}^2.
    \end{align*}
We also need to evaluate $\LF_{\C}^{\left(\frac{2}{\gamma}, u_1\right),\left(\frac{2}{\gamma}, u_2\right),\left(\gamma, u_3\right)}\left[A e^{-A}\right]$. Note that for $\alpha\in (0,\frac{4}{\gamma})\backslash\{\gamma\}$,
\begin{equation*}
\LF_{\C}^{\left(\frac{2}{\gamma}, u_1\right),\left(\frac{2}{\gamma}, u_2\right),\left(\alpha, u_3\right)}\left[A e^{-A}\right]=\frac{\alpha-\gamma}{2\gamma}C_\gamma^{\text {DOZZ }}\left(\frac{2}{\gamma}, \frac{2}{\gamma}, \alpha\right)
\end{equation*}
Then by continuity, we have
\begin{equation*}
\LF_{\C}^{\left(\frac{2}{\gamma}, u_1\right),\left(\frac{2}{\gamma}, u_2\right),\left(\gamma, u_3\right)}\left[A e^{-A}\right]=\lim _{\alpha \uparrow \gamma}\frac{\alpha-\gamma}{2\gamma}C_\gamma^{\rm DOZZ }\left(\frac{2}{\gamma}, \frac{2}{\gamma}, \alpha\right)=\frac{1}{\gamma},
\end{equation*}
where we use the explicit form of $C_\gamma^{\text {DOZZ }}$ \eqref{eq:DOZZ} and the fact $\Upsilon_{\frac{\gamma}{2}}(z)=\Upsilon_{\frac{\gamma}{2}}(Q-z)$ in the second step. Combining the above two expressions the result follows. 
\end{proof}

\begin{proof}[Proof of Theorem~\ref{thm:mag}, the non-simple case]
Combing Proposition \ref{prop:G3-solve} and \ref{prop:G2-solve} with the expression of $\frac{K'}{R_\gamma'^2}$ in Proposition \ref{weldgqd} yields the result as in the simple case.
\end{proof}

\begin{remark}\label{cor:mix3point}
    As an intermediate step in the computation of $G_2^{\rm loop}$, Proposition~\ref{lem-G2-alpha} and its generalization in the non-simple cases imply the following formula~\eqref{eq:mix}.  For $\kappa\in (\frac{8}{3},8)$, $b=2/\sqrt{\kappa}$ and $\alpha=Q-2P\in (Q-\frac{\sqrt{\kappa}}{4},Q)$ 
\begin{equation}\label{eq:mix}
\frac{2^{2\Delta_{\alpha}-2}|\sm^{0,1,e^{\frac{\pi i}{3}}}_\alpha|\E[\CR(\eta_0,0)^{2\Delta_{\alpha}-2}]}{\sqrt{C^{\rm CLE}_\kappa(2\Delta_\alpha-2,2\Delta_\alpha-2,0)}}=\omega^{(b)}_{(1,0),(1,0),(0,2\beta P)}
\end{equation}
where $\eta_0$ is the outmost $\CLE_\kappa$ loop in unit disk surrounding the origin and $C^{\rm CLE}_\kappa$ is the joint moment of conformal radius, see Theorem~\ref{thm:conformal radius} for its definition. This formula can be interpreted as the correlation function of two magnetic fields and one electric field mentioned in Section~\ref{section:loop-connectivity}.
\end{remark}

\section{Green function for the CLE cluster}\label{sec:carpet green}

In this section we prove Theorem \ref{thm: connectivity} using the coupling of CLE and LQG. The high level strategy is similar to the proof of Theorem~\ref{thm:mag} in Sections \ref{sec-loop-green-simple} and \ref{sec-loop-green-nonsimple}. The key step is to find a  proper setting where the Green functions for the CLE cluster naturally arise and are solvable. This is done in Proposition~\ref{prop:YA} in Section~\ref{sec:reduction}, which expresses $G_3^{\rm cluster}$ and $ G_2^{\rm cluster}$ in terms of the DOZZ formula and the joint moments of the total quantum area of the quantum disk and the total quantum measure of the CLE cluster. 
We provide background on the CLE/LQG coupling in Section \ref{couple} and recall the relation between the natural measure on the CLE cluster and its quantum analog in Section~\ref{sec:natural-measures}.  In Section~\ref{sec:solving-connectivity}, we finish the proof of Theorem \ref{thm: connectivity} based on Proposition~\ref{prop:YA}, assuming the formulae for the aforementioned joint moments
whose proofs are postponed to Section~\ref{levy}.

\subsection{The independent coupling of CLE and LQG}\label{couple}\label{subsec:MSW}
\newcommand{\outleng}{a}

We begin with a recent result on CLE$_\kappa$ coupled with $\gamma$-LQG disks where $\kappa\in (\frac83,4)$ and $\gamma=\sqrt{\kappa}$. For $a>0$,
suppose $(D,h)$ is an embedding of a sample from $\QD(a)^\#$. Let $\Gamma$ be  a $\CLE_\kappa$ on $D$ which is  independent of $h$.
Then we call  the decorated quantum surface	$(D,h,\Gamma)/{\sim_\gamma}$ a $\CLE_\kappa$-decorated quantum disk and 
denote its law  by $\QD(a)^\#\otimes \CLE_\kappa$.   By the conformal invariance of $\CLE_\kappa$,  the measure $\QD(a)^\#\otimes \CLE_\kappa$
does not depend on the choice of the embedding of  $(D,h)/{\sim_\gamma}$.

Let $(D,h,\Gamma)$ be an embedding of a sample from $\QD(\outleng)^{\#}\otimes \CLE_\kappa$ where $D$ is a bounded domain. We say a loop $\eta$ \emph{surrounds} a point $z$ if $z \not \in \eta$ and $\eta$ has a nonzero winding number with respect to $z$, and we call a loop $\eta \in \Gamma$  \emph{outermost} if no point of $\eta$ is surrounded by any loop in $\Gamma$.
Let  $(\ell_i)_{i\ge 1}$ be the collection of the quantum lengths of the outermost loops of $\Gamma$ listed in decreasing order. 
The crucial inputs to our proofs is the  law of  $(\ell_i)_{i\ge 1}$: 

\begin{proposition}[{\cite{msw-cle-lqg,bbck-growth-frag,ccm-perimeter-cascade}}]\label{prop-ccm}
For $\beta\in (1,2)$, let $(\zeta_t)_{t\geq 0}$ be a $\lexp$-stable L\'evy process  
	whose L\'evy measure is $1_{x>0} x^{-\beta-1} \, dx$, so $\zeta$ has no downward jumps. We denote its law by $\P^\beta$.
	Let $\tau_{-\outleng}=\inf\{ t: \zeta_t=-\outleng  \}$. Let  $ (x_i)_{i \geq 1}$ be the sequence of the sizes of the upward jumps of $\zeta$ on $[0,\tau_{-\outleng}]$ sorted in decreasing order.  Then for $\kappa\in (\frac{8}{3},4)$ and  $\lexp= \frac4{\kappa} + \frac12\in (\frac32,2)$, the law of $(\ell_i)_{i \geq 1}$ defined right above equals that of $ (x_i)_{i \geq 1}$ under the weighted probability measure $\frac{\tau_{-\outleng}^{-1}\P^\beta}{\E[\tau_{-\outleng}^{-1}]}$.
\end{proposition}
\begin{proof}
    This was first stated at the end of \cite[Section 1]{ccm-perimeter-cascade} as a consequence of \cite{msw-cle-lqg,bbck-growth-frag,ccm-perimeter-cascade}. See \cite[Proposition 4.1]{ACSW24a} for further proof details.
\end{proof}

Now we review the independent coupling of CLE$_{\kappa'}$ and $\gamma$-generalized quantum disks, where $\kappa'\in (4,8)$ and $\gamma=4/\sqrt{\kappa'}$. One can similarly define $\GQD(a)^\#\otimes\CLE_{\kappa'}$ by sampling independent $\CLE_{\kappa'}$ in each connected component of a sample from $\GQD(a)^\#$.
Let $(D,h,\Gamma)$ be an embedding of a sample from $\GQD(\outleng)^{\#}\otimes \CLE_{\kappa'}$ where $D$ is bounded. 
Define as above the notion of outermost loop,
and let $(\ell_i)_{i\ge 1}$ be the collection of the generalized quantum lengths of the outermost loops of $\Gamma$ listed in decreasing order. We have the following analogue  of Propositions~\ref{prop-ccm}.

\begin{proposition}[{\cite{msw-non-simple,bbck-growth-frag,ccm-perimeter-cascade}}]\label{prop-ccm-ns}
 For $\kappa'\in (4,8)$ and $\beta' := \frac4{\kappa'} + \frac12\in (1,\frac32)$, the law of $(\ell_i)_{i \geq 1}$ defined right above agrees with that of $ (x_i)_{i \geq 1}$ under the weighted probability measure $\frac{\tau_{-\outleng}^{-1}\P^{\beta'}}{\E[\tau_{-\outleng}^{-1}]}$, where $\P^{\beta'}$, $\tau_{-\outleng}$, and $ (x_i)_{i \geq 1}$ are as defined in Proposition~\ref{prop-ccm}. 
\end{proposition}
\begin{proof}
    Similarly to Proposition~\ref{prop-ccm}, we refer to \cite[Proposition 4.3]{ACSW24a} for proof details.
\end{proof}

\subsection{Natural measures on the CLE cluster}\label{sec:natural-measures}

For CLE on the disk, the outermost cluster is the set of points not surrounded by any loop. It is also called the CLE carpet (for $\kappa \in (8/3, 4]$) or gasket (for $\kappa \in (4,8)$).
We first recall the construction from~\cite{msw-cle-lqg} and \cite{msw-non-simple} of the \textit{quantum natural measures} on the outermost cluster of the CLE on a disk.
Fix $\kappa\in(\frac{8}{3},4)\cup(4,8)$ and let $\beta=\frac{4}{\kappa}+\frac{1}{2}$.   Let $(D,h,\Gamma)$ be an embedding of $\QD(1)^\#\otimes\CLE_\kappa$. For $\varepsilon>0$ and an open subset $U\subset D$, let $N_{\varepsilon,h,\Gamma}(U)$ be the number of outermost loops of $\Gamma$ contained in $U$ with quantum length for $\kappa\in(\frac{8}{3},4)$ (resp. generalized quantum length for $\kappa\in(4,8)$) no less than $\varepsilon$. Then by Theorem 1.3 in \cite{msw-cle-lqg} (resp. Proposition 5.3 in \cite{msw-non-simple}),
there is a measure $M_{h,\Gamma}$ (depending on $h$) supported on the outermost cluster of $\Gamma$, such that for any given open subset $U\subset D$, $M_{h,\Gamma}(U)=\lim_{\varepsilon\to0}\varepsilon^\beta N_{\varepsilon,h,\Gamma}(U)$. We call $M_{h,\Gamma}$ the {\it quantum natural measure} on the outermost cluster  of $\Gamma$.

The Miller-Schoug measure~\cite{miller2023existence}  on the outermost cluster of $\Gamma$ is defined via taking the GFF expectation over a modification of $M_{h,\Gamma}$.  Let $h_0$ be the Dirichlet GFF on $D$, and $M_{h_0,\Gamma}$ be the quantum natural measure on the outermost cluster with respect to  $h_0$; this is well-defined by local absolute continuity between the laws of $h_0$ and $h$. Let $r_D(z)$ be the conformal radius of $D$ seen from $z$, then
\begin{equation}\label{eq:dfms}
\mu_\Gamma(dz)=r_D^{-\frac{1}{2}+\frac{2}{\kappa}+\frac{\kappa}{32}}\E\left[M_{h_0,\Gamma}(dz)\big|\Gamma\right],\quad \kappa\in\left(\frac{8}{3},4\right)\cup(4,8)
\end{equation}
is called the {\it Miller-Schoug measure} on the outermost cluster of $\Gamma$. According to \cite{miller2023existence}, this measure is characterized by the axioms of conformal covariance, domain Markov property and finiteness (also see \cite{cai2022natural} for slightly different axioms). The Miller-Schoug measure can be extended to the case $\kappa=4$ by a limiting procedure; see Appendix~\ref{kappato4}. In this section we exclude this case. Although the Miller-Schoug measure should agree with the Minkowski content measure and the quantum natural measure should be a GMC over it, neither of these statements has been proved. However, the following lemma can serve as an effective alternative to Lemma~\ref{lem-add-insertion}, which
relates $M_{h,\Gamma}$ and $\mu_\Gamma$ by the Girsanov theorem.
\begin{lemma}[{\cite[Section 3, (1)]{cai2022natural}}]\label{resampling}
Fix $\kappa\in(\frac{8}{3},4)\cup(4,8)$ and let $\alpha=Q-\frac{\sqrt{\kappa}}{4}$. Suppose $h$ is locally mutually absolutely continuous with respect to the GFF (possibly with finite insertion points); denote its law by ${\rm Law}(dh)$. Then for $M_{h,\Gamma}$ and $\mu_\Gamma$ defined above, we have
\begin{equation*}
M_{h,\Gamma}(dz){\rm Law}(dh)\CLE_\kappa(d\Gamma)=e^{\alpha h(z)}{\rm Law}(dh)\mu_\Gamma(dz)\CLE_\kappa(d\Gamma)
\end{equation*}
where $e^{\alpha h(z)}{\rm Law}(dh) := \lim_{\varepsilon\to0}\varepsilon^{\frac{\alpha^2}{2}}e^{\alpha h_\varepsilon(z)}{\rm Law}(dh)$ in the topology of vague convergence of measures. In particular, $e^{\alpha h(z)}\LF_\C^{(\alpha_i,z_i)_i}(dh)=\LF_\C^{(\alpha_i,z_i)_i,(\alpha,z)}$ (see the paragraph below Definition \ref{def-RV-sph}).
\end{lemma}

The definition of the Miller-Schoug measure straightforwardly extends to the full-plane case.  Let $\Gamma$ be the full-plane $\CLE_\kappa$. For each loop $\eta$,
let $K(\eta)$ be the cluster surrounded by $\eta$, i.e., the set of points surrounded by $\eta$ which are not surrounded by any loop surrounded by $\eta$.
Denote the Miller-Schoug measure on $K(\eta)$ by $\mu(\cdot;K(\eta))$. The Green function  $G^{\rm cluster}_n$ from~\eqref{eq:greendef} can be written as 
\begin{equation}\label{eq:G-carpet}
G^{\rm cluster}_n(z_1,z_2,...,z_n)\prod_{k=1}^ndz_k=\int \prod_{k=1}^n\mu(dz_k;K(\eta)){\rm Count}_\Gamma(d\eta){\rm CLE}_\kappa^\C(d\Gamma)
\end{equation}
where ${\rm Count}_\Gamma(d\eta)$ is the counting measure on  $\Gamma$, and the integral is taken over $\eta$ and $\Gamma$. We now explain the conformal covariance of the Green function from Lemma~\ref{covariance_cluster}. 
\begin{proof}[Proof of Lemma~\ref{covariance_cluster}]
Let $f$ be a conformal automorphism on $\wh\C$. View both sides of \eqref{eq:G-carpet} as measures on $\C^n$, and consider the pushforwards of them under $f^{-1}$. The image measure of the left side is $G^{\rm cluster}_n(f(z_1),...,f(z_n))\prod_{k=1}^n |f'(z_k)|^{-2} dz_k$. By the conformal invariance of full-plane CLE and the conformal covariance of $\mu$ (see e.g. \cite[Definition 1.1,(i)]{miller2023existence}), the image of the right side is $\int \prod_{k=1}^n |f'(z_k)|^{-d}\mu(dz_k;K(\eta)){\rm Count}_\Gamma(d\eta){\rm CLE}_\kappa^\C(d\Gamma)$, where $d=2-\frac{(3\kappa-8)(8-\kappa)}{32\kappa}$ is the Hausdorff dimension of $\CLE_\kappa$ cluster. Then $G^{\rm cluster}_n(f(z_1),...,f(z_n))=\prod_{i=1}^n |f'(z_i)|^{2-d}G^{\rm cluster}_n(z_1,...,z_n)$ for any $f\in\conf(\wh\C)$, and the lemma follows.
\end{proof}

In our definition of $K(\eta)$, the point $\infty$ plays a special role, since a loop surrounds a point if and only if it separates the point from $\infty$. In our subsequent argument, we will replace $\infty$ by a generic point $z \in \C$. To that end, let $\wh K(\eta) := K(\eta')$ where $\eta'$ is the innermost CLE loop surrounding $\eta$, so $K(\eta)$ and $\wh K(\eta)$ are the two clusters adjacent to $\eta$. Let $K(\eta; z) = K(\eta)$ if $\eta$ does not surround $z$, and let $K(\eta ; z) = \wh K(\eta)$ otherwise;
 see Figure~\ref{fig:component} (left). In words, $K(\eta; z)$ is the cluster adjacent to $\eta$ such that $\eta$ separates $K(\eta; z)$ and $z$. With this definition, we have the following variant of~\eqref{eq:G-carpet}.

\begin{figure}[htb]
\centering
\includegraphics[width=0.4\linewidth]{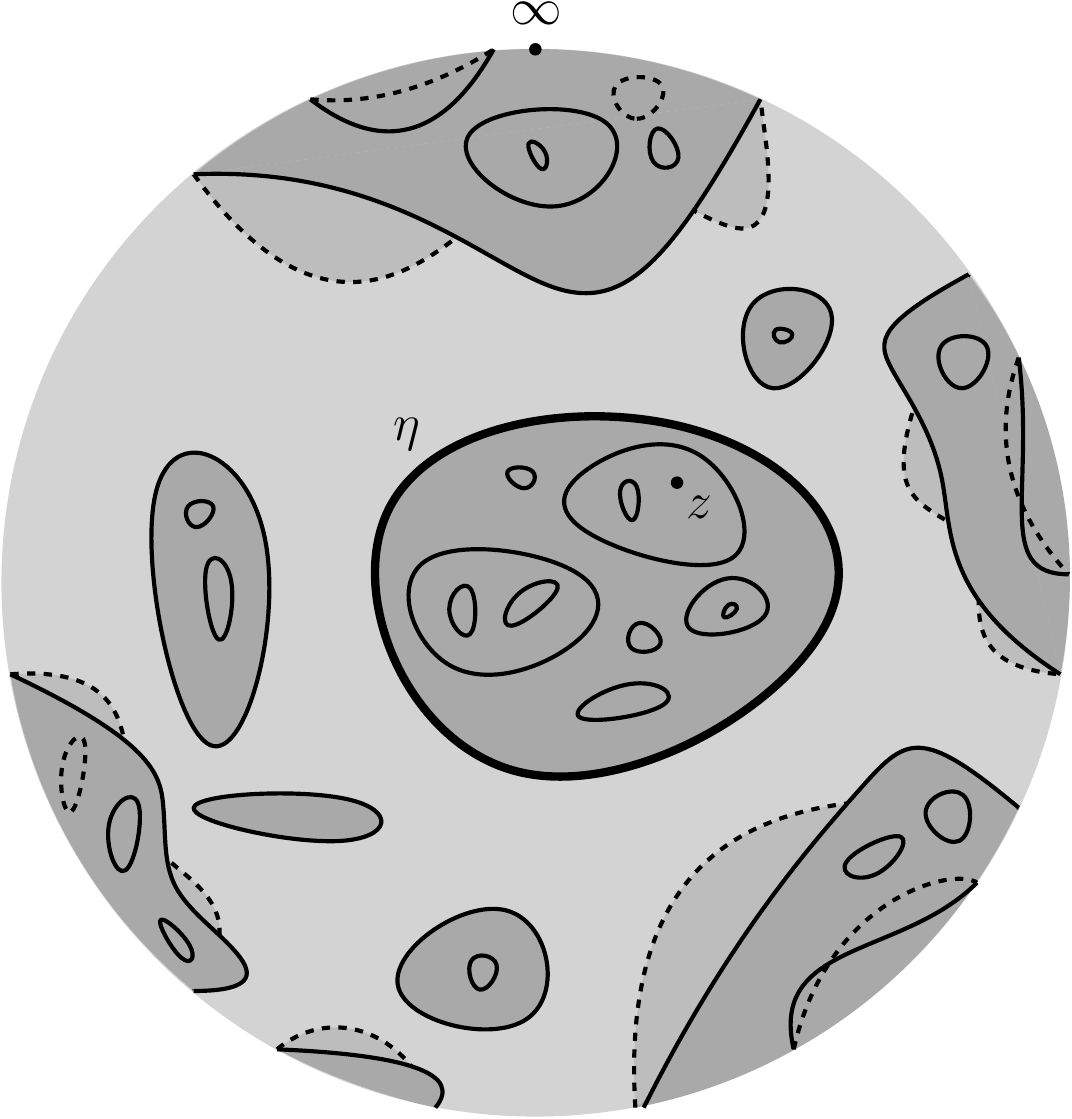}
\hspace{.6in}
\includegraphics[width=0.3\linewidth]{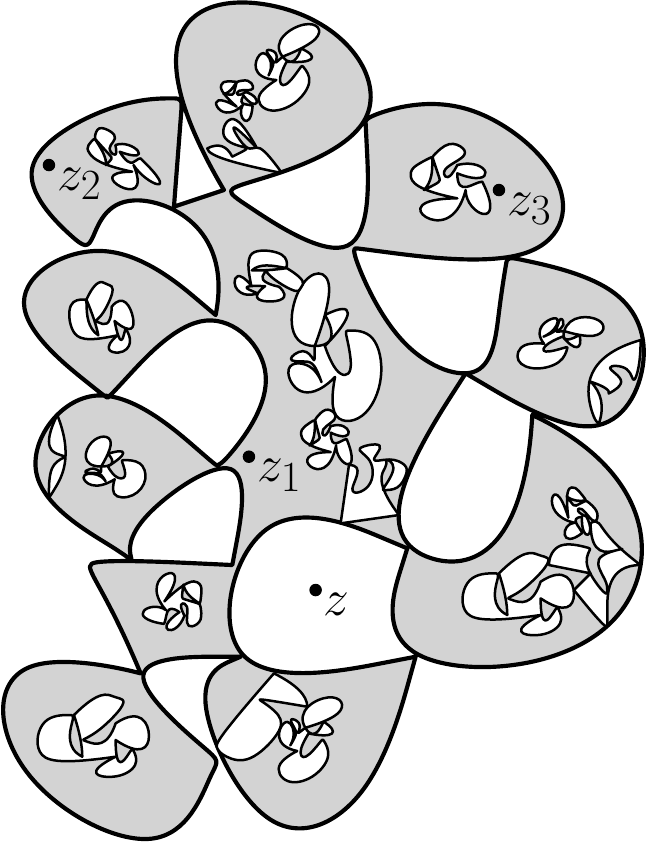}
\caption{\textbf{Left}: Illustration of the case where $\kappa\in(\frac{8}{3},4)$, and $\eta$ surrounds $z$. The set  $K(\eta,z)=\wh K(\eta)$ is colored in light grey. \textbf{Right}: Three points $z_1,z_2,z_3$ in the cluster $K(\eta,z)$ as in the proof of Proposition~\ref{prop:G-car-gsk}. Here $\kappa\in(4,8)$.}
\label{fig:component}
\end{figure}

\begin{proposition}\label{prop:G-car-gsk}
For $\kappa\in(\frac{8}{3},8)$, we have
\begin{equation}\label{eq:gf-0}
G^{\rm cluster}_n(z_1,z_2,...,z_n)\prod_{k=1}^ndz_k\cdot dz=\int\prod_{k=1}^n\mu(dz_k; K(\eta,z)){\rm Count}_\Gamma(d\eta){\rm CLE}_\kappa^\C(d\Gamma)dz
\end{equation}
where the integral on the right hand side is taken over ${\rm Count}_\Gamma(d\eta){\rm CLE}_\kappa^\C(d\Gamma)$.
\end{proposition}
\begin{proof}
Up to a rearrangement, $\{\mu(dz;K(\eta,z))|\eta\in\Gamma\}$ and $\{\mu(dz;K(\eta))|\eta\in\Gamma\}$ give the same collection of measures on $\C$. 
Therefore they share the same Green function. Namely, $$\int\prod_{k=1}^n\mu(dz_k;K(\eta,z)){\rm Count}_\Gamma(d\eta){\rm CLE}_\kappa^\C(d\Gamma)=G^{\rm cluster}_n(z_1,z_2,...,z_n)\prod_{k=1}^ndz_k \textrm{ for all }z\in \C.$$Integrating over $z\in \C$ we conclude.
\end{proof}

\subsection{Reduction to solvable observables of the CLE-decorated quantum disk}\label{sec:reduction}
In this section, we  express the Green functions $G^{\rm cluster}_3$ and $G^{\rm cluster}_2$ in terms of 
 observables of the CLE-decorated quantum disk which can be solved by analyzing a L\'evy process. 
As the following proposition shows, the key quantities to compute are
\begin{equation*}
\QD(\ell)^\#\otimes\CLE_\kappa[Y^ke^{-\mu A}]\ \text{and}\ \GQD(\ell)^\#\otimes\CLE_{\kappa'}[Y^ke^{-\mu A}] \quad \textrm{for } k=2,3,
\end{equation*}
 where $A$ is the total quantum area of the quantum disk, and $Y$ is the total mass of the quantum natural measure of the CLE cluster on the quantum disk.
\begin{proposition}\label{prop:YA}
Fix $(u_1,u_2,u_3)=(0,1,e^{i\pi/3})$ and $\mu>0$. Let $K = K(\gamma)$ be the constant from  Proposition~\ref{prop-tabula-rasa}. 
For  $\kappa\in (\frac{8}{3}, 4)$,  we have
\begin{equation}\label{weld3}
\begin{aligned}
K G^{\rm cluster}_3\!(u_1,\!u_2,\!u_3)\partial_\mu   \LF_\C^{(Q-\frac{\gamma}{4},u_i)_3}[e^{-\mu A}]=\!-\!\int_0^\infty\!\! \!2\ell |\QD(\ell)|^2 \QD(\ell)^\#\!\otimes\!\CLE_\kappa[Y^3 e^{-\mu A}]\QD(\ell)^\#[Ae^{-\mu A}]d\ell,
\end{aligned}
\end{equation}
\begin{equation}\label{weld2}
\begin{aligned}
K G^{\rm cluster}_2\!(u_1,\!u_2) 
\LF_\C^{(Q-\frac{\gamma}{4},u_i)_2,(\gamma,u_3)}[e^{-\mu A}]=\!\int_0^\infty \!\!\!2\ell |\QD(\ell)|^2 \QD(\ell)^\#\!\otimes\!\CLE_\kappa[Y^2 e^{-\mu A}]\QD(\ell)^\#[Ae^{-\mu A}]d\ell.
\end{aligned}
\end{equation}
Similarly, let $K' = K'(\gamma)$ be the constant from Proposition~\ref{weldgqd}. For  $\kappa'\in (4,8)$, we have
\begin{equation}\label{weld3dense}
\begin{aligned}
K' G^{\rm cluster}_3\!(u_1,\!u_2,\!u_3)\partial_\mu\LF_\C^{(Q-\frac{1}\gamma,u_i)_3}[e^{-\mu A}]\!=\!-\!\!\int_0^\infty \!\!\!\!2\ell |\GQD(\ell)|^2 \GQD^\#\!(\ell)\!\otimes\!\CLE_{\kappa'}[Y^3 e^{-\mu A}]\GQD^\#\!(\ell)[Ae^{-\mu A}]d\ell,
\end{aligned}
\end{equation}
\begin{equation}\label{weld2dense}
\begin{aligned}
K' G^{\rm cluster}_2\!(u_1,\!u_2) \LF_\C^{(Q-\frac{1}{\gamma},u_i)_2,(\gamma,u_3)}\![e^{-\mu A}]\!=\!\!\int_0^\infty \!\!\!\!2\ell |\GQD(\ell)|^2 \GQD^\#\!(\ell)\!\otimes\!\CLE_{\kappa'}[Y^2 e^{-\mu A}]\GQD^\#\!(\ell)[Ae^{-\mu A}]d\ell.
\end{aligned}
\end{equation}
\end{proposition} 
\begin{proof}
We focus on the case $\kappa\in(\frac{8}{3},4)$ first; the  $\kappa'\in (4,8)$ case is similar.
By Proposition \ref{prop-tabula-rasa}, we have
$\mathbf m_{\wh \C} \ltimes \left( \int_0^\infty \ell \mathrm{Weld}(\QD(\ell), \QD(\ell)) \, d\ell \right) = K \LF_\C \times \SLE_\kappa^\lp$.
Sampling one point from the quantum area measure on each side (and thus inducing a weighting by total quantum area), we get 
\begin{equation*}
\mathbf m_{\wh \C}\ltimes\left(\int_0^\infty 2\ell \Weld(\QD(\ell),\QD_1(\ell))d\ell \right)=K\LF_\C^{(\gamma,z)} \SLE_\kappa^{\text{loop}}(d\eta)\, dz
\end{equation*}
where the above factor 2 on the left side arises  from choosing which disk receives the marked bulk point. Let $D(\eta,z)$ be the connected component of $\hat \C \backslash \eta$ not containing $z$.
Sampling an $\CLE_\kappa$ configuration on the unmarked disk corresponding to $D(\eta,z)$ gives 
\begin{equation}\label{eq:loop}
\mathbf m_{\wh \C}\ltimes \left(\int_0^\infty 2\ell \Weld(\QD(\ell)\otimes\CLE_\kappa(d\Gamma),\QD_1(\ell))d\ell \right)=K\LF_\C^{(\gamma,z)}\CLE_\kappa^{D(\eta,z)}(d\Gamma) \SLE_\kappa^{\text{loop}}(d\eta)dz.
\end{equation}
Here $\CLE_\kappa^{D(\eta,z)}(d\Gamma)$ denotes the law of CLE$_\kappa$ in $D(\eta,z)$.  On both sides of~\eqref{eq:loop}, we weight the $k$-th power of total mass of the quantum natural measure 
on the CLE carpet of $\Gamma$ and sample $k$ points on $\Gamma$ according to the probability measure proportional to the quantum natural measure on the CLE cluster. Recall the Miller-Schoug measure $\mu_\Gamma(dz)$ defined in \eqref{eq:dfms}. By applying Lemma~\ref{resampling} $k$ times, the right hand side of~\eqref{eq:loop} becomes
\begin{equation}\label{eq:right hand side}
  K\LF_\C^{(Q-\frac{\gamma}{4},z_i)_k,(\gamma,z)}(d\phi)\prod_{i=1}^k\mu_\Gamma(dz_i)\CLE_\kappa^{D(\eta,z)}(d\Gamma) \SLE_\kappa^{\mathrm{loop}}(d\eta)dz.  
\end{equation}
Since $\SLE_\kappa^{\mathrm{loop}}$ is chosen from ${\rm Count}_\Gamma(d\eta){\rm CLE}_\kappa^\C(d\Gamma)$, by Proposition \ref{prop:G-car-gsk}, 
\[
\left(\int\prod_{i=1}^k\mu_\Gamma(dz_i)\CLE_\kappa^{D(\eta,z)}(d\Gamma) \SLE_\kappa^{\mathrm{loop}}(d\eta)\right)dz=
G_3^{\rm cluster}(z_1,z_2,z_3)\prod_{i=1}^k dz_i dz.
\]
Therefore integrating $e^{-\mu A}$ against \eqref{eq:right hand side} with respect to $(\Gamma, \eta, \phi)$ gives
\begin{equation}\label{eq:car-right hand side}
\begin{aligned}
&K\LF_\C^{(Q-\frac{\gamma}{4},z_i)_3,(\gamma,z)}[e^{-\mu A}]\left(\int\prod_{i=1}^k\mu_\Gamma(dz_i)\CLE_\kappa^{D(\eta,z)}(d\Gamma) \SLE_\kappa^{\mathrm{loop}}(d\eta)\right)dz\\
&=K\LF_\C^{(Q-\frac{\gamma}{4},z_i)_k,(\gamma,z)}[e^{-\mu A}]G_3^{\rm cluster}(z_1,z_2,z_3)\prod_{i=1}^k dz_i dz.
\end{aligned}
\end{equation}
 By Lemma~\ref{lem-add-gamma} we have $\int \LF_\C^{(Q - \frac\gamma4,z_i)_3, (\gamma, z)} [e^{-\mu A}] dz =  \LF_\C^{(Q - \frac\gamma4,z_i)_3} [A e^{-\mu A}] = -\partial_\mu   \LF_\C^{(Q-\frac{\gamma}{4},z_i)_3}[e^{-\mu A}]$. Setting $k=3$  and integrating~\eqref{eq:car-right hand side} over $z$ gives
\begin{equation} \label{eq:rhs}
K G^{\rm cluster}_3(z_1,z_2,z_3)\left(\int \LF_\C^{(Q-\frac{\gamma}{4},z_i)_3,(\gamma,z)}[e^{-\mu A}] dz\right)\prod_{i=1}^3 dz_i=-K G^{\rm cluster}_3(z_1,z_2,z_3)\partial_\mu   \LF_\C^{(Q-\frac{\gamma}{4},z_i)_3}[e^{-\mu A}]\prod_{i=1}^3 dz_i. 
\end{equation}

It remains to deal with the left hand side of~\eqref{eq:loop}. For $k=3$, let $(\widehat \C, \Gamma_0, \eta_0, h_0,\mathbf p_1,\mathbf p_2,\mathbf p_3,\mathbf p)$ be an embedding of the quantum  surface sampled from  $\left(\int_0^\infty 2\ell \Weld(\QD(\ell)\otimes\CLE_\kappa(d\Gamma),\QD_1(\ell))d\ell \right)$ after weighting by the third power of the total quantum natural measure of the outermost cluster  of $\Gamma$ and adding three points $\mathbf p_1,\mathbf p_2,\mathbf p_3$ according to the quantum natural measure on  this cluster, and $\mathbf p$ corresponds  to the marked point from $\QD_1(\ell)$. There are infinitely many possible embeddings of the quantum surface, and one can uniquely specify an embedding by fixing the locations of $\mathbf p_1, \mathbf p_2, \mathbf p_3$. For three distinct points $z_1,z_2,z_3$ on $\widehat \C$,  let $M_{z_1,z_2,z_3}(dh,d\eta,d\Gamma,dz)$ be the law of $(\Gamma_0, \eta_0, h_0,\mathbf p)$ when the embedding is fixed by $(\mathbf p_1,\mathbf p_2,\mathbf p_3)=(z_1,z_2,z_3)$. Then 
 \begin{equation*} 
M_{z_1,z_2,z_3}[e^{-\mu A}]= \int_0^\infty \!\!2\ell \! \Weld(\QD(\ell)\!\otimes\!\CLE_\kappa(d\Gamma),\!\QD_1(\ell))[Y^3 e^{-\mu A}]d\ell.
\end{equation*}where $Y$ is the total mass of the quantum natural measure on the CLE cluster.
Let $(\Gamma, \eta, h,\mathbf z_1,\mathbf z_2,\mathbf z_3,\mathbf z)$ be the image of $(\Gamma_0, \eta_0, h_0,\mathbf p_1,\mathbf p_2,\mathbf p_3,\mathbf p)$ under the uniform embedding from $\mathbf{m}_{\wh \C}$.  By Lemma~\ref{Haar-density} and Theorem~\ref{thm:embed-qs}, the law of this tuple is
$M_{z_1,z_2,z_3}(dh,d\eta,d\Gamma,dz)  |(z_1-z_2)(z_2-z_3)(z_3-z_1)|^2  dz_1dz_2dz_3$.
Therefore, integrating $e^{-\mu A}$ with respect to $(\Gamma, \eta, \phi,z)$, we get 
\begin{equation}\label{eq:car-left hand side}
\int_0^\infty \!\!2\ell \! \Weld(\QD(\ell)\!\otimes\!\CLE_\kappa,\!\QD_1(\ell))[Y^3 e^{-\mu A}]d\ell\cdot  |(z_1\!-\!z_2)(z_2\!-\!z_3)(z_3\!-\!z_1)|^2  dz_1dz_2dz_3,
\end{equation}
Comparing  \eqref{eq:rhs} and \eqref{eq:car-left hand side}, and restricting
$(z_1,z_2,z_3)$ to $(u_1,u_2,u_3)=(0,1,e^{i\pi/3})$ by disintegration, we get~\eqref{weld3}.

For the case $k=2$, evaluating $e^{-\mu A}$ on the left hand side of \eqref{eq:loop} after the reweighting gives
\begin{equation*}
\int_0^\infty 2\ell  \Weld(\QD(\ell)\otimes\CLE_\kappa,\QD_1(\ell))[Y^2 e^{-\mu A}]d\ell\cdot  |(z_1-z_2)(z_2-z)(z-z_1)|^2  dz_1dz_2dz.
\end{equation*}
Comparing this with \eqref{eq:car-right hand side} and choosing $(z_1,z_2,z)=(u_1,u_2,u_3)=(0,1,e^{i\pi/3})$ gives \eqref{weld2}.

For $\kappa'\in (4,8)$, from Proposition \ref{weldgqd},  the equation~\eqref{eq:loop} holds with $K'$ in place of $K$ and $\GQD$ in place of $\QD$.
The rest of the argument is the same as in the $\kappa\in (\frac{8}{3},4)$ case. 
\end{proof}

\subsection{Proof of Theorems \ref{thm: connectivity}: the final calculation}\label{sec:solving-connectivity}
In this section, we prove  Theorems \ref{thm: connectivity}  based on Proposition~\ref{prop:YA}. 
We need the following special case of the DOZZ formula to 
compute $\LF_\C^{(Q-\frac{\gamma}{4},u_i)_2,(\gamma,u_3)}[e^{-\mu A}]$ and $\LF_\C^{(Q-\frac{\gamma}{4},u_i)_3}[e^{-\mu A}]$:
\begin{lemma}\label{lem:dozz}
     For $2\alpha+\gamma>2Q$,$$C_\gamma^{\rm DOZZ}\left(\alpha,\alpha,\gamma\right)=\frac{2}{\pi\gamma}\cdot(\pi \frac{\Gamma(\frac{\gamma^2}{4})}{\Gamma(1-\frac{\gamma^2}{4})})^{\frac{2Q-2\alpha}{\gamma}}\frac{\Gamma(\frac{2}{\gamma}(\alpha-\frac{2}{\gamma}))}{\Gamma(1-\frac{2}{\gamma}(\alpha-\frac{2}{\gamma}))}\frac{\Gamma(\frac{\gamma}{2}(\alpha-\frac{\gamma}{2}))}{\Gamma(1-\frac{\gamma}{2}(\alpha-\frac{\gamma}{2}))}$$
\end{lemma}
\begin{proof}
    By definition, 
    \begin{align*}
C^\mathrm{DOZZ}_\gamma(\alpha, \alpha, \gamma) = \left( \pi(\frac\gamma2)^{2-\gamma^2/2} \ell(\frac{\gamma^2}{4}) \right)^{\frac{2Q-2\alpha-\gamma}{\gamma}}\frac{\Ug'(0) \Ug(\alpha)\Ug(\alpha)\Ug(\gamma)}{\Ug(\alpha - \frac{2}{\gamma}) \Ug(\frac{\gamma}{2})\Ug(\frac{\gamma}{2})\Ug(\alpha-\frac{\gamma}{2})}.
\end{align*}
To simplify, we use the shift equations for $\Upsilon_{\frac\gamma2}$ which follow from~\eqref{eq:shift double gamma}  and $\Ug(z)=\frac{1}{\Gamma_{\frac{\gamma}{2}}(z)\Gamma_{\frac{\gamma}{2}}(Q-z)}$:

\begin{equation}\label{shiftUpsilon}
    \Upsilon_{\frac\gamma2}(z + \frac\gamma2) = \frac{\Gamma(\frac\gamma2 z)}{\Gamma(1-\frac\gamma2z)} (\frac\gamma2)^{1-\gamma z} \Upsilon_{\frac\gamma2}(z) \quad\text{and} \quad \Upsilon_{\frac\gamma2}(z + \frac2\gamma) = \frac{\Gamma(\frac2\gamma z)}{\Gamma(1 - \frac2\gamma z)} (\frac\gamma2)^{\frac4\gamma z - 1} \Upsilon_{\frac\gamma2}(z).
\end{equation}
Thus,
$\Ug(\frac{\gamma}{2})=\Ug'(0)$, $\Ug(\gamma) = \frac{\Gamma(\frac{\gamma^2}{4})}{\Gamma(1-\frac{\gamma^2}{4})} (\frac\gamma2)^{1-\gamma^2/2} \Ug(\frac{\gamma}{2})$ and $\Upsilon_{\frac\gamma2}(\alpha) = \frac{\Gamma(\frac\gamma2 (\alpha-\frac{\gamma}{2}))}{\Gamma(1-\frac\gamma2(\alpha-\frac{\gamma}{2}))} (\frac\gamma2)^{1-\gamma (\alpha-\frac{\gamma}{2})} \Upsilon_{\frac\gamma2}(\alpha-\frac{\gamma}{2})$, $\Upsilon_{\frac\gamma2}(\alpha) = \frac{\Gamma(\frac2\gamma (\alpha-\frac{2}{\gamma}))}{\Gamma(1 - \frac2\gamma (\alpha-\frac{2}{\gamma}))} (\frac\gamma2)^{\frac4\gamma (\alpha-\frac{2}{\gamma}) - 1} \Upsilon_{\frac\gamma2}(\alpha-\frac{2}{\gamma})$. Combining gives the conclusion.
\end{proof}

The following two propositions, whose proofs are postponed  to Section \ref{levy}, give expressions for $\QD(\ell)^\#\otimes\CLE_\kappa[Y^ke^{-\mu A}]$ and $ \GQD(\ell)^\#\otimes\CLE_{\kappa'}[Y^ke^{-\mu A}]$ with $k=2,3$. 
\begin{prop}\label{levy1}
For $\kappa\in(\frac{8}{3},4)$, let $\beta=\frac{4}{\kappa}+\frac{1}{2}$. For $\mu>0$, let $M=\sqrt{\frac{\mu}{\sin\frac{\kappa\pi}{4}}}$ and $L=-\frac{2^{1-\frac{4}{\kappa}}M^{\frac{4}{\kappa}+\frac{1}{2}}}{\Gamma(\frac{4}{\kappa})}\sqrt{2\pi}\frac{\pi}{\cos\frac{4\pi}{\kappa}}$. Then
\begin{align*}
\QD(\ell)^\#\otimes\CLE_{\kappa}[Y^3e^{-\mu A}]&=\frac{\sin(-\pi\beta)}{\beta^3\pi}e^{-M\ell}\left[\ell^{2+\beta}\frac{4M^2}{L^2}+\ell^{1+\beta}\frac{M}{2L^2}\left(\frac{64}{\kappa^2}-1\right)\right]; \\
\QD(\ell)^\#\otimes\CLE_{\kappa}[Y^2e^{-\mu A}]&=\frac{\sin(-\pi\beta)}{\beta^2\pi}\ell^{1+\beta}e^{-M\ell}\frac{2M}{L}.
\end{align*}
\end{prop}
\begin{prop}\label{levy2}
For $\kappa'\in(4,8)$, let $\beta'=\frac{4}{\kappa'}+\frac{1}{2}$. For $\mu>0$, let  $M'=2\left(\frac{\mu}{4\sin\frac{4\pi}{\kappa'}}\right)^{\frac{\kappa'}{8}}$ and $L'=-\frac{2^{1-\frac{4}{\kappa'}}M'^{\frac{4}{\kappa'}+\frac{1}{2}}}{\Gamma(\frac{4}{\kappa'})}\sqrt{2\pi}\frac{\pi}{\cos\frac{4\pi}{\kappa'}}$. Then
\begin{align*}
\GQD(\ell)^\#\otimes\CLE_{\kappa'}[Y^3e^{-\mu A}]&=\frac{\sin(-\pi\beta')}{\beta'^3\pi}e^{-M\ell}\left[\ell^{2+\beta'}\frac{4M'^2}{L'^2}+\ell^{1+\beta'}\frac{M'}{2L'^2}\left(\frac{64}{\kappa'^2}-1\right)\right];\\
\GQD(\ell)^\#\otimes\CLE_{\kappa'}[Y^2e^{-\mu A}]&=\frac{\sin(-\pi\beta')}{\beta'^2\pi}\ell^{1+\beta'}e^{-M'\ell}\frac{2M'}{L'}.
\end{align*}
\end{prop}

\begin{proof}[Proof of Theorem \ref{thm: connectivity} for $\kappa\in (\frac{8}{3},4{]}$ assuming Proposition~\ref{levy1}]
For $\kappa\in(\frac{8}{3},4)$, let $\kappa=\gamma^2$.
We first calculate $G^{\rm cluster}_3$ using~\eqref{weld3}. By Proposition~\ref{prop-DOZZ}, the left hand side of~\eqref{weld3} equals
\begin{equation*}
KC^3 G^{\rm cluster}_3(u_1,u_2,u_3) \frac{1}{2\gamma}(Q-\frac{3\gamma}{4})C_\gamma^{\rm DOZZ}\left(Q-\frac{\gamma}{4},Q-\frac{\gamma}{4},Q-\frac{\gamma}{4}\right)\mu^{-\frac{Q-3\gamma/4}{\gamma}-1}.
\end{equation*}
where $K$ is from Proposition~\ref{prop-tabula-rasa}. 
By Proposition \ref{levy1}, the right hand side of~\eqref{weld3} equals
\begin{align*}
&\int_0^\infty 2\ell |\mathrm{QD}(\ell)|^2 \mathrm{QD}^{\#}(\ell)\left[Ae^{-\mu A}\right]\cdot \frac{\sin(-\pi\beta)}{\beta^3\pi}e^{-M\ell}\left[\ell^{2+\beta}\frac{4M^2}{L^2}+\ell^{1+\beta}\frac{M}{2\ell^2}\left(\frac{64}{\kappa^2}-1\right)\right] d\ell
\\&=R_\gamma^2 2^{4/\gamma^2}\Gamma\left(\frac{4}{\gamma^2}\right)\frac{1}{\beta^3\pi^3}\left(\sin\frac{\pi\gamma^2}{4}\right)^{2/\gamma^2-1/4}\cos^2\frac{4\pi}{\gamma^2}\times \sqrt{\frac{\pi}{2}}\left[\frac{4}{\gamma^2}-\frac{1}{2}\right]\mu^{-\frac{Q-3\gamma/4}{\gamma}-1},
\end{align*}
where $R_\gamma$ is given in Proposition \ref{rem-QD-length}.  Hence
\begin{align*}
G^{\rm cluster}_3(u_1,u_2,u_3)=&\frac{4}{KC^3} C_\gamma^{\rm DOZZ}\left(Q-\frac{\gamma}{4},Q-\frac{\gamma}{4},Q-\frac{\gamma}{4}\right)^{-1}R_\gamma^2 2^{4/\gamma^2}\Gamma\left(\frac{4}{\gamma^2}\right)\\
&\frac{1}{\beta^3\pi^3}\left(\sin\frac{\pi\gamma^2}{4}\right)^{2/\gamma^2-1/4}\cos^2\frac{4\pi}{\gamma^2} \sqrt{\frac{\pi}{2}}.
\end{align*}

For $G^{\rm cluster}_2$, the left hand side of~\eqref{weld2} equals
$KC^2 G^{\rm cluster}_2(u_1,u_2) \frac{1}{2}C_\gamma^{\rm DOZZ}(Q-\frac{\gamma}{4},Q-\frac{\gamma}{4},\gamma)\mu^{-1/2}$ by Proposition~\ref{prop-DOZZ}, whereas by Proposition \ref{levy1} the right hand side of~\eqref{weld2} equals
\begin{align*}
\int_0^\infty 2\ell |\mathrm{QD}(\ell)|^2 \mathrm{QD}^{\#}(\ell)\left[Ae^{-\mu A}\right]\cdot \frac{\sin(-\pi\beta)}{\beta^2\pi}\ell^{1+\beta}e^{-M\ell}\frac{2M}{L} d\ell=-\frac{1}{\pi\beta^2}R_\gamma^2\frac{\cos \frac{4\pi}{\gamma^2}}{\sqrt{\sin\frac{\pi\gamma^2}{4}}}\mu^{-1/2}.
\end{align*}
Hence $G^{\rm cluster}_2(u_1,u_2)=-\frac{2}{KC^2}C_\gamma^{\rm DOZZ}\left(Q-\frac{\gamma}{4},Q-\frac{\gamma}{4},\gamma\right)^{-1}\frac{1}{\pi\beta^2}R_\gamma^2\frac{\cos \frac{4\pi}{\gamma^2}}{\sqrt{\sin\frac{\pi\gamma^2}{4}}}.$
Therefore
\begin{equation*} 
\frac{G^{\rm cluster}_3(u_1,u_2,u_3)}{G^{\rm cluster}_2(u_1,u_2)^{3/2}}=\frac{K^{1/2}}{\pi R_\gamma}2^{4/\gamma^2}\Gamma\left(\frac{4}{\gamma^2}\right)\left(\sin\frac{\pi\gamma^2}{4}\right)^{2/\gamma^2+1/2}(-\cos\frac{4\pi}{\gamma^2})^{1/2}\times\frac{C_\gamma^{\rm DOZZ}\left(Q-\frac{\gamma}{4},Q-\frac{\gamma}{4},\gamma\right)^{3/2}}{C_\gamma^{\rm DOZZ}\left(Q-\frac{\gamma}{4},Q-\frac{\gamma}{4},Q-\frac{\gamma}{4}\right)}.
\end{equation*}
Plugging in the expressions of $C_\gamma^{\rm DOZZ}$  from Lemma~\ref{lem:dozz}  and $\frac{K}{R_\gamma^2}$ from \eqref{eq:k-r}, we get the $\kappa\in (\frac{8}{3},4)$ case for Theorem \ref{thm: connectivity}  after simplifying the expression.
The case $\kappa=4$ follows from taking the limit as $\kappa \nearrow 4$; see Lemmas~\ref{lem:Green-cont} and~\ref{lem:cont4} in Appendix~\ref{kappato4}.
\end{proof}

\begin{proof}[Proof of Theorem \ref{thm: connectivity} for $\kappa'\in (4,8)$ assuming Proposition~\ref{levy2}.]
Let $\gamma^2=\frac{16}{\kappa'}$.
Similar to the simple case, we evaluate both sides of~\eqref{weld3dense} and~\eqref{weld2dense}.
By Proposition~\ref{prop-DOZZ}, the left hand side of~\eqref{weld3dense} equals
\begin{equation*}
K'{C'}^3 G^{\rm cluster}_3(u_1,u_2,u_3) \frac{1}{2\gamma}(Q-\frac{3}{\gamma})C_\gamma^{\rm DOZZ}\left(Q-\frac{1}{\gamma},Q-\frac{1}{\gamma},Q-\frac{1}{\gamma}\right)\mu^{-\frac{Q-3/\gamma}{\gamma}-1}
\end{equation*}
where $K'$ is as in~\eqref{wc}, and by Proposition~\ref{levy2} the right hand side of~\eqref{weld3dense} equals
\begin{equation*}
2^{4/\kappa'}{M'}^{1-\beta'}{R_\gamma'}^2\frac{1}{\pi^3\beta'^3}\sqrt\frac{\pi}{2}\cos^2\frac{4\pi}{\kappa'}\frac{\kappa'}{4}\Gamma\left(\frac{4}{\kappa'}\right)\left[\frac{4}{\kappa'}-\frac{1}{2}\right]\frac{1}{\mu},
\end{equation*}
where $R_\gamma'$ is as in Lemma~\ref{lem:gqd-bdy-length-law}, and $M'$ is as in Proposition~\ref{levy2}.
Combining these two equations gives \begin{align*}
G^{\rm cluster}_3(u_1,u_2,u_3)=\frac{4}{K'{C'}^3} C_\gamma^{\rm DOZZ}\left(Q-\frac{1}{\gamma},Q-\frac{1}{\gamma},Q-\frac{1}{\gamma}\right)^{-1}{R_\gamma'}^2 2^{4/\kappa'}\Gamma\left(\frac{4}{\kappa'}\right)\frac{1}{\beta'^3\pi^3}{M'}^{1-\beta'}\cos^2\frac{4\pi}{\kappa'} \sqrt{\frac{\pi}{2}}.
\end{align*}

For $G^{\rm cluster}_2$, the left hand side of~\eqref{weld2dense} equals
$K'{C'}^2 G^{\rm Gar}_2(u_1,u_2) \frac{1}{2}C_\gamma^{\rm DOZZ}(Q-\frac{1}{\gamma},Q-\frac{1}{\gamma},\gamma)\mu^{-1+\frac{\kappa'}{8}}$ by Proposition~\ref{prop-DOZZ}, whereas by Proposition~\ref{levy2} the right hand side of~\eqref{weld2dense} equals
$-\frac{1}{\pi\beta^2}{R_\gamma'}^2 M'\frac{\kappa'}{4\mu}\cos\frac{4\pi}{\kappa'}$. 
Hence $G^{\rm cluster}_2(u_1,u_2)=-\frac{2}{K'{C'}^2}C_\gamma^{\rm DOZZ}\left(Q-\frac{1}{\gamma},Q-\frac{1}{\gamma},\gamma\right)^{-1}\frac{1}{\pi\beta'^2}{R_\gamma'}^2 M'\frac{\kappa'}{4}\cos\frac{4\pi}{\kappa'}.$
Therefore
\begin{equation*} 
\frac{G^{\rm cluster}_3(u_1,u_2,u_3)}{G^{\rm cluster}_2(u_1,u_2)^{3/2}}=\frac{{K'}^{1/2}}{\pi R_\gamma'}2^{4/\kappa'}\left(\frac{4}{\kappa'}\right)^{3/2}\Gamma\left(\frac{4}{\kappa'}\right){M'}^{-\beta'-1/2}(-\cos\frac{4\pi}{\kappa'})^{1/2}\frac{C_\gamma^{\rm DOZZ}\left(Q-\frac{1}{\gamma},Q-\frac{1}{\gamma},\gamma\right)^{3/2}}{C_\gamma^{\rm DOZZ}\left(Q-\frac{1}{\gamma},Q-\frac{1}{\gamma},Q-\frac{1}{\gamma}\right)},
\end{equation*}
where $K'/{R_\gamma'}^2$ is given by~\eqref{wc}. Applying Lemma~\ref{lem:dozz} and simplifying the above expression gives the final result.
\end{proof}

\section{Proof of Propositions \ref{levy1} and \ref{levy2} via L\'evy process}\label{levy}
Recall  the  L\'evy process from Propositions~\ref{prop-ccm} and \ref{prop-ccm-ns}. Namely, for $\beta\in (1,2)$, let $(\zeta_t)_{t\geq 0}$ be a $\lexp$-stable L\'evy process whose L\'evy measure is $1_{x>0} x^{-\beta-1} \, dx$, and denote its law by $\P^\beta$.
Let $\tau_{-\outleng}=\inf\{ t: \zeta_t=-\outleng  \}$. Let  $ (x_i)_{i \geq 1}$ be the sequence of the sizes of the upward jumps of $\zeta$ on $[0,\tau_{-\outleng}]$ sorted in decreasing order. The starting point of our proof of Propositions \ref{levy1} and \ref{levy2} is the following proposition. 
\begin{proposition}\label{cm}
Recall the setting of Propositions \ref{levy1} and \ref{levy2}, where for $\ell>0$ we  consider $\QD(\ell)^\#\otimes\CLE_\kappa$ when $\kappa\in (\frac{8}{3},4)$
and $\GQD(\ell)^\#\otimes\CLE_{\kappa'}$ when $\kappa'\in (4,8)$. Moreover,  $Y$ is the total quantum natural measure  on the CLE cluster, and $A$ is the total quantum area of the quantum disk.
Then for positive integer $k$,
\begin{equation}\label{eq:simpleYA}
\QD(\ell)^\#\otimes\CLE_\kappa[Y^ke^{-\mu A}]=\frac{1}{\beta^k}\frac{\E[\tau_{-\ell}^{k-1} e^{-\mu\sum x_i^2 A_i}]}{\E \tau_{-\ell}^{-1}},
\end{equation}
where $A_i$ denotes  independent copies of the quantum area of a sample from $\QD(1)^\#$ and the expectation  is taken with respect to $\P^\beta$ with $\beta=\frac{4}{\kappa}+\frac{1}{2}$. 
Similarly, \begin{equation}\label{eq:nonsimpleYA}
\GQD(\ell)^\#\otimes\CLE_{\kappa'}[Y^ke^{-\mu A}]=\frac{1}{\beta'^k}\frac{\E[\tau_{-\ell}^{k-1} e^{-\mu\sum x_i^{8/\kappa'} A_i}]}{\E \tau_{-\ell}^{-1}},
\end{equation}
where $A_i$ denotes  independent copies of the quantum area of a sample from $\GQD(1)^\#$ and the expectation is taken with respect to $\P^{\beta'}$ with $\beta'=\frac{4}{\kappa'}+\frac{1}{2}$.
\end{proposition}
\begin{proof}
To prove~\eqref{eq:simpleYA},
for each $T>0$, let $N_\eps(T)$ be the number of jumps of $\zeta$ of size greater than $\eps$ in the interval $[0,T]$. 
Since the Levy measure of $\zeta$ is  $1_{x>0} x^{-\beta-1} \, dx$,   for each fixed $T$ we have that  $\P^\beta$-a.s.  
\[
\lim_{\eps \to 0} \eps^\beta N_\eps(T) =\lim_{\eps \to 0} \eps^\beta \E[N_\eps(T)] = \lim_{\eps \to 0} \eps^\beta \int_0^T\int_{\eps}^\infty x^{-\beta-1} \, dx\, dt = \frac1\beta T.
\]	
Therefore $\P^\beta$-a.s.\ $\lim_{\eps \to 0} \eps^\beta N_\eps(\tau_{-L}) =  \frac1\beta \tau_{-L}$, hence the same holds under $\frac{\tau_{-L}^{-1}\P^\beta}{\E[\tau_{-L}^{-1}]}$.
In the coupling of $\frac{\tau_{-L}^{-1}\P^\beta}{\E[\tau_{-L}^{-1}]}$ and $\QD(\ell)^\# \otimes \CLE_\kappa$ described in Proposition~\ref{prop-ccm},  
the number of outermost  loops in the $\CLE_\kappa$ with quantum length greater than $\eps$ is $N_\eps(\tau_{-L})$. By the definition of the quantum natural measure on the CLE carpet we have  $Y= \lim_{\eps\to 0} \eps^\lexp N_\eps(\tau_{-L}) = \frac1\beta \tau_{-L}$ a.s. This gives~\eqref{eq:simpleYA}.

Equation~\eqref{eq:nonsimpleYA} follows from Proposition~\ref{prop-ccm-ns} via the same argument with  $x_i^{8/\kappa'}$ in place  of $x_i^2$ because of Proposition~\ref{prop:ns-area-law}.
\end{proof}

The key to the proofs of Propositions~\ref{levy1} and~\ref{levy2} is to express the quantities on the right hand sides of \eqref{eq:simpleYA} and \eqref{eq:nonsimpleYA}  into  certain integrals over  the L\'evy excursion measure. This is done in Section \ref{sec:reduce}, after we review the necessary background on L\'evy excursions in  Section~\ref{sec:excursion}. Then  in Section \ref{sec:solve-levy1}, we evaluate these integrals and prove Propositions \ref{levy1} and \ref{levy2}.

\subsection{Background on the L\'evy process and L\'evy excursion}\label{sec:excursion}
For $\beta\in(1,2)$,
let $(\zeta_t)_{t\geq 0}$ be a $\beta$-stable L\'evy process sampled from $\P^{\beta}$. In this section we recall some basic facts and constructions concerning  $(\zeta_t)_{t\geq 0}$. For further background, see e.g.~\cite[Section 1]{duquesne-legall-levy-trees},  \cite[Section 3.1.1]{curien-kortchemski-looptree-def}, and \cite[Section 1]{duquesne2002random}.
Let $J:=\{(x,t): t\ge 0 \textrm{ and }\zeta_t-\zeta_{t^-}=x>0 \}$ be the set of jumps, where $\zeta_{t-}:=\lim_{s\to t-}\zeta(s)$. 
For each $a>0$, we denote $\tau_{-a}=\inf\{ t: \zeta_t=-a\}$ for the hitting time of $-a$ and $J_a=\{(x,t)\in J: t\le \tau_{-a}\}$ for the jumps before hitting $-a$.

\begin{lemma}\label{tau}
    $\E \tau_{-\ell}^{-1}=\E[\tau_{-1}^{-1}]\ell^{-\beta}=\frac{\pi}{\sin(-\pi\beta)}\ell^{-\beta}.$
\end{lemma}
\begin{proof}
    Since $\E[e^{-\lambda\zeta_t}]=e^{\Gamma(-\beta)t\lambda^\beta}$, by \cite[Theorem 7.1]{bertoin-book}, we have $\mathbb{E}\left[e^{-\lambda \tau_{-\ell}}\right]=e^{\ell \Gamma(-\beta)^{-1/\beta}\lambda^{1 /\beta}}$. Then $\mathbb{E}\left[\tau_{-\ell}^{-1}\right]=\int_0^{\infty} \mathbb{E}\left[e^{-\lambda \tau_{-\ell}}\right] d \lambda=\ell^{-\beta}\Gamma(-\beta)\Gamma(1+\beta)=\frac{\pi}{\sin (-\pi \beta)}\ell^{-\beta}$.
\end{proof}

Let $I_t= \inf\{\zeta_t: s\in [0,t] \}$. Then $(\zeta_t-I_t)_{t\ge 0}$ is a Markov process and we can consider its excursions away from 0. Write $\ul N$ for the excursion measure, which is a measure on the space of non-negative functions of the form  $e:[0,T] \rta [0,\infty)$ with $e(0)=e(T)=0$; we say $T$ is the duration of the excursion.
Let $\{x_i: i\ge 1\}$ be the set of jumps in $\zeta$ before $\tau_{-L}$. Each excursion of $(\zeta_t - I_t)_{t \in [0,\tau_{-L}]}$ corresponds to a pair $(e, s)$ as follows: let $\sigma$ be the starting time of the excursion, then $e$ is the excursion started at time $0$ instead of time $\sigma$, and $s = -I_\sigma$. Then the set $\mathrm{Exc}_{L}=\{ (e,s): s\le L \}$ has the law of a Poisson point process with intensity measure $\ul N\times 1_{s\in[0,L]}  ds$. Let $T(e)$ denote the duration $T$ of $e$;	note that $\tau_{-L}=\sum_{ (e,s)\in \mathrm{Exc}_{L}}  T(e)$ and 
 \(	\tau_{-\ell} = \sum_{ (e,s)\in \mathrm{Exc}_{L}}  T(e)1_{\{s\le \ell\}}\).

We will use the following lemma to compute the terms $\E[\tau_{-\ell}^{k-1} e^{-\mu\sum x_i^2 A_i}]$  and 
$\E[\tau_{-\ell}^{k-1} e^{-\mu\sum x_i^{8/\kappa'} A_i}]$ in Propositions~\ref{levy1} and \ref{levy2}.
\begin{lemma}\label{palm}
Let $\ell > 0$ and let  $(x_i)$ be the set of jumps of $\zeta$ up until time  $\tau_{-\ell}$. We have
\begin{equation*}
\E[\tau_{-\ell}\prod_{i\ge 1}g(x_i)]=\ell\E[\prod_{i\ge 1}g(x_i)]\int T(e)\prod_{x\in J_e}g(x) \underline N(de)
\end{equation*}
and
\begin{equation*}
\E[\tau_{-\ell}^2\prod_{i\ge 1}g(x_i)]=\ell^2\E[\prod_{i\ge 1}g(x_i)]\left(\int T(e)\prod_{x\in J_e}g(x) \underline N(de)\right)^2+\ell\E[\prod_{i\ge 1}g(x_i)]\int T(e)^2\prod_{x\in J_e}g(x) \underline N(de).
\end{equation*}
where $J_e$ stands for the set of jumps in $e$.
\end{lemma}
\begin{proof}
Suppose $F$ is a function on the space of excursions. Since $\tau_{-\ell}=\sum_{e\in \mathrm{Exc}_{\ell}}  T(e)$, we have $\E[\tau_{-\ell}\prod_{e\in \mathrm{Exc}_{\ell}}F(e)]=\E\left[\left(\sum_{e\in \mathrm{Exc}_{\ell}}  T(e)\right)\prod_{e\in \mathrm{Exc}_{\ell}}F(e)\right]$.
By Palm's theorem (see e.g.\ \cite[Page 5]{kallenberg2017random}) for Poisson point process, for $(e,s)$ chosen from the counting measure on ${\rm Exc}_\ell$, the law of $(e,s)$ is precisely the intensity measure $\underline N(de)\times 1_{s\in[0,\ell]}ds$, and conditioned on $e$, the conditional law of ${\rm Exc}_\ell\backslash\{(e,s)\}$ is the original Poisson point process with intensity measure $\underline N(de)\times 1_{s\in[0,\ell]}ds$. Hence
\begin{equation}\label{eq:palm0}
\E[\tau_{-\ell}\prod_{e\in \mathrm{Exc}_{\ell}}F(e)]=\E\left[\left(\sum_{e\in \mathrm{Exc}_{\ell}}  T(e)\right)\prod_{e\in \mathrm{Exc}_{\ell}}F(e)\right]=\ell\E[\prod_{e\in \mathrm{Exc}_{\ell}}F(e)]\int T(e)F(e) \underline N(de).
\end{equation}
Setting $F(e)=\prod_{x\in J_e} g(x)$ gives the first result. For the second result, we choose $F(e)=e^{\lambda T(e)}\prod_{x\in J_e} g(x)$; the result follows by taking the derivative of $\lambda$ at $\lambda=0$.
\end{proof}

\subsection{Reduction to L\'evy excursion functionals}\label{sec:reduce}

To prove Propositions \ref{levy1} and \ref{levy2}, in Lemma \ref{palm}, we choose $g(x)=\QD(1)^\#[e^{-x^2\mu A}]$ for the simple $\CLE_\kappa$ and $g(x)=\GQD(1)^\#[e^{-x^{8/\kappa'}\mu A}]$ for the non-simple $\CLE_{\kappa'}$. Then  it remains to calculate $\E[\prod_{i\ge 1}g(x_i)]$ and the following integrals:
\begin{equation}\label{eq:remain}
\int T(e)\prod_{x\in J_e}g(x) \underline N(de)\quad \text{and}\quad\int T(e)^2\prod_{x\in J_e}g(x) \underline N(de).
\end{equation}
We will use results in \cite{holden2022liouville} to compute~\eqref{eq:remain}. To be consistent with the normalization in  \cite{holden2022liouville},
let $(\wt\zeta_t)$ be the $\beta$-stable process
 with L\'evy measure $\frac{1}{\Gamma(-\beta)}1_{\{x>0\}}x^{-\beta-1}dx$, and denote its law by $\wt \P^\beta$. Note that $(\wt\zeta_{\Gamma(-\beta)t})\overset{d}{=}(\zeta_t)$ for $(\zeta_t)\sim \P^\beta$ from Section \ref{sec:excursion}.
 For this process let  $\wt N_\ell$ be the probability measure corresponding to duration $\ell$ excursions, and let $\wt {\ul N}$ be the excursion measure. Then 
\begin{equation}\label{excursion1}
\wt{\ul N} = \frac{1}{\beta\cdot \Gamma(1-1/\beta)}\cdot \int_0^\infty d\ell\ \ell^{-1-1/\beta} \wt N_\ell.
\end{equation}
See e.g.\ \cite[Section 3.1.1]{curien-kortchemski-looptree-def}. Through the time change $(\wt\zeta_{\Gamma(-\beta)t})\overset{d}{=}(\zeta_t)$, the excursion with duration $\Gamma(-\beta)\ell$ in $\ul N$ has a one-to-one correspondence to the excursion with duration $\ell$ in $\wt{\ul N}$. 

\begin{proposition}[{\cite{holden2022liouville}}]\label{lem:palm-levy-excursion}
Fix $\beta\in (1,2)$. Suppose $G\colon [0,\infty)\to [0,\infty)$ is twice continuously differentiable with $G(0)=G'(0)=0$. Then  for $\lambda\ge0$ we have
\begin{equation}\label{eq：Gpositive}
 \log \E\left( e^{-\lambda \tau_{-1}-\sum_{t< \tau_{-1},\Delta B_t\neq0} G(\Delta B_t)} \right)
		= \frac{1}{\beta \Gamma(1-1/\beta)} \int_0^\infty \frac{d\ell}{\ell^{1+1/\beta}}\left( e^{-\lambda \ell}\,\E\left( e^{-\sum_{t\le 1,\Delta b_t\neq0} G(\ell^{1/\beta}\Delta b_t) } \right) - 1\right),
\end{equation}
where  the first (resp.\ second) expectation is taken with respect to $(B_t)$ sampled from $\wt\P^\beta$ (resp.\ $(b_t)$ sampled from $\wt N_1$), and 
 $\Delta B_t:= B_t - B_{t-}$  (resp. $\Delta b_t:= b_t - b_{t-}$).
Furthermore,  suppose $\lambda>0$ and $\rho(\lambda)>0$ are such that $\lambda = \frac{1}{\Gamma(-\beta)} \int_0^\infty \frac{dy}{y^{1+\beta}}\,(e^{-\rho(\lambda) y - G(y)}-1+\rho(\lambda) y)$. Then
	\begin{align}
		\label{eq44}
		-\rho(\lambda) = \log \E\left( e^{-\lambda \tau_{-1}-\sum_{t< \tau_{-1},\Delta B_t\neq0} G(\Delta B_t)} \right). 
	\end{align}
\end{proposition}
\begin{proof}
 Eq.~\eqref{eq44} is exactly \cite[Lemma 4.2]{holden2022liouville}. Eq.~\eqref{eq：Gpositive} is an intermediate step in the proof of that lemma; see the indented equation right below  \cite[Eq. (4-6)]{holden2022liouville}.
\end{proof}

\subsection{Proof of Propositions \ref{levy1} and~\ref{levy2}}\label{sec:solve-levy1}
We first treat the case $\kappa\in (\frac{8}{3},4)$, where the relevant L\'evy process has parameter $\beta=\frac{4}{\kappa}+\frac{1}{2}$.  The key to our proof is the following proposition, which comes from an involved calculation.
\begin{proposition}\label{cal1} 
For $\mu>0$, let $g(x):=\QD(1)^\#[e^{-x^2\mu A}]$. Recall $\underline{\wt N}$ from~\eqref{excursion1}. Set
\begin{equation}\label{eq:defu}
u(s):=\int \Big(e^{-\frac{T(e)}{\Gamma(-\beta)}s} \prod_{x\in J_e} g(x)-1\Big) \wt{\ul N} (d e) \quad \textrm{for } s\ge 0.
\end{equation}
Then $u(s)$ is finite for all $s\ge0$. Furthermore, we have
\begin{equation*}
u(0)=-M,\quad u'(0)=-\frac{2M}{L},\quad  \textrm{and}\quad  u''(0)=\frac{M}{2L^2}\left(\frac{64}{\kappa^2}-1\right),
\end{equation*}
where  $M=\sqrt{\frac{\mu}{\sin\frac{\gamma^2\pi}{4}}}$ and $L=-\frac{2^{1-\frac{4}{\kappa}}M^{\frac{4}{\kappa}+\frac{1}{2}}}{\Gamma(\frac{4}{\kappa})}\sqrt{2\pi}\frac{\pi}{\cos\frac{4\pi}{\kappa}}$.
\end{proposition}

To prove Proposition~\ref{cal1} we need the following crucial lemma.

\begin{lemma}\label{lem-6}
Let 
\begin{equation}\label{4}
\lambda(\sigma) = \int_0^\infty \frac{dy}{y^{1+\beta}}\,(e^{-\sigma y}g(y)-1+\sigma y) \quad \textrm{for }\sigma>0.
\end{equation}
Let $_2F_1(\cdot,\cdot;\cdot;\cdot)$ be the hypergeometric function; see Appendix~\ref{appendix:integration}.
Then\begin{equation}\label{6'}
\lambda(\sigma)=L\left(\frac{2M}{M+\sigma}\right)^{\frac{4}{\kappa}-\frac{1}{2}}\frac{\sigma-M}{\sigma+M}\,_2F_1\left(\frac{1}{2}+\frac{4}{\kappa},\frac{4}{\kappa}+\frac{3}{2};2;\frac{\sigma-M}{\sigma+M}\right).
\end{equation}
\end{lemma}
\begin{proof}
We prove Lemma~\ref{lem-6} based on Lemma~\ref{lem:bigint} via an analytic continuation of $\int_0^\infty \frac{dy}{y^{1+\beta}} e^{-\sigma y} g(y) y^{\frac12+\nu}$.
For $\nu \in \{z \in \C\: : \:  \Re z > \beta - \frac52\} \backslash \{ \beta - \frac 12, \beta - \frac32\}$ and $\sigma>0$, define 
\begin{equation*}
\Lambda_\sigma(\nu) = \int_0^1 \frac{dy}{y^{1+\beta}}\,(e^{-\sigma y}g(y)-1+\sigma y)y^{\frac{1}{2}+\nu}+\int_1^\infty \frac{dy}{y^{1+\beta}}e^{-\sigma y}g(y)y^{\frac{1}{2}+\nu}-\frac{1}{-\frac{1}{2}+\beta-\nu}+\frac{\sigma}{\beta-\frac{3}{2}-\nu}.
\end{equation*}
We evaluate $\Lambda_\sigma(-\frac12)$ in two ways.
On the one hand, when $\Re \nu \in (\beta - \frac52, \beta - \frac32)$, this can be simplified as 
$\Lambda_\sigma (\nu) = \int_0^\infty \frac{dy}{y^{1+\beta}}\,(e^{-\sigma y}g(y)-1+\sigma y)y^{\frac{1}{2}+\nu}$, hence $\Lambda_\sigma(-\frac12)=\lambda(\sigma)$.
On the other hand, when $\Re \nu > \beta -\frac12$ we have $\Lambda_\sigma(\nu) = \int_0^\infty \frac{dy}{y^{1+\beta}} e^{-\sigma y} g(y) y^{\frac12+\nu}$. By  Proposition~\ref{prop-QD-area} we have $g(x)=\QD(1)^\#[e^{-x^2\mu A}]=\ol K_{\frac{4}{\gamma^2}}(Mx)$, hence by Lemma~\ref{lem:bigint} for $\Re \nu > \beta -\frac12$ we have
\begin{equation}\label{5'}
\Lambda_\sigma(\nu)=\frac{2^{1-\frac{4}{\kappa}}M^{\frac{4}{\kappa}}}{\Gamma(\frac{4}{\kappa})}\frac{\sqrt{\pi}(2M)^{\frac{4}{\kappa}}}{(M+\sigma)^{\frac{4}{\kappa}+\nu}}\frac{\Gamma(\nu+\frac{4}{\kappa})\Gamma(\nu-\frac{4}{\kappa})}{\Gamma(\nu+\frac{1}{2})}\,_2F_1\left(\nu+\frac{4}{\kappa},\frac{4}{\kappa}+\frac{1}{2};\nu+\frac{1}{2};\frac{\sigma-M}{\sigma+M}\right)
\end{equation}
Since the right hand side of~\eqref{5'} is meromorphic,~\eqref{5'} extends to the domain $\{z \in \C :  \Re z > \beta - \frac52\} \backslash \{ \beta - \frac 12, \beta - \frac32\}$. Using  $\lim_{c\to0}\frac{_2F_1(a,b;c;z)}{\Gamma(c)}=abz\,_2F_1(a+1,b+1;2;z)$ and simplifying, we see that $\Lambda_\sigma(-\frac12)$ agrees with the right hand side of~\eqref{6'}, as desired.
\end{proof}

\begin{proof}[Proof of Proposition~\ref{cal1}]
Set $G(x)=-\log g(x)$; note that $G\colon [0,\infty)\to [0,\infty)$ is twice continuously differentiable with $G(0)=G'(0)=0$. For $s\ge0$ we have 
\begin{align}
u(s)&=\frac{1}{\beta \Gamma(1-1/\beta)} \int_0^\infty \frac{d\ell}{\ell^{1+1/\beta}}\left( e^{-\frac{s}{\Gamma(-\beta)} \ell}\,\E\left(\prod_{t<1,\Delta b_t\neq0} g(\ell^{1/\beta}\Delta b_t)) \right) - 1\right)\label{equ10}\\ 
    &= \log \E\left( e^{-\frac{s}{\Gamma(-\beta)} \tau_{-1}-\sum_{t< \tau_{-1},\Delta B_t\neq0} G(\Delta B_t)} \right),  \label{equ10'}  
\end{align}
where~\eqref{equ10} is by \eqref{eq:defu}, and \eqref{excursion1} and~\eqref{equ10'} is by~\eqref{eq：Gpositive}.
From~\eqref{equ10'} we see that $u$ is  finite and continuous on $[0,\infty)$.

Now we use Lemma~\ref{lem-6} to compute $u(0)$, $u'(0)$ and $u''(0)$. For $\lambda=\lambda(\sigma)$ defined in \eqref{4},  since $_2F_1(a,b;c;0)=1$, by~\eqref{6'} in Lemma~\ref{lem-6} we have $\lambda (M) = 0$ and $\lambda'(M)> 0$.
Therefore, for some $\lambda_0>0$, there is a continuous inverse function $\sigma: [0,\lambda_0) \to [M,+\infty)$ of $\lambda(\cdot)$ near $M$.  Namely,  $\sigma(0) = M$, and for each $\lambda \in [0,\lambda_0)$, we see that \eqref{4} holds with $\sigma = \sigma(\lambda)\in[M,+\infty)$. 

Fix $\lambda\in (0,\lambda_0)$. Since $\sigma=\sigma(\lambda)$ satisfies \eqref{4}, by~\eqref{eq44} in Proposition \ref{lem:palm-levy-excursion} we have
	\begin{align}\label{equ9}
		-\sigma(\lambda) &= \log \E\left( e^{-\frac{\lambda}{\Gamma(-\beta)} \tau_{-1}-\sum_{t< \tau_{-1},\Delta B_t\neq0} G(\Delta B_t)} \right)=u(\lambda),
	\end{align}
where the last equality is due to \eqref{equ10}. By continuity of $u$ and $\sigma$ at $0$, we have $u(0)=-\sigma(0)$, which equals $-M$.
We compute $\sigma'(0)$ and $\sigma''(0)$ using \eqref{6'}. For $\sigma'(0)$,
we have $\lambda'(M)\sigma'(0)=1$ hence $ \sigma'(0)=\frac{2M}{L}$.
For $\sigma''(0)$, using $\frac{d}{dz}\,_2F_1(a,b;c;z)\big|_{z=0}=\frac{ab}{c}$, we have 
\begin{equation*}
0=2\cdot \frac{\sigma'(0)}{2M}\cdot \frac{1}{2}\left(\frac{1}{2}+\frac{4}{\kappa}\right)\left(\frac{3}{2}+\frac{4}{\kappa}\right)\frac{\sigma'(0)}{2M}+\frac{\sigma''(0)}{2M}-\frac{\sigma'(0)^2}{2M^2}\left(\frac{1}{2}+\frac{4}{\kappa}\right).
\end{equation*}
Therefore $\sigma''(0)=\frac{\sigma'(0)^2}{2M}\left(\frac{1}{2}+\frac{4}{\kappa}\right)\left(\frac{1}{2}-\frac{4}{\kappa}\right)=-\frac{M}{2L^2}\left(\frac{64}{\kappa^2}-1\right)$.
Since $u'(0)=-\sigma'(0)$ and $u''(0)=-\sigma''(0)$ we conclude the proof.
\end{proof}

As a byproduct of the proof of Proposition~\ref{cal1}, we get the following proposition, which serves as a basic input for the calculations in Sections~\ref{sec-welding} and~\ref{sec:3-pt}.  
\begin{proposition}\label{l}
 Let $\kappa\in (\frac{8}3,4)$, $\gamma=\sqrt{\kappa}$ and $\beta=\frac{4}{\kappa}+\frac{1}{2}$.
 Fix $\ell>0$. Suppose $\zeta$ is a $\beta$-stable L\'evy process.	
 Let  $ (x_i)_{i \geq 1}$ be the collection of sizes of the jumps of $\zeta$ on $[0,\tau_{-\ell}]$ sorted in decreasing order.
  	Let $(A_i)_{i\ge 1}$ be i.i.d. random variables  independent of $(x_i)_{i\ge 1}$, whose law is given by the quantum area of a sample from $\QD(1)^\#$. Then
    \begin{equation}\label{eq:area-annulus-imp}
 		\E[e^{-\mu \sum x_i^2 A_i}] = e^{-\ell\sqrt{\mu/\sin(\pi\gamma^2/4)}} \quad \text{for }\mu\geq0.
 	\end{equation}
\end{proposition}
\begin{proof}
In the proof of Proposition~\ref{cal1}, setting $s=0$ in~\eqref{equ10} gives $u(0) = \log \E \left[\prod_{t< \tau_{-1},\Delta B_t\neq0} g(\Delta B_t) \right]$.
Since  $u(0)=-M=-\sqrt{\frac{\mu}{\sin\frac{\gamma^2\pi}{4}}}$ by Proposition~\ref{cal1}, we get the $\ell=1$ case of $\E[ e^{-\mu\sum x_i^2 A_i}]=e^{-\ell\sqrt{\mu/\sin(\pi\gamma^2/4)}}$. The claim for general $\ell$ follows from scaling.
\end{proof}

\noindent  The next lemma ensures that $u^{(k)}(0)= \int (-1)^k\frac{T(e)^k}{\Gamma(-\beta)^k} \prod_{x\in J_e} g(x) \wt{\ul N} (d e)$ which gives Proposition~\ref{levy1}.
\begin{lemma}\label{lem:change}
Suppose $F(x):[0,\infty)\to[0,\infty)$, $\alpha\in(1,2)$ and $u(\lambda):=\int_0^\infty x^{-\alpha} (e^{-\lambda x}F(x)-1)dx$ exists for $\lambda\ge0$. If $u:[0, \infty) \to \R$ is twice-differentiable at $0$,
then $\int_0^\infty x^{-\alpha+1} F(x)dx$ and $\int_0^\infty x^{-\alpha+2}  F(x)dx$ are finite and agree with $-u'(0)$ and $u''(0)$ respectively.
\end{lemma}
\begin{proof}
Note that $\frac{u(0)-u(\lambda)}{\lambda}=\int_0^\infty x^{-\alpha} F(x)\frac{1-e^{-\lambda x}}{\lambda}dx$.
Taking the $\lambda\to0$ liminf and using the Fatou's lemma we find
$-u'(0)\ge\int_0^\infty x^{-\alpha+1} F(x)dx$, hence $\int_0^\infty x^{-\alpha+1}  F(x)dx$ is finite. Now since $\frac{1-e^{-\lambda x}}{\lambda}<x$, the finiteness of $\int_0^\infty x^{-\alpha+1} F(x)dx$ and the dominated convergence theorem then implies $-u'(0)=\int_0^\infty x^{-\alpha+1} F(x)dx$.

Similarly consider $\frac{4}{\lambda^2}(u(0)+u(\lambda)-2u(\frac{\lambda}{2}))$, we get the finiteness of $\int_0^\infty x^{-\alpha+2}  F(x)dx$. Then $\frac{4}{\lambda^2}(1+e^{-\lambda x}-2e^{\frac{1}{2}\lambda x})<x^2$ and the dominated convergence theorem gives $u''(0)=\int_0^\infty x^{-\alpha+2} F(x)dx$.
\end{proof}

\begin{proof}[Proof of Proposition \ref{levy1}]
By Proposition \ref{cal1}, the function $u(x):[0,\infty)\to\R$ defined in \eqref{eq:defu} is twice differentiable at $0$. For $k=1,2$, by \eqref{equ10} and choosing $F(x)=\E\left(\prod_{t<1,\Delta b_t\neq0} g(x^{1/\beta}\Delta b_t)) \right)$ in Lemma~\ref{lem:change}, we have $u^{(k)}(0)= \int (-1)^k\frac{T(e)^k}{\Gamma(-\beta)^k} \prod_{x\in J_e} g(x) \wt{\ul N} (d e)$. The scaling relation between $\ul N$ and $\wt{\ul N}$ yields $\int (-1)^kT(e)^k \prod_{x\in J_e} g(x) \ul N (d e)=u^{(k)}(0)$, whose value is given in Proposition~\ref{cal1}.
Now by  Lemma~\ref{palm}  and Proposition~\ref{l} we get
$\E[\tau_{-\ell}  e^{-\mu\sum x_i^2 A_i}]=-\ell \E[e^{-\mu\sum x_i^2 A_i}] u'(0)=\ell \frac{2M}{L}e^{-M\ell}$. Furthermore, $\E[\tau_{-\ell}^{2} e^{-\mu\sum x_i^2 A_i}]=\ell^2\frac{4M^2}{L^2}e^{-M\ell}+\ell\frac{M}{2L^2}(\frac{64}{\kappa^2}-1)e^{-M\ell}$.
Since $\E \tau_{-\ell}^{-1}=\frac{\pi}{\sin(-\pi\beta)}\ell^{-\beta}$ by Lemma \ref{tau},
Proposition \ref{levy1} follows from Proposition~\ref{cm}.
\end{proof}

The proof of the non-simple case, i.e. Proposition \ref{levy2}, is similar to that of Proposition \ref{levy1}. Now the relevant L\'evy process has parameter $\beta'=\frac{4}{\kappa'}+\frac{1}{2}$. 
The following proposition is  the counterpart of Proposition~\ref{cal1}.

\begin{proposition}\label{cal2}
For $\mu>0$, let $g(x):=\GQD(1)^\#[e^{-x^{8/\kappa'}\mu A}]$.
Let \begin{equation*}
u(s):=\int \Big(e^{-\frac{T(e)}{\Gamma(-\beta)}s} \prod_{x\in J_e} g(x)-1\Big) \wt{ \ul N }(d e)\quad for\ s\ge0.
\end{equation*}
The $u(s)$ defined above is finite for $s\ge0$. Furthermore, we have
\begin{equation*}
u(0)=-M',\quad  u'(0)=-\frac{2M'}{L'},\quad \textrm{and}\quad  u''(0)=\frac{M'}{2L'^2}\left(\frac{64}{\kappa'^2}-1\right)
\end{equation*}
for $M'=2\left(\frac{\mu}{4\sin\frac{\gamma^2\pi}{4}}\right)^{\kappa'/8}$ and $L'=-\frac{2^{1-\frac{4}{\kappa}}M'^{\frac{4}{\kappa'}+\frac{1}{2}}}{\Gamma(\frac{4}{\kappa'})}\sqrt{2\pi}\frac{\pi}{\cos\frac{4\pi}{\kappa'}}$.
\end{proposition}
\begin{proof}
This follows from the calculations in the proof of Proposition \ref{cal1}  with $\kappa$ replaced by $\kappa'$. 
The counterpart of the key Lemma~\ref{lem-6} is
\begin{lemma}\label{lem:7}
For $g(x)$ as above, let $\lambda(\sigma) = \int_0^\infty \frac{dy}{y^{1+\beta'}}\,(e^{-\sigma y}g(y)-1+\sigma y)$ for $\sigma>0$.
Then\begin{equation*}
\lambda(\sigma)=L'\left(\frac{2M'}{M'+\sigma}\right)^{\frac{4}{\kappa'}-\frac{1}{2}}\frac{\sigma-M'}{\sigma+M'}\,_2F_1\left(\frac{1}{2}+\frac{4}{\kappa'},\frac{4}{\kappa'}+\frac{3}{2};2;\frac{\sigma-M'}{\sigma+M'}\right).
\end{equation*}
\end{lemma}
Since $g(x)=\ol K_{4/\kappa'}(M'x)$ from Propostion~\ref{prop:ns-area-law},
Lemma~\ref{lem:7} then follows from Lemma~\ref{lem:bigint} and the same analytic continuation argument as in Lemma~\ref{lem-6}.
\end{proof}

The exact same proof for Proposition~\ref{l} yields its non-simple counterpart. 
\begin{proposition}\label{l-nonsimple}
Let $\kappa'\in (4,8)$, $\gamma=\sqrt{16/\kappa'}$ and $\beta=\frac{4}{\kappa'}+\frac{1}{2}$. Let  $ (x_i)_{i \geq 1}$ and
$(A_i)_{i\ge 1}$ be defined in the same way as in Proposition~\ref{l} except $\QD(1)^\#$ is replaced by $\GQD(1)^\#$.  Then
\begin{equation}\label{eq:area-annulus-imp-2}
 		\E[e^{-\mu \sum x_i^2 A_i}] = \exp\{-2\ell\left(\frac{\mu}{4\sin(\pi\gamma^2/4)}\right)^{\kappa'/8}\} \quad \text{for }\mu\geq0.
 	\end{equation}
\end{proposition}

\begin{proof}[Proof of Proposition \ref{levy2}]
Similarly to the simple case, by  Lemma~\ref{lem:change}, for $k=1$ and 2, we have $\int T(e)^k \prod_{x\in J_e} g(x) \ul N (d e)=(-1)^ku^{(k)}(0)$, which is computed in Proposition~\ref{cal2}.
The rest of the argument is the same by combining Proposition \ref{cm}, Lemma~\ref{palm} and Proposition~\ref{l-nonsimple}.
\end{proof}

\section{Quantum annulus and quantum pair of pants}\label{sec-welding}
In this section, we introduce and study the quantum annulus and the quantum pair of pants. 
The key outputs results are a conformal welding result for the quantum pair of pants (Theorem~\ref{thm-weld-QA-3}) and the joint law of the area and boundary lengths of the quantum pair of pants (Theorem~\ref{thm-QA3-laplace}).   
These results will be used in Section~\ref{sec:3-pt} to prove Theorems \ref{thm:nesting}  for $\kappa \in (8/3, 4]$.
In Section~\ref{subsec:def-QA-QP} we define the quantum annulus and quantum pair of pants. In Section~\ref{subsec:3-pt-welding} we prove several conformal welding results including Theorem~\ref{thm-weld-QA-3}. 
In Section~\ref{subsec-QA-law} we derive the law of the area and boundary lengths of the quantum annulus.
In Section~\ref{subsec-QP-law} we prove Theorem~\ref{thm-QA3-laplace}.
Throughout this section we assume $\kappa\in (\frac83, 4)$ and $\gamma=\sqrt{\kappa}$.

\subsection{Definitions}\label{subsec:def-QA-QP}
We first review the definition of the quantum annulus given in \cite{ACSW24a}. Fix $a>0$. 
Suppose $(D, h, z)$ is an embedding of a sample from $\QD_{1,0}(a)$. Let $\Gamma$ be a $\CLE_\kappa$ on $D$ independent of $h$. 
Recall from Section~\ref{subsec:MSW} that $\QD_{1,0}\otimes \CLE_\kappa$ is the law of the decorated quantum surface  $(D, h,  \Gamma, z)/{\sim_\gamma}$. 
Let  $\eta$ be the outermost loop of $\Gamma$ surrounding $z$ and let $\frk l_\eta$ be the quantum length of  $\eta$. 
According to \cite[Lemma 5.1]{ACSW24a},   the disintegration of $\QD_{1,0}(a)\otimes\CLE_\kappa$ over $\frk l_\eta$ exists, which we denote by   
\(\{\QD_{1,0}(a)\otimes \CLE_\kappa(\ell): \ell\in (0,\infty)\}\).
\begin{definition}[Quantum annulus]\label{def-QA}
	Given $a>0$ and $(D,h,\Gamma, \eta, z)$ defined right above, let $A_{\eta}$ be the non-simply-connected component of $D\setminus \eta$. 
	For $b>0$, 	let $\wt \QA(a,b)$ be the law of the quantum surface  $(A_\eta, h)/{\sim_\gamma}$ under the measure $\QD_{1,0}(a)\otimes \CLE_\kappa(b)$. 
	Let $\QA(a,b)$ be such that 
	\begin{equation}\label{eq:QA}
		b|\QD_{1,0}(b)|\QA(a,b)= \wt \QA(a,b).
	\end{equation}
	We call a sample from $\QA(a,b)$ a \emph{quantum annulus}.
\end{definition}

As discussed in \cite[Remark 5.3]{ACSW24a}, Definition~\ref{def-QA} defines $\QA(a,b)$ only for Lebesgue-almost-every $b>0$, but an equivalent definition is given in \cite[Theorem 1.4]{ars-annuli} for \emph{all} $a,b>0$.

We now introduce the quantum pair of pants.
Suppose $(\C, h, z_1,z_2,z_3)$ is an embedding of a sample from $\QS_3$. Let $\Gamma$ be the full-plane $\CLE_\kappa$ on $\C$ independent of $h$. Let 
$\QS_3\otimes \CLE_\kappa$ denote the law of the decorated quantum surface 
$(\C, h,  \Gamma,  z_1,z_2,z_3)/{\sim_\gamma}$. For $1\le i\le 3$, let  $\eta_i$ be the outermost loop in $\Gamma$ separating $z_i$ from the other two points. Namely,
for each loop $\eta$ separating $z_1$ and $\{z_2,z_3\}$, let $D_\eta$ be the simply connected component of $\hat \C\setminus \eta$
containing $z_1$. (The domain $D_\eta$ could contain $\infty$.) There exists a loop $\eta_1$ with the largest $D_\eta$, which we call the \emph{outermost loop separating} $z_1$ and $z_2,z_3$. For $i=2,3$, we similarly let $\eta_i$ be the outermost loop in $\Gamma$ separating $z_i$ and $\{z_1,z_2,z_3\}\setminus \{z_i \}$.

Let $\frk l_i$ be the quantum boundary length of  $\eta_i$. In light of \cite[Lemma 5.1]{ACSW24a}, we similarly conclude that $\QS_3\otimes \CLE_\kappa[(\frk l_1,\frk l_2,\frk l_3)\in E]=0$ for any Borel set $E\subset \R^3$ of zero Lebesgue measure. This ensures the existence of the disintegration over $\frk l_1,\frk l_2,\frk l_3$,  see \cite[Lemma 2.2]{ACSW24a} for details.
Then we define the quantum pair of pants in the same way as for the quantum annulus.
\begin{definition}[Quantum pair of pants]\label{def-QP}
	Let $\{\QS_3\otimes \CLE_\kappa(\ell_1,\ell_2,\ell_3):(\ell_1,\ell_2,\ell_3)\in (0,\infty)^3 \}$ 
	be the disintegration of $\QS_3\otimes\CLE_\kappa$ over $(\frk l_1,\frk l_2,\frk l_3)$. 
	Let $P_{\eta_1,\eta_2,\eta_3}$ be the connected component of $\hat\C\setminus (\cup_{i=1}^3 \eta_i)$ which is not simply-connected.  For $\ell_1,\ell_2,\ell_3>0$, 
	let $\wt \QP(\ell_1,\ell_2,\ell_3)$ be the law of the quantum surface  $(P_{\eta_1,\eta_2,\eta_3}, h)/{\sim_\gamma}$ under  $\QS_3\otimes \CLE_\kappa(\ell_1,\ell_2,\ell_3)$.  Let 
	$$\QP(\ell_1,\ell_2,\ell_3)    = \left(\prod_{i=1}^3 \ell_i |\QD_{1,0}(\ell_i)| \right)^{-1} \wt \QP(\ell_1,\ell_2,\ell_3).$$
	We call a sample from $\QP(\ell_1,\ell_2,\ell_3)$ a \emph{quantum pair of pants}.
\end{definition}

Similarly to $\QA(a,b)$, Definition~\ref{def-QP} only defines  $\QP(\ell_1,\ell_2,\ell_3)$ for Lebesgue-almost-every triple $(\ell_1, \ell_2, \ell_3) \in (0,\infty)^3$. In practice we will always integrate over all triples in $(0,\infty)^3$, so Definition~\ref{def-QP} is sufficient for our purposes, and we will omit the clause ``for almost every $\ell_1,\ell_2,\ell_3$''.

\subsection{Conformal weldings for \texorpdfstring{$\QA$}{ } and \texorpdfstring{$\QP$}{g}}\label{subsec:3-pt-welding}
We now give three conformal welding results (Proposition~\ref{prop-QA2-pt}, Theorem~\ref{thm-weld-QA-3}, Lemma~\ref{lem-QP-weld-disk})  involving the quantum annulus and quantum pair of pants. 
Recall the notion of uniform conformal welding from Section~\ref{subsec:weld-disk}.  For a sample from $\QA(a,b)\times\QD_{1,0}(b)$, the uniform conformal welding of the boundary components having quantum length $b$ gives a loop-decorated quantum surface with one marked point; denote its law by  $\mathrm{Weld}(\QA(a,b),\QD_{1,0}(b))$.

\begin{proposition}\label{prop-QA2-pt}
	For $a>0$, let $(D,h,\Gamma,z)$ be an embedding of a sample from $\QD_{1,0}(a)\otimes\CLE_\kappa$. Let $\eta$ be the outermost loop of $\Gamma$ surrounding $z$.
	Then the law of the decorated quantum surface $(D,h,\eta,z)/{\sim_\gamma}$ equals 
	\(\int_0^\infty  b \mathrm{Weld}(\QA(a,b),\QD_{1,0}(b)) \, db\).
\end{proposition}
\begin{proof}
This result was implicitly obtained in \cite{msw-cle-lqg}; see \cite[Proposition 5.4]{ACSW24a} for details.
\end{proof}

 We now give the analogous statement for the quantum pair of pants. For distinct $\ell_1,\ell_2,\ell_3$,
given quantum surfaces sampled from $\QD_{1,0}(\ell_1) \times \QD_{1,0}(\ell_2)\times \QD_{1,0}(\ell_3) \times \QP(\ell_1,\ell_2,\ell_3)$,
we can use uniform conformal welding to identify boundary components with the same quantum length, producing a quantum surface decorated with three loops and three points.  We denote its law by  
\begin{equation}\label{eq:weld-QP}
	\mathrm{Weld}\Big(\QD_{1,0}(\ell_1) , \QD_{1,0}(\ell_2), \QD_{1,0}(\ell_3), \QP(\ell_1, \ell_2, \ell_3) \Big).
\end{equation}

\begin{theorem}\label{thm-weld-QA-3}
	Sample $(\C, h, \Gamma, z_1, z_2, z_3)/{\sim_\gamma}$ from $\QS_3 \otimes \CLE_\kappa$. For $i=1,2,3$,  let $\eta_i$ be the outermost loop of $\Gamma$ around $z_i$ separating it from the other two points. Then the law of the decorated quantum surface $(\C, h, \eta_1, \eta_2, \eta_3,  z_1,z_2,z_3)/{\sim_\gamma}$ is	
	\eqb \label{eq-pants-welding}
	\iiint_{\R^3_+} \ell_1\ell_2\ell_3\mathrm{Weld}\Big( \QD_{1,0}(\ell_1), \QD_{1,0}(\ell_2), 
	\QD_{1,0}(\ell_3), \QP(\ell_1, \ell_2, \ell_3) \Big) \, d\ell_1 \, d\ell_2 \, d\ell_3.
	\eqe
\end{theorem}
Theorem~\ref{thm-weld-QA-3} follows quickly from Proposition~\ref{prop-M2} below. For $i=1,2,3$,
let $D_{\eta_i}$ and $D^c_{\eta_1}$ be the two connected components of $\C\setminus \eta_i$, where $D_{\eta_i}$ is the component containing $z_i$. 

\begin{proposition}\label{prop-M2}
In the setting of Theorem~\ref{thm-weld-QA-3}, given $(h,\eta_1,\eta_2,\eta_3)$ let $p_3$ be a point sampled from the probability measure proportional to the quantum length measure of $\eta_3$. Then conditioned on $(D^c_{\eta_3}, h,\eta_{1}, \eta_{2}, z_1, z_2,p_3)/{\sim_\gamma}$, the conditional law of $(D_{\eta_3} , h,z_3, p_3)/{\sim_\gamma}$ is $\QD_{1,1}(\mathfrak l_3)^\#$, where $\mathfrak l_3$ is the quantum length of $\eta_3$. 
\end{proposition}

\begin{proof}
Let $\mathbb F$ be a measure on $H^{-1}(\C)$ such that the law of $(\C, h)/{\sim_\gamma}$ is $\QS$
if $h$ is sampled from $\mathbb F$. Sample $(h,\Gamma,\eta,z_1,z_2,z_3)$  from $1_E\prod_{i=1}^{3}\mu_h(d^2z_i)\, \mathbb F(dh)\,\mathrm{Count}_{\Gamma}(d \eta)\CLE_\kappa^\C(d\Gamma)$, where $E$ is the event that $\eta$ separates $z_3$ from $\{z_1,z_2\}$. Let $\mathfrak \ell$ be the quantum length of $\eta$.
Let  $p$ be sampled on $\eta$ from the probability measure proportional to the quantum length of $\eta$.  Let $D_{\eta}$ and $D^c_{\eta}$ be the two components of $\C\setminus \eta$ where $D_{\eta}$ contains $z_3$. 
By \cite[Theorem 7.1]{ACSW24a}, conditioned on $(D^c_{\eta}, h,\Gamma|_{D^c_{\eta}}, z_1, z_2,p)/{\sim_\gamma}$, the conditional law of $(D_{\eta} , h,z_3, p)/{\sim_\gamma}$ is $\QD_{1,1}(\mathfrak l)^\#$.

In this sample space we still define $\eta_1,\eta_2,\eta_3$ in terms of  $\Gamma, z_1,z_2,z_3$ as in Proposition~\ref{prop-M2}. 
Let $F$ be the event that $z_1$ and $z_2$ are surrounded by distinct  outermost loops of $\Gamma|_{D^c_{\eta}}$. 
Then $E\cap F$ occurs if and only if $\eta=\eta_3$.  Therefore under $1_F\cdot 1_E\prod_{i=1}^{3}\mu_h(d^2z_i) \mathbb F(dh)\mathrm{Count}_{\Gamma}(d \eta)\CLE_\kappa^\C(d\Gamma)$ the law of $(\C, h,\Gamma,z_1,z_2,z_3)/{\sim_\gamma}$ is simply $\QS_3\otimes \CLE_\kappa$ as in Proposition~\ref{prop-M2}.
Since $F$ only depends on $(D^c_{\eta}, h,\Gamma|_{D^c_{\eta}}, z_1, z_2,p)/{\sim_\gamma}$, even under the further conditioning on $F$, the conditional law of $(D_{\eta} , h,z_3, p)/{\sim_\gamma}$ is $\QD_{1,1}(\mathfrak l)^\#$. Since under this conditioning $\eta=\eta_3$, we conclude the proof.
\end{proof}

\begin{proof}[Proof of Theorem~\ref{thm-weld-QA-3}]
Note that Proposition~\ref{prop-M2} still holds if $D_{\eta_1}$ is replaced by $D_{\eta_2}$ or $D_{\eta_3}$. 
Recall $\wt \QP$ from the disintegration in Definition~\ref{def-QP} . By repeatedly applying Proposition~\ref{prop-M2}, we see that 
the law of  $(\C, \phi, \eta_1, \eta_2, \eta_3, z_1,z_2,z_3)/{\sim_\gamma}$ is	
\[	
\iiint_{\R^3_+} \mathrm{Weld}\Big( \QD_{1,0}(\ell_1)^{\#}, \QD_{1,0}(\ell_2)^{\#}, 
\QD_{1,0}(\ell_3)^{\#}, 
\wt\QP(\ell_1, \ell_2, \ell_3) \Big) \, d\ell_1 \, d\ell_2 \, d\ell_3.
\]
Now Theorem~\ref{thm-weld-QA-3} follows from $\QD_{1,0}(\ell) =|\QD_{1,0}(\ell)| \QD_{1,0}(\ell)^{\#}$.
\end{proof}

Finally, we give a conformal welding result we need in Section~\ref{subsec-QP-law}.
 \begin{lemma}\label{lem-QP-weld-disk}
For $a>0$, sample  $(D,h,\Gamma,\eta, \hat \eta,z)$ as follows: sample $(D,h,\Gamma,z)/{\sim_\gamma}$ from 
$\QD_{1,0}(a)\otimes\CLE_\kappa$ and sample $\hat \eta$  from the counting measure on the set of outermost loops of $\Gamma$ except the loop $\eta$ surrounding $z$. 
 Then for some constant $C = C(\gamma)$ the law of $(D, h,\eta,  \hat \eta, z)/{\sim_\gamma}$ is
 \[C \iint_{\R^2_+} bc \mathrm{Weld}\left(\QD_{1,0}(b), \QD(c), \QP(a,b,c)\right) \, db \, dc.\]
 \end{lemma}
 \begin{proof}
We retain the notations in the proof of Proposition~\ref{prop-M2}. Sample $(h, \Gamma, \eta_1, \eta_2, \eta_3, z_1, z_2)$ from $\mathsf M := 1_G \mu_h(d^2z_1)\,\mu_h(d^2z_2)\, \mathbb F(dh) \,\mathrm{Count}_\Gamma(d\eta_1)\,\mathrm{Count}_\Gamma(d\eta_2)\,\mathrm{Count}_\Gamma(d\eta_3)\CLE_\kappa^\C(d\Gamma)$, where $G$ is the event that $\eta_1, \eta_2, \eta_3$ are distinct, $\eta_1, \eta_2, z_1, z_2$ all lie on the same side of $\eta_3$, and in $\C \backslash \eta_3$ the outermost loops around $z_1, z_2$ are $\eta_1, \eta_2$ respectively. 
We claim that the $\mathsf M$-law of $(\C, \eta_1, \eta_2, \eta_3, z_1, z_2)/{\sim_\gamma}$ is
\eqb\label{eq-disk-QP-1}
\iiint_{\R_+^3} abc \mathrm{Weld}(\QD_{1,0}(a), \QD_{1,0}(b), \QD(c), \QP(a,b,c)) \, da\, db\, dc. 
\eqe
Indeed, by Theorem~\ref{thm-weld-QA-3}, if we sample $(h, \Gamma, \eta_1, \eta_2, \eta_3, z_1, z_2)$ from $\mu_h|_{D_3}(d^2z_3) \mathsf M$ where $D_3$ is the component of $\C \backslash \eta_3$ not containing $z_1, z_2$, then the law of $(\C, \eta_1, \eta_2, \eta_3, z_1, z_2,z_3)/{\sim_\gamma}$ is~\eqref{eq-pants-welding}. Forgetting the point $z_3$ and unweighting by $\mu_h(D_3)$, we get~\eqref{eq-disk-QP-1}.

On the other hand, by Proposition~\ref{prop-tabula-rasa} applied to the loop $\eta_1$ and the two-pointed quantum sphere $(\C, h, z_1, z_2)/{\sim_\gamma}$, we see that the $\mathsf M$-law of $(\C, \eta_1, \eta_2, \eta_3, z_1, z_2)/{\sim_\gamma}$ is 
$C\int_0^\infty a \mathrm{Weld}(\QD_{1,0}(a), \widetilde{ \mathsf{M}}(a)) \, da$,
where $\wt{ \mathsf{M}}(a)$ is the measure described in Lemma~\ref{lem-QP-weld-disk}. Comparing this to~\eqref{eq-disk-QP-1}, we conclude that $C \wt{\mathsf M}(a) = \iint_{\R_+^2} bc \mathrm{Weld}(\QD_{1,0}(b), \QD(c), \QP(a,b,c)\, db \, dc$ as desired.
 \end{proof}

\subsection{Area distribution for the quantum annulus}\label{subsec-QA-law}
The goal of this section is to prove the following. 

\begin{theorem}\label{thm-QA2-area}
	For $a>0,b>0$, let $A$ be the quantum area of a sample from $\QA(a, b)$. Then
\begin{equation*}
\QA(a, b)[e^{-\mu A}]= \frac{\cos(\pi (\frac4{\gamma^2}-1))}\pi  \cdot \frac {e^{-(a+b)\sqrt{\mu/\sin(\pi\gamma^2/4)}} }{\sqrt{ab} (a+b)} \quad \textrm{for }\mu\ge 0.
\end{equation*}
\end{theorem}

Theorem~\ref{thm-QA2-area} itself will not  be directly used in the rest of the paper, but we find it interesting in its own right, and its proof can serve as a warm up for its counterpart for $\QP$. We first recall the total mass of $\QA(a,b)$ derived in \cite{ACSW24a} via Propositions~\ref{prop-ccm} and~\ref{prop-QA2-pt}.

\begin{proposition}[{\cite[Proposition 5.7]{ACSW24a}}]\label{prop-QA-partition}
	For $a>0,b>0$, we have 
	\[\left|\QA(a,b)\right|=  \frac{\cos(\pi (\frac4{\gamma^2}-1))}{\pi\sqrt{ab} (a+b)}.\]
\end{proposition}

Proposition~\ref{prop-QA-partition} was proved in~\cite{ACSW24a} using the following fact on Levy process.  Recall the stable L\'evy process $\P^\beta$ from Proposition~\ref{prop-ccm} with $\beta = \frac4{\gamma^2} + \frac12$.
\begin{lemma}[{\cite[Lemma 5.10]{ACSW24a}}]\label{lem:palm}
For $(\zeta_t)$ sampled from $\P^\beta$, let $J$ be the set of jumps and $J_a$ be the jumps before hitting $-a$; see the beginning of Section~\ref{sec:excursion}. Let $(\bfb, \bft)\in J$ be sampled from the counting measure on $J$, and denote the joint law of $(\zeta,\bfb,\bft)$ by $M^\beta$. Then the joint law of $(\bfb, \bft)$ and $(\wt\zeta_t)=(\zeta_t-1_{t\ge\bft}\bfb)$ is $(1_{b>0,t>0} b^{-\beta-1} \, db \, dt) \times \P^\lexp$.
\end{lemma}
We now use Lemma~\ref{lem:palm} to give a description of the quantum lengths of all the outermost loops in $\Gamma$ except $\eta$  and the laws of the quantum surfaces surrounded by these loops. 
\begin{proposition}\label{prop-QA2-Levy}
For $a>0$, let $(D,h,\Gamma,z)/{\sim_\gamma}$ be a sample of $\QD_{1,0}(a)\otimes\CLE_\kappa$. 
Let $(\ell_i)_{i\ge 1}$ be the quantum lengths of the outermost loops in $\Gamma$ 
and $\ell_h(\eta)$ the quantum length of the  loop $\eta$ surrounding $z$.  
Then conditioned on $\ell_h(\eta)$, the conditional law of $\{ \ell_i: \ell_i\neq \ell_h(\eta), i\ge 1 \}$ 
equals the $\P^\lexp$-law of  $\{x: (x,t)\in  J_{a+\ell_h(\eta)} \textrm{ for some }t  \}$.  Moreover, conditioned on $(\ell_i)_{i \geq 1}$, the quantum surfaces surrounded by these outermost loops apart from $\eta$ are independent quantum disks with boundary lengths $\{ \ell_i : \ell_i \neq \ell_h(\eta), i \geq 1\}$.
\end{proposition}
\begin{proof}
In the setting of Lemma~\ref{lem:palm}, with $\wt\tau_{-\ell} = \inf\{t: \wt\zeta_t=-\ell  \}$, we have $\tau_{-a}= \wt \tau_{-a-\bfb}$.
Let $\wt M^\beta_a = \frac{\tau_{-a}^{-1}}{\E^\lexp[\tau_{-a}^{-1}]}M^\beta 1_{\bft\le\tau_{-a}}$; here the expectation $\E^\beta$ is taken with respect to $\P^\beta$. Then under $\wt M^\beta_a$, the conditional law of $J_a\setminus \{ (\bfb, \bft)\}$ conditioning on $\bfb$ is the $\P^\beta$-law of $J_{a+\bfb}$.
By Proposition~\ref{prop-ccm}, the law of $(\ell_i)_{i\ge1}$ equals $J_a$ under $\wt M_a^\beta$.
Note that $(D,h,\Gamma,z,\eta)/{\sim_\gamma}$ can be obtained by first sampling $(D,h,\Gamma)/{\sim_\gamma}$ from $\QD(a)\otimes\CLE_\kappa$, then sampling $\eta$ from the counting measure on $\Gamma$ and sampling $z$ from the restriction of  $\mu_h(d^2z)$ on the region surrounded by $\eta$. 
This gives  the first assertion.
The second assertion is implicitly proved in \cite{msw-cle-lqg}; see \cite[Proposition 4.2]{ACSW24a} for a detailed justification.
\end{proof}
\begin{proof}[Proof of Theorem~\ref{thm-QA2-area}]
Since $\QA(a,b)=|\QA(a,b)| \QA(a,b)^{\#}$ and Proposition~\ref{prop-QA-partition} gives $|\QA(a,b)|$, it remains to compute $\QA^{\#}(a,b)[e^{-\mu A}]$. Let  $(\zeta_t)_{t\ge 0}$ be sampled from $\P^\beta$. For $a,b>0$, let  $ (x_i)_{i \geq 1}$ be the collection of sizes of the jumps of $\zeta$ on $[0,\tau_{-a-b}]$ sorted in decreasing order. 
	Independently of $(x_i)_{i\ge 1}$, let $(A_i)_{i\ge 1}$ be independent copies of the quantum area of a sample from 
	$\QD(1)^\#$. Then  the quantum area of a sample from $\QA(a,b)^{\#}$ has the law of 
	$\sum x_i^2 A_i$.
	This follows from Proposition~\ref{prop-QA2-Levy} and the fact that the quantum area of a sample from $\QD(x)^\#$ agrees in law with $x^2 A$ where $A$ is the quantum area of a sample from $\QD(1)^\#$, see e.g.\ Proposition~\ref{prop-QD-area}. 
Now by Proposition~\ref{l}, we have $ \QA(a, b)^{\#} [e^{-\mu A}] =   e^{-(a+b) \sqrt{\mu/\sin(\pi\gamma^2/4)}}.$ 
By Proposition~\ref{prop-QA-partition} we conclude the proof.
\end{proof}
We record the following corollary of Proposition~\ref{prop-QA2-Levy}, which will be used in the next subsection.
\begin{lemma}\label{lem:QA-QP}
Let $(D,h,\Gamma,z)/{\sim_\gamma}$ as in Proposition~\ref{prop-QA2-Levy}. 
Let $\hat\eta$ be chosen from the counting measure on the outermost loops of $\Gamma$ except $\eta$. 
Let $(\ell_i)_{i\ge 1}$ be the lengths of the outermost loops of $\Gamma$ except $\eta$ and $\hat\eta$, ranked in decreasing order. 
Then for $b,c>0$, the disintegration of the law of $\{ \ell_i: i\ge 1 \}$ over $\{\ell_h(\eta)=b,\ell_h(\hat \eta)=c\}$ is the law of 
$\{x: (x,t)\in  J_{a+b+c}  \textrm{ for some }t  \}$ under
\[
b|\QA (a,b) | |\QD_{1,0}(b)|   c^{-\beta-1} \tau_{-a-b}  \P^\beta(d\zeta).
\]
\end{lemma}
\begin{proof}
By Propositions~\ref{prop-QA2-pt} and~\ref{prop-QA2-Levy}, the disintegration of the law of $\{ \ell_i: \ell_i\neq \ell_h(\eta), i\ge 1 \}$  over $\{\ell_h(\eta)=b\}$  is 
the law of $\{x: (x,t)\in  J_{a+b}  \textrm{ for some }t  \}$ under $b|\QA (a,b) | |\QD_{1,0}(b)|  \P^\beta(d\zeta)$.
Since we further remove the loop $\hat\eta$ chosen from counting measure, by Lemma~\ref{lem:palm}, we get the additional factor $c^{-\beta-1} \tau_{-a-b}$ and the change from $J_{a+b}$ to $J_{a+b+c}$ for the jump set.
\end{proof}

\subsection{Area and length distributions for the quantum pair of pants}\label{subsec-QP-law}
The goal of this section is to establish the following. 

\begin{theorem}\label{thm-QA3-laplace}
For $\ell_1, \ell_2, \ell_3>0$. let $A$ be the total quantum area of a sample from $\QP(\ell_1, \ell_2, \ell_3)$.
Then there is a  constant $C\in (0,\infty)$ only depending on $\gamma$ such that
\begin{equation*}
\QP(\ell_1, \ell_2, \ell_3) [e^{-\mu A}] = C\mu^{\frac14 - \frac2{\gamma^2}}\frac{1}{\sqrt{\ell_1\ell_2\ell_3}} e^{-(\ell_1+ \ell_2+ \ell_3)\sqrt{\mu/\sin(\pi \gamma^2/4)}} \quad \textrm{for }\mu>0.
\end{equation*}
\end{theorem}

Sending $\mu$ to $0$ in  Theorem~\ref{thm-QA3-laplace},  we see that   $\QP(\ell_1, \ell_2, \ell_3)$ is an infinite measure.
However, since $\QP(\ell_1, \ell_2, \ell_3)[A<x] < \infty$ for $x>0$, it is still $\sigma$-finite.

The proof of Theorem~\ref{thm-QA3-laplace} is based on the L\'evy process description of CLE loop lengths in Proposition~\ref{prop-ccm} (and the  L\'evy excursion theory developed in Section \ref{sec:excursion}), the conformal welding result Lemma~\ref{lem-QP-weld-disk}, and the joint area and length law of $\QD_{1,0}$ given in Theorem~\ref{thm-FZZ}. 

\begin{lemma}\label{lem-QP-translate}
For $\mu>0$, $x>0$, let $g(x)= \E[e^{-\mu x^2A_1}]$ where $A_1$ is the quantum area of a  sample from $\QD(1)^{\#}$.  	Then there is a constant $C\in (0,\infty)$ such  that for   $a,b,c>0$, we have
	\begin{equation}\label{eq:QP-Levy}
\QP(a, b, c) [e^{-\mu A}]=   
C|\QA(a,b)|c^{-1/2} \E^\beta[\tau_{-a-b} \prod_{i\ge 1}g(x_i)],	
	\end{equation}
where $\E^\beta$ is taken over $\P^\beta$ and $(x_i)_{i\ge 1}$ is the set of jumps that occur before time $\tau_{-a-b-c}$.
\end{lemma}
\begin{proof}
Sample $(D,h,\Gamma,z)/{\sim_\gamma}$ from 
$\QD_{1,0}(a)\otimes\CLE_\kappa$ and sample $\hat \eta$  from the counting measure on the outermost loops of $\Gamma$ except the loop $\eta$ surrounding $z$. Let $A_{\eta,\hat \eta}$ be the quantum area of the region bounded by $\eta,\eta'$ and $\partial D$.  Lemma~\ref{lem:QA-QP} gives the law of the quantum lengths of the outermost loops.	By Proposition~\ref{prop-QA2-Levy}, conditioning on the lengths,  the quantum surfaces surrounded by the outermost loops are independent quantum disks. 
Therefore the integral of $e^{-\mu A_{\eta,\hat \eta}}$ over this sample space is 
\begin{equation}\label{eq:cal1}
\iint_{\R_+^2} b|\QA (a,b) | |\QD_{1,0}(b)|   c^{-\beta-1} \E[\tau_{-a-b} \prod_{i\ge 1} g(x_i)]\, db \, dc.
\end{equation}
By the conformal welding result Lemma~\ref{lem-QP-weld-disk}, with the constant $C$ in Lemma~\ref{lem-QP-weld-disk} we have
\[
 \eqref{eq:cal1} = C\iint_{\R_+^2} bc |\QD_{1,0}(b)| |\QD(c)| \QP(a,b,c)[e^{-\mu A}]\, db \, dc
\]
Since $|\QD(c)| \propto c^{-\beta-\frac32}$ by Lemma~\ref{rem-QD-length}, we  get~\eqref{eq:QP-Levy} after disintegrating over $b,c$.
\end{proof}

We now compute $\E^\beta[\tau_{-a-b} \prod_{i\ge 1}g(x_i)]$ using the excursion theory for the Levy process $\zeta$,  as reviewed in Section \ref{sec:excursion}. Let $I_t= \inf\{\zeta_t: s\in [0,t] \}$, then $(\zeta_t-I_t)_{t\ge 0}$ is a Markov process and we can consider its excursions away from 0. We write $\ul N$ for the excursion measure.

\begin{lemma}\label{lem:Palm-Levy}
Given $\E^\beta$,  $\tau$, $\mu$, and $g$  in Lemma~\ref{lem-QP-translate}, let  $C(\mu)= \int T(e) \prod_{x\in J_e} g(x) \ul N (d e)$.  Then
\[
\E^\beta[\tau_{-a-b} \prod_{i\ge 1}g(x_i)]=
C(\mu)(a+b)e^{-(a+b+c) \sqrt{\mu/\sin(\pi\gamma^2/4)}}.
\]
\end{lemma}
\begin{proof}
	 Recall that $\{x_i: i\ge 1\}$ is the set of jumps in $\zeta$ before $\tau_{-a-b-c}$. As stated in Section \ref{sec:excursion}, the excursions of $(\zeta_t-I_t)_{t\in  [0,\tau_{-a-b-c}]}$ can be viewed as a Poisson point process $\mathrm{Exc}_{a+b+c}=\{ (e,s): s\le a+b+c \}$ with intensity measure $\ul N\times 1_{s\in[0,a+b+c] }  ds$. For $(e,s)\in \mathrm{Exc}_{a+b+c}$, let $\sigma$ be the time when $(\zeta_t - I_t)_{ t\ge 0}$ starts the excursion corresponding to $e$. 
	Then $s = -I_\sigma$. Let $T(e)$ be the duration $T$ of $e$.
	Then $\tau_{-a-b-c}=\sum_{ (e,s)\in \mathrm{Exc}_{a+b+c}}  T(e)$ and 
 \(	\tau_{-a-b} = \sum_{ (e,s)\in \mathrm{Exc}_{a+b+c}}  T(e)1_{s\le a+b}\).
 
Let $\cF_{a+b+c}$ be the sigma algebra generated by $\{ e: (e,s)\in \mathrm{Exc}_{a+b+c} \}$ and denote the conditional expectation over $\cF_{a+b+c}$ by $\E_{a+b+c}$. As a general property of Poisson point processes, conditioning on $\cF_{a+b+c}$,
the conditional law of 	the time set $\{ s: (e,s)\in \mathrm{Exc}_{a+b+c} \}$ is given by a collection of independent uniform  random variables in $(0,a+b+c)$. Therefore $\E_{a+b+c}[1_{s\le a+b}]= \frac{a+b}{a+b+c}$. Then 
		\begin{equation}\label{eq:cond-exc}
		\E_{a+b+c}[\tau_{-a-b}] = \sum_{ (e,s)\in \mathrm{Exc}_{a+b+c}}  T(e)\E_{a+b+c}[1_{s\le a+b}]= \frac{a+b}{a+b+c} \tau_{-a-b-c}.
		\end{equation}	
	Let $F(e)=\prod_{x\in J_e} g(x) $ so that $\prod_{i\ge 1}g(x_i)= \prod_{e\in \mathrm{Exc}_{a+b+c} }F(e)$.
	 Since $\prod_{i\ge 1}g(x_i) $ is measurable with respect to $\cF_{a+b+c}$, we have 
	\[
	\E_{a+b+c}\big[\tau_{-a-b} \prod_{i\ge 1}g(x_i) \big]= \E_{a+b+c}\big[\tau_{-a-b} \prod_{e\in \mathrm{Exc}_{a+b+c} }F(e) \big]=
	\frac{a+b}{a+b+c} \tau_{-a-b-c}  \prod_{e\in \mathrm{Exc}_{a+b+c} }F(e).
	\]
		Further taking the expectation over $\cF_{a+b+c}$ we have 
		 \[
		 \E^\beta\big[\tau_{-a-b} \prod_{i\ge 1}g(x_i) \big]= \frac{a+b}{a+b+c} \E^{\beta} \big[\tau_{-a-b-c}  \prod_{e\in \mathrm{Exc}_{a+b+c} }F(e)\big].
		 \]
	 Since $\tau_{-a-b-c}=\sum_{e\in \mathrm{Exc}_{a+b+c} } T(e)$ we have
	\[
	\E^{\beta}  \big[\tau_{-a-b-c}  \prod_{e\in \mathrm{Exc}_{a+b+c} }F(e) \big]=\E^{\beta}  \big[\big(\sum_{e\in \mathrm{Exc}_{a+b+c} } T(e) \big) \prod_{e\in \mathrm{Exc}_{a+b+c} }F(e) \big].
	\]
By the first equation in Lemma~\ref{palm}, we have 
	\[
 \E^{\beta}  \big[\big(\sum_{e\in \mathrm{Exc}_{a+b+c} } T(e) \big) \prod_{e\in \mathrm{Exc}_{a+b+c} }F(e) \big]=(a+b+c)\E^{\beta}\big[ \prod_{e\in \mathrm{Exc}_{a+b+c} }F(e)\big] \int T(e) F(e) \ul N (d e).
\] 
	Therefore \(\E^{\beta}[\tau_{-a-b} \prod_{i\ge 1}g(x_i)]= C(\mu)(a+b) \E^{\beta}\big[\prod_{i\ge 1} g(x_i)\big]\). By Proposition~\ref{l},
we have	\(	\E^{\beta}\big[\prod_{i\ge 1} g(x_i)\big]= e^{-(a+b+c) \sqrt{\mu/\sin(\pi\gamma^2/4)}}\), which concludes the proof.
\end{proof}

\begin{proof}[Proof of Theorem~\ref{thm-QA3-laplace}]
	We write $X\propto_\gamma Y$ if $X=C_\gamma Y$  for some $\gamma$-dependent constant $C_\gamma\in (0,\infty)$. 
	By  Lemmas~\ref{lem-QP-translate} and~\ref{lem:Palm-Levy}, and expression for $|\QA(a,b)|$ in Proposition~\ref{prop-QA-partition},  we have 
	\begin{equation}\label{eq:QP-area-C}
		\QP(\ell_1, \ell_2, \ell_3) [e^{-\mu A}] \propto_\gamma \frac{C(\mu)}{\sqrt{\ell_1\ell_2\ell_3}} e^{-(\ell_1+ \ell_2+ \ell_3)\sqrt{\mu/\sin(\pi \gamma^2/4)}}
	\end{equation}
	
	It remains to show that $C(\mu) \propto_{\gamma} \mu^{\frac14 - \frac2{\gamma^2}}$. 
	Recall the setting of Theorem~\ref{thm-weld-QA-3}, where $(\C,h,\Gamma,z_1,z_2,z_3)/{\sim_\gamma}$ has the law of 
	$\QS_3 \otimes \CLE_\kappa$. By the welding equation~\eqref{eq-pants-welding}, the expression for $\QP(\ell_1, \ell_2, \ell_3) [e^{-\mu A}]$  in~\eqref{eq:QP-area-C}, and the areas of quantum disks (Theorem~\ref{thm-FZZ}), we have
	\[
	\QS_3[e^{-\mu \mu_h (\C) }] \propto_\gamma C(\mu) \prod_{i=1}^3 \int_0^\infty   \mu^{\frac1\gamma(Q-\gamma)} K_{\frac2\gamma(Q-\gamma)}\left(\ell_i \sqrt{\frac\mu{\sin (\pi\gamma^2/4)}}\right) \cdot  \frac{e^{-\ell_i \sqrt{\mu/\sin(\pi\gamma^2/4)}}}{\sqrt{\ell_i}} \, d\ell_i.
	\]
	By Lemma~\ref{lem:integration-fun} we have \(\int_0^\infty  y^{-\frac12}e^{-cy} K_\nu(cy) \, dy \propto_\gamma  c^{-\frac12} 
	\) 	
	for $c =\sqrt{\frac\mu{\sin (\pi\gamma^2/4)}}$ and $\nu =\frac2\gamma(Q-\gamma)$. Therefore
	$\QS_3[e^{-\mu \mu_h (\C) }] \propto_\gamma C(\mu) \left(\mu^{\frac1\gamma (Q-\gamma)-\frac14}\right)^3$.
	By Theorems~\ref{prop-DOZZ} and~\ref{def-QS-2}, $\QS_3[e^{-\mu \mu_h (\C) }] \propto_\gamma\mu^{\frac{2Q - 3\gamma}{\gamma}}$. Since \(\frac{2Q - 3\gamma}{\gamma}-
	3(\frac1\gamma (Q-\gamma)-\frac14) =\frac14 - \frac2{\gamma^2} \), we have $C(\mu) \propto_\gamma \mu^{\frac14 - \frac2{\gamma^2}}$. 
	\end{proof}

\section{The three-point  function for the nesting loop statistics}\label{sec:3-pt}
In this section, we prove  Theorem~\ref{thm:nesting} following the strategy in Section~\ref{subsec:method}. In Section~\ref{sec:loop-counting}, 
we reduce the nesting loop statistics to the the joint moment of conformal radius, which is exactly solvable as stated in Theorem~\ref{thm:conformal radius}.  The rest of this section is devoted to proving Theorem~\ref{thm:conformal radius}. The case  $\kappa\in (\frac{8}{3},4]$ is done in Sections~\ref{subsec-general-QP} and~\ref{subsec:area-matching} using the quantum pair of pants from  Section~\ref{sec-welding}. In Section~\ref{sec:cr-non-simple}, we introduce the generalized quantum pair of pants and prove the $\kappa\in (4,8)$ case.

\subsection{From conformal radius to loop counting: proof of Theorem~\ref{thm:nesting}}\label{sec:loop-counting}

The following lemma  relates the nesting loop statistic to the conformal radius. 
For a loop $\eta$,  let $\CR(\eta,z)$ be the conformal radius of the connected component of $\C\backslash\eta$ containing $z$, viewed from $z$.
\begin{lemma}\label{lem-trunc}
Fix $\kappa\in(\frac{8}{3},8)$ and $n = 2 \cos (\pi (1-\frac4\kappa))$.	Suppose $\eta_0$ is a loop surrounding $0$. 
Sample a $\CLE_\kappa$ in the domain surrounded by $\eta_0$ and let $S = \{ \eta_0\} \cup \{ \CLE_\kappa \text{ loops surrounding }0\}$. Let $X_\eps$ (resp. $\wt X_\eps$) be the number of loops $\eta \in S$ such that $\mathrm{dist}(\eta, 0) > \eps$ (resp.\ $\mathrm{CR}(\eta, 0) > \eps$).
For $n'\in (0,2]$, let $\Delta=\frac{\kappa}{8\pi^2}\arccos^2(\frac{n'}{2})-\frac{1}{2}(\frac{2}{\sqrt{\kappa}}-\frac{\sqrt{\kappa}}{2})^2$. 
Then there exists $c(\Delta)$ and  $\wt c(\Delta)$ such that
	\[\lim_{\eps \to 0} \E[\eps^{-\Delta} ( n'/n )^{X_\eps}] =  c(\Delta) \mathrm{CR}(\eta_0, 0)^{-\Delta} \quad \textrm{and}\qquad \lim_{\eps \to 0} \E[\eps^{-\Delta} (n'/n)^{\wt X_\eps}] = \wt c(\Delta) \mathrm{CR}(\eta_0, 0)^{-\Delta}.\]
	Moreover, $\E[(\eps/\mathrm{CR}(\eta_0,0))^{-\Delta} ( n'/n )^{X_\eps}]$ (resp.\  $\E[(\eps/\mathrm{CR}(\eta_0,0))^{-\Delta} ( n'/n )^{\wt X_\eps}]$) is uniformly bounded over all $\eta_0$ and $\eps < \mathrm{dist}(\eta_0, 0)$  (resp.\ $\eps < \mathrm{CR}(\eta_0,0)$). 
\end{lemma}
The assertion for $\wt X_\eps$ is \cite[Lemma 1.7]{holden2022liouville}. 
The proof for the assertion  for $X_\eps$ requires some additional technical work which we defer
to the end of this section. From Lemma~\ref{lem-trunc}, we can express  $\lim_{\eps \to 0} \E \left[ \prod_{i=1}^3 \eps^{-\Delta_i} (n_i / n)^{X_\eps (z_i)} \right]$ using the joint moment of the conformal radii of $(\eta_i)_{1\le i\le 3}$, where $\eta_i$ is the outermost CLE loop separating $z_i$ and other two $z_j$'s as defined above Definition~\ref{def-QP}. 
We now define the three-point function for conformal radii.

\begin{lemma}\label{lem-main-confinv}
There exists a function $\CCLE \in (0, \infty]$ such that 
\begin{equation}\label{def-structure-CLE}
\E[\prod_{i=1}^3\CR( \eta_i, z_i)^{\lambda_i}] = \CCLE \prod_{i=1}^3 |z_i - z_{i+1}|^{\lambda_i + \lambda_{i+1} - \lambda_{i+2}},
\end{equation}
where we identify $z_j$ with $z_{j-3}$ and $\lambda_j$ with $\lambda_{j-3}$.
\end{lemma}
\begin{proof}
This follows from the conformal invariance of the full plane $\CLE_\kappa$, the scaling property of conformal radius, and the fact that any triple of points on the plane are conformally equivalent.
\end{proof}

Now Theorem~\ref{thm:nesting} will follow from the explicit formula of $\CCLE$.
\begin{theorem}\label{thm:conformal radius}
For $\kappa\in (\frac83,8)$, we have $\CCLE < \infty$ if and only if $\lambda_1, \lambda_2, \lambda_3 > \frac{3\kappa}{32} -1+\frac2\kappa$.
Set $\gamma=\sqrt{\kappa}$ for $\kappa\in(\frac{8}{3},4]$ and $\gamma=\frac{4}{\sqrt{\kappa}}$ for $\kappa\in(4,8)$. Set $Q=\frac{\gamma}2+\frac2\gamma$.
Let $\alpha_i$ be a root of $\alpha(Q-\frac{\alpha}2) - 2 =\lambda_i$ for $i=1,2,3$. Then 
\begin{equation}\label{eq:3ptCLE}
\CCLE =  \frac {\prod_{i=1}^3 N_\gamma(\alpha_i)}{C^\mathrm{DOZZ}_\gamma(\alpha_1, \alpha_2, \alpha_3)},
\end{equation}
where for $\kappa\in(\frac{8}{3},4]$
\eqb\label{eq-Ngamma}
 N_\gamma(\alpha)= \left( -\pi\cos(\frac{4\pi}{\gamma^2}) \frac{\Gamma(\frac4{\gamma^2}-1)}{\Gamma(1-\frac{\gamma^2}4)} C_\gamma^\mathrm{DOZZ}(\gamma, \gamma, \gamma)^{1/3}\right)
\frac{\Gamma(\frac\gamma2(\alpha - \frac\gamma2))}{\Gamma(\frac2\gamma(Q-\alpha)) \cos(\frac{2\pi}\gamma(Q-\alpha))} \left( \frac{\pi \Gamma(\frac{\gamma^2}4)}{\Gamma(1-\frac{\gamma^2}4)} \right)^{-\frac\alpha\gamma}.
\eqe
and for $\kappa\in(4,8)$,
\eqb\label{eq-N'gamma} 
 N_\gamma(\alpha)= \left( -\pi\cos(\frac{\gamma^2\pi}{4}) \frac{\Gamma(\frac{\gamma^2}{4})}{\Gamma(2-\frac{4}{\gamma^2})} C_\gamma^\mathrm{DOZZ}(\gamma, \gamma, \gamma)^{1/3}\right)
\frac{\Gamma(\frac2\gamma(\alpha - \frac2\gamma))}{\Gamma(\frac\gamma2(Q-\alpha)) \cos(\frac{\gamma\pi}2(Q-\alpha))} \left( \frac{\pi \Gamma(\frac{\gamma^2}4)}{\Gamma(1-\frac{\gamma^2}4)} \right)^{-\frac\alpha\gamma}.
\eqe
\end{theorem}

Sections~\ref{subsec-general-QP} and~\ref{subsec:area-matching} are devoted to proving Theorem~\ref{thm:conformal radius}. 
Theorem~\ref{thm:nesting} is now a corollary of Theorem~\ref{thm:conformal radius} and the following variant of Lemma~\ref{lem-trunc}.

\begin{lemma}\label{lem-main-limit}
Let $c(\Delta)$ be the function defined in Lemma~\ref{lem-trunc} and $\eta_i$'s be as just above. Then
\[ \lim_{\eps \to 0} \E \left[ \prod_{i=1}^3 \eps^{-\Delta_i} (n_i / n)^{X_\eps (z_i)} \right]= c(\Delta_1) c(\Delta_2) c(\Delta_3) \E \left[ \prod_{i=1}^3 \mathrm{CR}(\eta_i, z_i)^{-\Delta_i} \right]. \]
\end{lemma}
\begin{proof}
    Let $E_{\eps, i } = \{ \mathrm{dist}(\eta_i, z_i) < \eps \}$ and $E_\eps = \bigcup_i E_{\eps, i}$. By Lemma~\ref{lem-trunc}, on $E_{\eps, i}^c$ the random variable $Y_{\eps, i} := \E[(\eps/{\mathrm{CR}(\eta_i, z_i)})^{-\Delta_i} (n_i/n)^{X_\eps(z_i)} \mid \eta_i ]$ is uniformly bounded by a deterministic constant $C$. Moreover, by Theorem~\ref{thm:conformal radius} $\E [ \prod_{i=1}^3 \mathrm{CR}(\eta_i, z_i)^{-\Delta_i}]$ is finite.
    Thus, Lemma~\ref{lem-trunc} and
	the dominated convergence theorem give
	\alb
	\lim_{\eps \to 0} \E \left[ 1_{E_\eps^c} \prod_{i=1}^3 \eps^{-\Delta_i} (n_i / n)^{X_\eps (z_i)} \right] &=  \E\left[\lim_{\eps \to 0} 1_{E_\eps^c}\prod_{i=1}^3 \mathrm{CR}(\eta_i, z_i)^{-\Delta_i} Y_{\eps,i}\right] = \prod_{i=1}^3 c(\Delta_i) \E \left[ \prod_{i=1}^3 \mathrm{CR}(\eta_i, z_i)^{-\Delta_i} \right].
	\ale
	Now we show that $\E[1_{E_\eps} \prod_{i=1}^3 \eps^{-\Delta_i} (n_i / n)^{X_\eps (z_i)} ]\to 0$. Pick an index $i$. 
    On the event $E_{\eps,i}^c$, we have $\eps^{-\Delta_i} (n_i/n)^{X_\eps(z_i)} = \CR(\eta_i, z_i)^{-\Delta_i} Y_{\eps, i} \leq C \CR(\eta_i, z_i)^{-\Delta_i}$. On the event $E_{\eps,i}$, we have $X_\eps(z_i) = 0$; if $-\Delta_i \geq 0$, then  $\eps^{-\Delta_i} (n_i/n)^{X_\eps(z_i)} = \eps^{-\Delta_i} < 1$, and if instead $-\Delta_i < 0$, since $\eps > \mathrm{dist}(\eta_i,z_i) > \CR(\eta_i, z_i)/4$ by the Koebe quarter theorem, we have  $\eps^{-\Delta_i} (n_i/n)^{X_\eps(z_i)} = (\eps/\CR(\eta_i, z_i))^{-\Delta_i} \CR(\eta_i, z_i)^{-\Delta_i} \leq 4^{\Delta_i} \CR(\eta_i, z_i)^{-\Delta_i}$. Combining all these cases, we conclude that 
    \[\E[ 1_{E_\eps} \eps^{-\Delta_i}(n_i/n)^{X_\eps(z_i)} \mid \eta_1, \eta_2, \eta_3] \leq  1_{E_\eps} (1+ 4^{\Delta_i} + C) (\CR(\eta_i, z_i)^{-\Delta_i} + 1). \]
    
    Multiplying these for $i=1,2,3$ gives, for some constant $C'>0$,
	\[\E[1_{E_\eps} \prod_{i=1}^3 \eps^{-\Delta_i} (n_i / n)^{X_\eps (z_i)}] \leq C'
    \E[1_{E_\eps}\prod_{i=1}^3 (\mathrm{CR}(\eta_i, z_i)^{-\Delta_i} + 1)]. \]
	Since $\E[\prod_{i=1}^3 (\mathrm{CR}(\eta_i, z_i)^{-\Delta_i} + 1)]< \infty$ by Theorem~\ref{thm:conformal radius} and $\lim_{\eps \to 0} 1_{E_\eps} = 0$ a.s., the dominated convergence theorem gives $\lim_{\eps \to 0}\E[1_{E_\eps} \prod_{i=1}^3 \eps^{-\Delta_i} (n_i / n)^{X_\eps (z_i)}] = 0$. Combining this with the first equation gives the desired result.
\end{proof}

\begin{proof}[{Proof of Theorem~\ref{thm:nesting} given Theorem~\ref{thm:conformal radius}}]
    The existence of the limit is shown in Lemma~\ref{lem-main-limit}. 
    In Lemma~\ref{lem-trunc}, since $n' = n$ corresponds to $\Delta = 0$, we have $c(0) $ in Lemma~\ref{lem-main-limit} equals $1$.
    Combining Lemmas~\ref{lem-main-limit}~\ref{lem-main-confinv} and Theorem~\ref{thm:conformal radius} gives $C_3(\Delta_1,\Delta_2,\Delta_3)=c(\Delta_1)c(\Delta_2)c(\Delta_3)\frac {\prod_{i=1}^3 N_\gamma(\alpha_i)}{C^\mathrm{DOZZ}_\gamma(\alpha_1, \alpha_2, \alpha_3)}$. Hence,
    \begin{equation}\label{eq:loop-counting}
        \frac{C_3(\Delta_1, \Delta_2, \Delta_3)}{\sqrt{\prod_{i=1}^3 C_3(\Delta_i, \Delta_i, 0)}} = N_\gamma(\gamma)^{-3/2} \frac{\sqrt{\prod_{i=1}^3 C_\gamma^{\mathrm{DOZZ}}(\alpha_i, \alpha_i, \gamma) }}{C_\gamma^{\mathrm{DOZZ}}(\alpha_1, \alpha_2, \alpha_3)},
    \end{equation}
     and $N_\gamma(\gamma)= C_\gamma^{\mathrm{DOZZ}}(\gamma,\gamma,\gamma)^{1/3}$ according to \eqref{eq-Ngamma}.  
      By \eqref{legfactor}, the right side of \eqref{eq:loop-counting} equals $C_b^{\rm ImDOZZ}(\wh\alpha_1,\wh\alpha_2,\wh\alpha_3)$ with  $b=\frac{2}{\sqrt{\kappa}}$ and $\wh\alpha_i=\frac{\alpha}{2}-b$ for $i=1,2,3$, as desired.
\end{proof}

\begin{remark}
Note that $\E[\prod_{i=1}^3\CR( \eta_i, z_i)^{\lambda_i}]$ is finite iff $\lambda_i$'s are all larger than $\frac{3\kappa}{32} -1+\frac2\kappa$, which equals $d-2$ where $d$ is the Hausdorff dimension of $\CLE_{\kappa}$ clusters. Indeed, for $\varepsilon>0$, the event $\eta_i\subset B_\varepsilon(z_i)$ means that $B_\varepsilon(z_i)$ intersects with the cluster associated with a macroscopic CLE loop surrounding $z_i$ and an other $z_j$. The probability for this event is of order $\varepsilon^{2-d}$. Therefore, to make $\E[\prod_{i=1}^3\CR( \eta_i, z_i)^{\lambda_i}]$ not explode, one needs  $\varepsilon^{\lambda_i}\cdot\varepsilon^{2-d}\to 0$ as $\varepsilon\to0$, i.e. $\lambda_i>d-2$.
\end{remark}

We end this subsection by the proof of Lemma~\ref{lem-trunc}.
\begin{proof}[Proof of Lemma~\ref{lem-trunc}]
The claims about $\wt X_\eps$ are shown in \cite[Lemma 1.7]{holden2022liouville}, and the uniform bound for $X_\eps$ follows from that of $\wt X_\eps$ and the Koebe quarter theorem. We now derive the limit for $X_\eps$.
In this proof we use $C$ to denote a positive constant whose value may change from line to line. 

Let $\mathfrak L$ be the collection of all loops $\eta_0$ surrounding 0 such that $\mathrm{CR}(\eta_0, 0) = 1$. 
Fix $\eta_0\in \mathfrak L$. Sample $\CLE_\kappa$ inside $\eta_0$ and let $\eta_1$ be the outermost loop around 0. Let $Y = -\log\mathrm{CR}(\eta_1, 0)$.
By \cite{ssw-radii} \eqb\label{eq-ssw}
    \E[e^{-s Y}] = \frac{\cos(\pi(1-\frac4\kappa))}{\cos(\pi\sqrt{(1-\frac4\kappa)^2 - \frac{8s}\kappa})} \qquad \text{ for } s \geq -1 +\frac2\kappa + \frac{3\kappa}{32} 
\eqe
and in particular $\E[\frac{n'}n e^{\Delta Y}] = 1$.

For $t>0$, let $X^{\eta_0}(t)$ be the number of  $\CLE_\kappa$ loops surrounding $0$ with distance at least $e^{-t}$ from $0$. 
Let $E$ be the event that $\eta_1$ hits $B_{e^{-t}}(0)$. 
Writing $k = n'/n$, we have the recursive relation
\[
e^{\Delta t} k^{X^{\eta_0}(t)} =  1_{E^c}  e^{\Delta Y} k \cdot e^{\Delta(t-Y)} k^{\wt X^{\wt \eta_1}(t-Y)} + 1_E e^{\Delta t}
\]
where $\wt \eta_1 = e^Y \eta_1$ is a rescaling of $\eta_1$ such that $\CR(\wt \eta_1, 0) = 1$, and $\wt X^{\wt \eta_1}(t-Y)$ counts the (rescaled) CLE loops not hitting $B_{e^{-(t-Y)}}(0)$.  Thus, taking expectations and defining 
\[ \qquad g(t) = \inf_{\eta_0 \in \mathfrak L} \E[e^{\Delta t} k^{X^{\eta_0}(t)}] \qquad \textrm{and}  \qquad  G(t) = \sup_{\eta_0 \in \mathfrak L} \E[e^{\Delta t} k^{X^{\eta_0}(t)}],\]
we obtain 
\eqb\label{eq-recur}
g(t) \geq \E [ 1_{Y \leq t} k e^{\Delta Y} g(t-Y)] \qquad \textrm{and}  \qquad G(t) \leq  \E [ 1_{Y \leq t} k e^{\Delta Y} G(t-Y)] + O(e^{-Ct}).
\eqe
Here we bound  $G(t)$ using the fact that $\E[1_E e^{\Delta t}] = e^{\Delta t}\P[E]$ decays exponentially in $t$, which follows from~\eqref{eq-ssw} and the Koebe quarter theorem.  

By Koebe distortion estimate for conformal maps,  there is a universal constant $C$ such that the following holds. If $D, D'$ are regions enclosed by $\eta_0, \eta_0'\in\mathfrak L$ and $f: D \to D'$ is the conformal map with $f(0) = 0$ and $f'(0) =1$, then for $r < 1/C$ we have $B_{r - Cr^2}(0) \subset f( B_r(0)) \subset B_{r + Cr^2} (0)$. Consequently, for any $\eta_0, \eta_0' \in \mathfrak L$, if we define $X^{\eta_0}(t)$ and $X^{\eta_0'}(t)$ where the CLE in $\eta'_0$ is that of $\eta_0$ under $f$, then $X^{\eta_0'}(t + \log (1 - Ce^{-t})) \leq X^{\eta_0}(t) \leq X^{\eta_0'}(t + \log (1 + Ce^{-t}))$. This implies
\begin{equation}\label{eq:gG}
G(t)\geq g(t) \geq (1 - O(e^{-t}))G(t).
\end{equation}

Let $M(T) = \sup_{t \leq T} G(t)$ and $q = \E[k e^{\Delta Y}1_{Y < 1}]$. Since $q<\E[k e^{\Delta Y}] = 1$, by~\eqref{eq-recur} we get $M(T) \leq O(e^{-CT}) + q M(T) + (1-q) M(T-1)$. Therefore $M(T) \leq M(T-1) + O(e^{-CT})$, so $\sup_{t \geq 0} G(t) < \infty$. Applying the same argument to $g$, we get
\eqb\label{eq-limit-bounded}
\liminf_{t \to \infty} g(t) > 0 \quad\textrm{and}\quad \sup_{t\geq0} G(t)< \infty.
\eqe
Combining with~\eqref{eq-recur} and~\eqref{eq:gG}, we obtain 
$G(t) \leq  \E[1_{Y \leq t} k e^{\Delta Y} g(t-Y)] + O(e^{-Ct})$,
and hence 
\[ g(t) = \E[1_{Y \leq t} k e^{\Delta Y} g(t-Y)] + O(e^{-Ct}).\]
Let $p: (0, \infty) \to (0,\infty)$ be the probability density function of $Y$, and let $\cL_g(s) = \int_0^\infty e^{-st} g(t)\,dt$ and $\cL_p(s) = \int_0^\infty e^{-st} p(t)\,dt=\E[e^{-sY}]$ be the Laplace transforms. 
Then for $\Re s > -C$ we have \[\cL_g(s) = k\cL_p(s - \Delta) \cL_g(s) + e(s) \]
where $e$ is an analytic function on $\{\Re s > -C\}$. Rearranging gives $\cL_g (s) = \frac{e(s)}{1 - k \cL_p(s - \Delta)}$. From \eqref{eq-ssw} the denominator $1 - k \cL_p(s - \Delta)$ has only a single root in $\{ \Re s \geq 0\}$, which is $s = 0$. Thus $\cL_g$ has no poles in $\{ \Re s \geq 0\} \backslash \{ 0\}$ and has at most a single pole at $s=0$. By the final value theorem implies the existence of the limit $c(\Delta):=\lim_{t\to\infty} g(t)$. By ~\eqref{eq-limit-bounded}, we have $c(\Delta)>0$.
By~\eqref{eq:gG} $\lim_{t\to\infty} G(t)=c(\Delta)$ as well. Therefore $\lim_{t\to\infty}\E[e^{\Delta t} k^{X^{\eta_0}(t)}]=c(\Delta)$ for all $\eta_0\in\mathfrak L$, as desired.
\end{proof}

\subsection{Reweighting and conformal welding with generic insertions}\label{subsec-general-QP}
The remaining sections are devoted to proving Theorem~\ref{thm:conformal radius}. We first restrict to the case $\kappa\in(\frac{8}{3},4)$ in this section. Let $\gamma=\sqrt{\kappa}$, and fix $(z_1, z_2, z_3) = (0,1,e^{i\pi/3})$, and let $\eta_i$ be the outermost loop of $\Gamma$ separating $z_i$ from the other two points, as defined in the paragraph above Lemma~\ref{lem-main-limit}.
We write $\sm_3$  for the law of the triples of loops $(\eta_1, \eta_2, \eta_3)$. For $\alpha_1,\alpha_2,\alpha_3\in \R$, let 
$\sm_3^{\alpha_1, \alpha_2, \alpha_3}$ be the following reweighting of $\sm_3$:
\begin{equation}\label{eq:loop-weight}
	\frac{d \sm_3^{\alpha_1, \alpha_2, \alpha_3}}{d \sm_3}(\eta_1,\eta_2,\eta_3) = \prod_{i=1}^3 \left(\frac12\CR(\eta_i, z_i)\right)^{\lambda_i}\quad \textrm{where }\lambda_i = -\frac{\alpha_i^2}2 + Q \alpha_i  - 2.
\end{equation}
By definition, the total mass of $\sm_3^{\alpha_1, \alpha_2, \alpha_3}$  is 
\begin{equation}\label{key}
	|\sm_3^{\alpha_1, \alpha_2, \alpha_3}| = \E\left[\prod_{i=1}^3 \left(\frac12\CR(\eta_i, z_i)\right)^{\lambda_i}\right]=\frac{\CCLE}{2^{\lambda_1+\lambda_2+\lambda_3}}.
\end{equation}

As outlined in Section~\ref{subsec:method}, we consider the Liouville field on the sphere with three insertions $\LF_\C^{(\alpha_i, z_i)_i}$.
Then by~\eqref{key} we have
\begin{lemma}\label{lem-sphere-laplace}
	For $\alpha_1,\alpha_2,\alpha_3$ satisfying the Seiberg bound~\eqref{eq-seiberg} in Theorem~\ref{prop-DOZZ}, we have 
	\begin{equation}\label{eq:prod-DOZZ}
		\left(\LF_\C^{(\alpha_i, z_i)_i} \times \sm_3^{\alpha_1, \alpha_2, \alpha_3}\right) [e^{-\mu_\phi(\C)}] =2^{-\lambda_1-\lambda_2-\lambda_3-1}  C^\mathrm{DOZZ}_\gamma(\alpha_1,\alpha_2,\alpha_3) \CCLE.
	\end{equation}
\end{lemma}
\begin{proof}
	Theorem~\ref{prop-DOZZ} yields  $\LF_\C^{(\alpha_i, z_i)_i} [e^{-\mu_\phi(\C)}] =\frac12C^\mathrm{DOZZ}_\gamma(\alpha_1,\alpha_2,\alpha_3)$. Now~\eqref{key} gives~\eqref{eq:prod-DOZZ}. 
\end{proof}

The product measure 
$\LF_\C^{(\alpha_i, z_i)_i} \times \sm_3^{\alpha_1, \alpha_2, \alpha_3}$ can be obtained from  conformally welding the quantum pair of pants  $\QP(\ell_1,\ell_2,\ell_3)$ and $\cM_{1,0}^\disk(\alpha_i;\ell_i)$ as in Theorem~\ref{thm-weld-QA-3}, 
which will allow us to compute $\CCLE$ for $\kappa\in (8/3,4)$ in Section~\ref{subsec:area-matching}.
\begin{theorem}\label{prop-weight-QA-3}
	Let $\alpha_i \in (Q - \frac\gamma4, Q)$ for $i=1,2,3$. 
	Let $(\phi,\eta_1,\eta_2,\eta_3)$ be sampled from $\LF_\C^{(\alpha_i, z_i)_i}\times \sm_3^{\alpha_1, \alpha_2, \alpha_3}$. 
	Then the law of the decorated quantum surface $(\C, \phi, \eta_1, \eta_2, \eta_3,  z_1,z_2,z_3)/{\sim_\gamma}$ is
	\alb
	\frac{\gamma^2}{4\pi^4(Q-\gamma)^4}\iiint_0^\infty \ell_1 \ell_2 \ell_3 \mathrm{Weld}\Big( \cM_{1,0}^\disk(\alpha_1;\ell_1), \cM_{1,0}^\disk(\alpha_2;\ell_2), 
	\cM_{1,0}
	^\disk(\alpha_3; \ell_3), \QP(\ell_1, \ell_2, \ell_3) \Big)\, d \ell_1 \, d\ell_2\, d\ell_3.
	\ale
	Here $\mathrm{Weld}$ means uniform conformal welding in the same sense as in Theorem~\ref{thm-weld-QA-3}.
\end{theorem}

We fix the range $(Q - \frac\gamma4, Q)$ in Theorem~\ref{prop-weight-QA-3} for concreteness since this is the range we will use to prove Theorem~\ref{thm:conformal radius}.
Since  $\gamma=\sqrt{\kappa}\in (\sqrt{8/3},2)$, we have $\gamma\in (Q - \frac\gamma4, Q)$. 
We first observe that Theorem~\ref{thm-weld-QA-3} is the special case of Theorem~\ref{prop-weight-QA-3} where $\alpha_1=\alpha_2=\alpha_3=\gamma$.
\begin{lemma}\label{prop-speical}
	Theorem~\ref{prop-weight-QA-3} holds for $\alpha_1=\alpha_2=\alpha_3=\gamma$.
\end{lemma}
\begin{proof}
	By  the relation between $\LF_\C^{(\gamma, z_i)_i}$ and $\QS_3$ from
	Definition~\ref{def-QS-2} and the relation between $\QD_{1,0}(\ell) $ and $\cM_{1,0}^\disk(\gamma; \ell)$ from Definition~\ref{def-QD-alpha}, we see that 
	Lemma~\ref{prop-speical} follows from Theorem~\ref{thm-weld-QA-3}.
\end{proof}

We first explain how to
go from $\alpha_1=\alpha_2=\alpha_3=\gamma$ to the following case.
\begin{proposition}\label{prop-goal}
	Theorem~\ref{prop-weight-QA-3} holds for $\alpha_1=\alpha\in  (Q - \frac\gamma4, Q)$ and $\alpha_2=\alpha_3=\gamma$.
\end{proposition}

The proof of Proposition~\ref{prop-goal} is based on the reweighting argument as in  the proof of Proposition \ref{lem-G2-alpha} in Section \ref{section:reweight}.
For notational simplicity for $\ell>0$ we let $\mathfrak M_\ell$ be the law of the  decorated quantum surfaces corresponding to 
\[\frac{\gamma^2}{4\pi^4(Q-\gamma)^4}\iint_0^\infty \ell_2 \ell_3\mathrm{Weld}( \cM_{1,0}^\disk(\gamma; \ell_2), \cM_{1,0}^\disk(\gamma; \ell_3), \QP(\ell, \ell_2, \ell_3)) \, d\ell_2 \, d\ell_3,\]
so that the relevant integral for Proposition~\ref{prop-goal} is $\int_0^\infty \ell \mathrm{Weld}( \cM_{1,0}^\disk(\alpha; \ell),\mathfrak M_\ell )\, d\ell $.
We sample a decorated quantum surface from 
$	\int_0^\infty \ell\mathrm{Weld}( \cM_{1,0}^\disk(\alpha; \ell),\mathfrak M_\ell )\, d\ell $
and let $(\C, \phi, \eta_1, \eta_2, \eta_3, z_1,z_2,z_3)$ be its embedding;  since we specify the locations of the three marked points, the tuple $(\phi,\eta_1,\eta_2,\eta_3)$ is uniquely specified by the decorated quantum surface. 

Recall the notations in Lemma~\ref{lem:reweight}.
Let $D_{\eta_1}$ and $D^c_{\eta_1}$ be the connected components of $\C \backslash \eta_1$ such that $D_{\eta_1}$  contains $z_1$. 
Let $p\in \eta_1$ be a point sampled from the harmonic measure of $\bdy D_{\eta_1}$ viewed from $z_1$ and set 
$\cD_1=(D_{\eta_1}, \phi,z_1,p )/{\sim_\gamma}$. 
Let $\psi: \bbH\to D_{\eta_1}$ be the conformal map with $\psi(i) = z_1$ and $\psi(0) = p$.
Let $X = \phi\circ \psi + Q\log |\psi'|$ so that $\cD_1=(\bbH,X,i,0)/{\sim_\gamma}$.
Let $\cD^c_1$ be the decorated quantum surface $(D^c_{\eta_1}, \phi,\eta_2, \eta_3,z_2,z_3,p)/{\sim_\gamma}$. See Figure~\ref{fig-3-point} for an illustration.
The next lemma describes the  law of $(X, \cD^c_1)$.
\begin{lemma}\label{lem:XDc}
	Given a decorated quantum surface $\cS$ sampled from $\mathfrak M_\ell$,  we write $\cS^\bullet$ as the decorated quantum surface obtained by further sampling a point on the boundary of $\cS$ according to the probability measure  proportional to its quantum boundary length measure. 
	Let $\mathfrak M_\ell^\bullet$ be the law of $\cS^\bullet$.  Then the  law of $(X, \cD^c_1)$ defined right above is $\int_0^\infty  \ell\LF_\bbH^{(\alpha,i)} (\ell) \times \mathfrak M_\ell^{\bullet}d\ell$.
\end{lemma}
\begin{proof}
	By the definition of uniform welding, conditioning on $\cD_1$, the conditional law of $\cD^c_1$ is the probability measure proportional to $\mathfrak M_\ell^\bullet$ with $\ell$ being the boundary length of $\cD_1$.
	Since the marked boundary point of $\cD_1$ is sampled from the harmonic measure,  Lemma~\ref{lem:har} now yields Lemma~\ref{lem:XDc}.
\end{proof}

\begin{figure}[ht!]
	\begin{center}
		\includegraphics[scale=0.4]{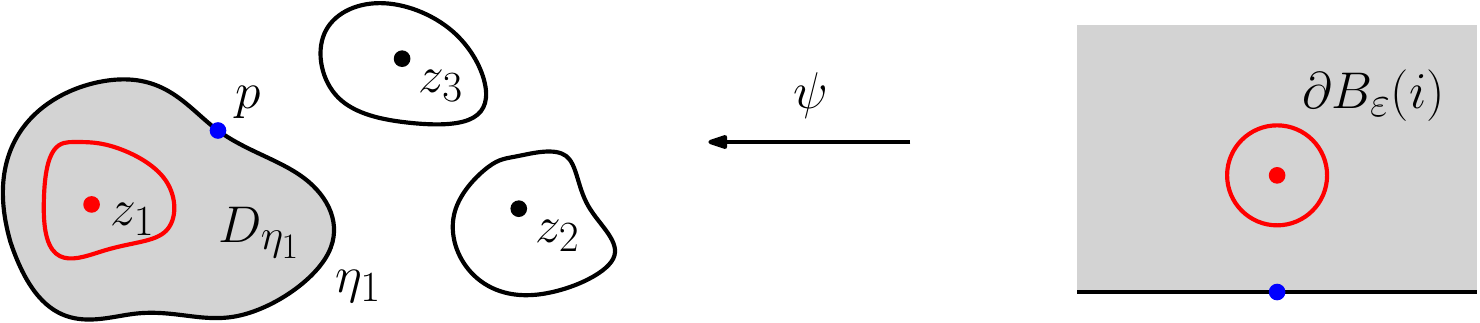}
	\end{center}
	\caption{\label{fig-3-point}   Illustration for Lemma~\ref{lem:XDc} and the proof of Proposition~\ref{prop-goal}.  The region $D_{\eta_1}$ surrounded by $\eta_1$ is colored grey. Both $(D_{\eta_1},\phi, z_1,p)$ and $(\bbH,X,0,i)$ are embeddings of $\cD_1$, which are related by $\psi$.  The domain $\C_{\eta_1,p,\eps}$ equals $\C\setminus \psi(B_\eps(i))$, the region outside the red curve on the left.
	} 
\end{figure}
\begin{proof}[Proof of Proposition~\ref{prop-goal}]
	In the setting of Lemma~\ref{lem:XDc} with $\alpha=\gamma$, by Lemma~\ref{prop-speical}
	the law of $(\phi, \eta_1, \eta_2, \eta_3,p)$ is \(\mathrm{Harm_{z_1,\eta_1}}(dp)\,\LF_\C^{(\gamma, z_1), (\gamma, z_2), (\gamma, z_3)}  (d\phi) \,\sm_3(d\eta_1,d\eta_2,d\eta_3)\),
	where $\mathrm{Harm}_{z_1,\eta_1}$ means the harmonic measure on $\eta_1$ viewed from $z_1$. Therefore
	Lemma~\ref{lem:XDc} with $\alpha=\gamma$  can be written as 
	\begin{equation}\label{eq:gamma}
		\LF_\C^{(\gamma, z_1), (\gamma, z_2), (\gamma, z_3)} (d\phi)\,\mathrm{Harm_{z_1,\eta_1}}(dp)\, \sm_3(d\eta_1,d\eta_2,d\eta_3)= \int_0^\infty  \ell\LF_\bbH^{(\gamma,i)} (\ell) \times \mathfrak M_\ell^{\bullet}d\ell,
	\end{equation}
	in the sense that $(\phi, \eta_1, \eta_2, \eta_3,p)$ and $(X, \cD^c_1)$  determine each other, and the two sides of~\eqref{eq:gamma} give the laws of $(\phi, \eta_1, \eta_2, \eta_3,p)$ and $(X, \cD^c_1)$, respectively.

	Now we use Lemma \ref{lem-disk-reweight} and Lemma \ref{lem:reweight} to prove the result from \eqref{eq:gamma}. For $\eps\in (0,\frac14)$, let $\C_{\eta_1,p,\eps}=\C\setminus \psi(B_\eps(i))$. For any nonnegative  measurable function $f$ of $\phi|_{\C_{\eta_1,p,\eps}}$, and any nonnegative measure function $g$ of $(\eta_1,\eta_2,\eta_3)$, we get from~\eqref{eq:gamma} that
	\begin{align}
		&\int f(\phi|_{\C_{\eta_1,p,\eps}})g(\eta_1,\eta_2,\eta_3) \eps^{\frac12(\alpha^2 - \gamma^2)} e^{(\alpha - \gamma)X_\eps(i)} \LF_\C^{(\gamma, z_1), (\gamma, z_2), (\gamma, z_3)}(d\phi)\, \mathrm{Harm_{z_1,\eta_1}}(dp)\, \sm_3(d\eta_1,d\eta_2,d\eta_3) \nonumber\\
		&= \int_0^\infty \left(\int f(\phi|_{\C_{\eta_1,p,\eps}})g(\eta_1,\eta_2,\eta_3) \eps^{\frac12(\alpha^2 - \gamma^2)} e^{(\alpha - \gamma)X_\eps(i)}  \ell\LF_\bbH^{(\gamma,i)} (\ell) \times \mathfrak M_\ell^{\bullet}\right)d\ell. \label{eq:gamma-eps}
	\end{align} 
	Recall that \(\sm_3^{\alpha,\gamma,\gamma} =   \left(\frac12\CR(\eta_1, z_1)\right)^{-\frac{\alpha^2}2 + Q \alpha  - 2}  \sm_3\).
 By Lemma~\ref{lem:reweight}, the left side of~\eqref{eq:gamma-eps} equals 
	\begin{align}
		&\int f g  \left(\frac12\CR(\eta_1, z_1)\right)^{-\frac{\alpha^2}2 + Q \alpha  - 2} \LF_\C^{(\alpha, z_1), (\gamma, z_2), (\gamma, z_3)}(d\phi)\, \mathrm{Harm_{z_1,\eta_1}}(dp)\, \sm_3(d\eta_1,d\eta_2,d\eta_3) \nonumber\\
		= & \int f g \LF_\C^{(\alpha, z_1), (\gamma, z_2), (\gamma, z_3)} (d\phi)\mathrm{Harm_{z_1,\eta_1}}(dp)\, \sm_3^{\alpha,\gamma,\gamma}\,(d\eta_1,d\eta_2,d\eta_3). \label{eq:compare1}
	\end{align}
    By Lemma~\ref{lem-disk-reweight}, the right side of~\eqref{eq:gamma-eps} corresponds to changing the weight of the marked point on the Liouville field, and equals 
\begin{equation}\label{eq:compare2}
	\int_0^\infty \left(\int fg \ell\LF_\bbH^{(\alpha,i)} (\ell) \times \mathfrak M_\ell^{\bullet}\right)d\ell.
\end{equation}
 
Since $\eps,f,g$  are arbitrary, comparing~\eqref{eq:compare1} and~\eqref{eq:compare2}  we get  in the same sense as in~\eqref{eq:gamma} that 
	\[
	\LF_\C^{(\alpha, z_1), (\gamma, z_2), (\gamma, z_3)} (d\phi)\mathrm{Harm_{z_1,\eta_1}}(dp)\, \sm_3^{\alpha,\gamma,\gamma}(d\eta_1,d\eta_2,d\eta_3)= \int_0^\infty  \ell\LF_\bbH^{(\alpha,i)} (\ell) \times \mathfrak M_\ell^{\bullet}d\ell.
	\] Forgetting about the point $p$ we get Proposition~\ref{prop-goal}.
\end{proof}

\begin{proof}[Proof of Theorem~\ref{prop-weight-QA-3}]
	The case $(\alpha_1, \alpha_2, \alpha_3) = (\alpha, \gamma, \gamma)$ was proved in Proposition~\ref{prop-goal}
	by reweighting from the  $(\gamma, \gamma, \gamma)$ case. The exact same argument allows us to obtain the general $(\alpha_1, \alpha_2, \alpha_3) $ case from the $(\gamma,\gamma,\gamma)$ case by reweighing around the three points sequentially.
\end{proof}

\subsection{Matching the quantum area: proof of Theorem~\ref{thm:conformal radius} for \texorpdfstring{$\kappa\in(\frac{8}{3},4]$}{g}}\label{subsec:area-matching}
Theorem~\ref{thm:conformal radius} for $\kappa\in(\frac{8}{3},4]$ is then an immediate consequence of Lemma~\ref{lem-sphere-laplace} and the following.
\begin{proposition} \label{prop-Z3}
Fix $(z_1, z_2, z_3) = (0,1,e^{i\pi/3})$.
There is a constant $C = C(\gamma) \in (0, \infty)$ such that for $\alpha_1, \alpha_2, \alpha_3 \in (Q - \frac\gamma4, Q)$, and $\sm_3^{\alpha_1, \alpha_2, \alpha_3}$ as defined in~\eqref{eq:loop-weight}  we have
\begin{equation}\label{eq:prop-3pt}
\left(\LF_\C^{(\alpha_i, z_i)_3} \times \sm_3^{\alpha_1, \alpha_2, \alpha_3}\right) [e^{- \mu_\phi(\C)}] = 
C \prod_{i=1}^3 \frac{2^{\alpha_i^2/2 - Q\alpha_i} \Gamma(\frac{\gamma\alpha_i}2 -\frac{\gamma^2}4) }{\Gamma(\frac2\gamma(Q-\alpha_i)) \cos(\frac{2\pi}\gamma (Q-\alpha_i))} \left( \frac{\pi\Gamma(\frac{\gamma^2}4)}{\Gamma(1 - \frac{\gamma^2}4)}\right)^{-\frac{\alpha_i}\gamma}.
\end{equation}
\end{proposition}
\begin{proof} 
	Recall $\QP_3(\ell_1, \ell_2, \ell_3)[e^{-A}]$ from Theorem~\ref{thm-QA3-laplace}.  Recall $ \cM_{1,0}^\disk(\alpha_i; \ell_i)[e^{-A_i}]$ from  Theorem~\ref{thm-FZZ}, where $A_i$ is the quantum area of a sample from $\cM_{1,0}^\disk(\alpha_i; \ell_i)$.
	By Theorem~\ref{prop-weight-QA-3}, 	for some $\gamma$-dependent constants $C_1,C_2,C_3$, we have that
	$\left(\LF_\C^{(\alpha_i, z_i)_3} \times \sm_3^{\alpha_1, \alpha_2, \alpha_3}\right) [e^{- \mu_\phi(\C)}]$ equals	\alb
	C_1& \iiint_0^\infty \ell_1\ell_2\ell_3 \QP_3(\ell_1, \ell_2, \ell_3)[e^{-A}] \prod_{i=1}^3 \cM_{1,0}^\disk(\alpha_i; \ell_i)[e^{-A_i}] \, d\ell_1\, d\ell_2 \, d\ell_3 \\
	&= C_2   \prod_{i=1}^3 \frac{\ol U(\alpha_i) (4 \sin (\frac{\pi \gamma^2}4))^{\alpha_i/\gamma}}{2^{\alpha_i^2/2}\Gamma(\frac2\gamma(Q-\alpha_i))} \int_0^\infty \frac1{\sqrt{\ell_i}}e^{-\ell_i \sqrt{\frac{1}{\sin(\pi\gamma^2/4)}}} K_{\frac2\gamma(Q-\alpha_i)} \left( \ell_i\sqrt{\frac{1}{\sin(\pi\gamma^2/4)}} \right) \, d\ell_i \\
		&= C_3   \prod_{i=1}^3 \frac{\ol U(\alpha_i) (4 \sin (\frac{\pi \gamma^2}4))^{\alpha_i/\gamma}}{2^{\alpha_i^2/2}\Gamma(\frac2\gamma(Q-\alpha_i))  \cos (\frac{2\pi}\gamma(Q-\alpha_i) )} ,
	\ale
	where $\ol U$ is as in~\eqref{eq:U0-explicit} and Lemma~\ref{lem:integration-fun}  is used    to evaluate the integral on the second line.  Expanding $\ol U(\alpha_i)$ from~\eqref{eq:U0-explicit}  and substituting $\sin (\frac{\pi \gamma^2}4)$
	by $\frac{\pi}{\Gamma(\frac{\gamma^2}4)\Gamma(1-\frac{ \gamma^2}4)}$, we get~\eqref{eq:prop-3pt}.
\end{proof}

\begin{proof}[Proof of Theorem~\ref{thm:conformal radius}, {$\kappa\in(\frac{8}{3},4]$}]	
	We divide the proof into four cases depending on the parameter range. 
	
	\medskip 
	\noindent \textbf{Case I: $\kappa < 4$ and $\lambda_i \in (\frac{3\kappa}{32} - 1 + \frac2\kappa, \frac\kappa8 -1+\frac2\kappa)$ for all $i=1,2,3$.} 
	In this case we can find $\alpha_i \in (Q - \frac\gamma4, Q)$ satisfying $\lambda_i = -\frac{\alpha_i^2}2 + Q \alpha_i - 2$ for each $i$. Comparing Lemma~\ref{lem-sphere-laplace} and  Proposition~\ref{prop-Z3} yields that for some $\gamma$-dependent constant $C(\gamma)$  we have 
	\[\CCLE = \frac{C(\gamma)}{C_\gamma^{\mathrm{DOZZ}}(\alpha_1,\alpha_2,\alpha_3)} \prod_{i=1}^3 \frac{ \Gamma(\frac{\gamma\alpha_i}2 -\frac{\gamma^2}4) }{\Gamma(\frac2\gamma(Q-\alpha_i)) \cos(\frac{2\pi}\gamma (Q-\alpha_i))} \left( \frac{\pi\Gamma(\frac{\gamma^2}4)}{\Gamma(1 - \frac{\gamma^2}4)}\right)^{-\frac{\alpha_i}\gamma} .\]
	Now we can evaluate $C(\gamma)$ by setting $\alpha_1=\alpha_2=\alpha_3=\gamma$, in which case $\lambda_i = 0$ and $C_\kappa^{\CLE}(0,0,0) = 1$. This gives~\eqref{eq:3ptCLE} for $\CCLE$ and~\eqref{eq-Ngamma} for the factor $N_\gamma(\alpha)$. 
\medskip

	For the next two cases set $(z_1, z_2, z_3) = (0,1,e^{i\pi/3})$ so that $ C_\kappa^{\CLE}(\lambda_1, \lambda_2, \lambda_3)= \E[\prod_{i=1}^3 (\CR(\eta_i, z_i))^{\lambda_i}] $. Since the distance between $z_i$ and $\eta_i$ is less than $1$, Koebe 1/4 theorem yields $\frac14\CR(\eta_i, z_i) \leq 1$. Thus $ \E[\prod_{i=1}^3 ( \frac14\CR(\eta_i, z_i))^{\lambda_i}] \le \E[\prod_{i=1}^3 ( \frac14\CR(\eta_i,z_i))^{\wt \lambda_i}] $ 	for any real $\lambda_i, \wt \lambda_i$ such that  $\lambda_i \geq \wt \lambda_i$. Namely
	\begin{equation}\label{eq:mono}
		4^{-\lambda_1 - \lambda_2 - \lambda_3}C_\kappa^{\CLE}(\lambda_1, \lambda_2, \lambda_3) \le 4^{-\wt \lambda_1 - \wt \lambda_2 - \wt \lambda_3}C_\kappa^{\CLE}(\wt \lambda_1,\wt  \lambda_2,\wt  \lambda_3).
	\end{equation}	
	\medskip
	\noindent \textbf{Case II: $\kappa < 4$ and all $\lambda_i > \frac{3\kappa}{32} - 1 + \frac2\kappa$.} 
	By~\eqref{eq:mono} on $\{ (\lambda_1,\lambda_2,\lambda_3)\in \C^3: \Re \lambda_i > \frac{3\kappa}{32} -1+\frac2\kappa\}$, the function $\CCLE$ is finite, hence analytic. 
	On the other hand, the right hand side of~\eqref{eq:3ptCLE} is an explicit meromorphic function in $(\alpha_1,\alpha_2,\alpha_3)$.  Now, Case II follows from Case I.

	\medskip
	\noindent\textbf{Case III: $\kappa < 4$ and $\lambda_i \leq \frac{3\kappa}{32} - 1 + \frac2\kappa$ for some $i$.} 
	We will show that $\CCLE = \infty$ in this case. By the monotonicity~\eqref{eq:mono} and symmetry, it suffices to prove  $C_\kappa^{\CLE}(\lambda_1, \lambda_2, \lambda_3) = \infty$ for $\lambda_1 \leq \frac{3\kappa}{32} - 1 + \frac2\kappa$ and $\lambda_2, \lambda_3 > \frac{3\kappa}{32} -1+\frac2\kappa$. For $\lambda>\frac{3\kappa}{32} -1+\frac2\kappa$ we have 
	\[4^{-\lambda_1- \lambda_2 - \lambda_3}C_\kappa^{\CLE}(\lambda_1, \lambda_2, \lambda_3) > 4^{- \lambda - \lambda_2 - \lambda_3}C_\kappa^{\CLE}(\lambda, \lambda_2, \lambda_3), \]
	Suppose $(\alpha, \alpha_2, \alpha_3)\in (Q-\frac{\gamma}4, Q)^3$  correspond to $(\lambda, \lambda_2, \lambda_3)$. Then $\lambda\downarrow  \frac{3\kappa}{32} -1+\frac2\kappa $ means $\alpha \downarrow Q-\frac\gamma4$.
	Recall the explicit formula for $C_\kappa^{\CLE}(\lambda, \lambda_2, \lambda_3)$ in~\eqref{eq:3ptCLE} proven in Case II. Since $\lim_{\alpha \downarrow Q - \frac\gamma4}N_\gamma(\alpha) = \infty$ and $C_\gamma^\mathrm{DOZZ}(\alpha, \alpha_2, \alpha_3)\in (0,\infty)$ as $(\alpha, \alpha_2, \alpha_3)$ satisfies the Seiberg bounds~\eqref{eq-seiberg},
 we must have
	$C_\kappa^{\CLE}(\lambda, \lambda_2, \lambda_3)\rta \infty$ as $\lambda\downarrow  \frac{3\kappa}{32} -1+\frac2\kappa $.  Therefore $C_\kappa^{\CLE}(\lambda_1, \lambda_2, \lambda_3) = \infty$ as desired. 

	\medskip
	\noindent\textbf{Case IV: $\kappa =4$.} This follows from the $\kappa < 4$ case by sending $\kappa \nearrow 4$; see Lemmas~\ref{lem:cr-4} and~\ref{lem:cont4} in Appendix~\ref{kappato4}.
\end{proof}

\subsection{The three-point correlation function for \texorpdfstring{$\kappa'\in(4,8)$}{g}}\label{sec:cr-non-simple}
In this section, we explain how to prove the $\kappa'\in (4,8)$ case of Theorem~\ref{thm:conformal radius}. We first recall the definition of the generalized quantum annulus in \cite{ACSW24a} and introduce the generalized quantum pair of pants. We then compute their area and length distribution as the counterpart of Theorem \ref{thm-QA2-area} and Theorem \ref{thm-QA3-laplace} in the non-simple case.

\subsubsection{Generalized quantum annulus and generalized quantum pair of pants}\label{subsec:gqa}

The generalized quantum annulus is introduced in \cite[Definition 5.14]{ACSW24a}; we instead work with an equivalent definition stated in \cite[Proposition 5.16]{ACSW24a}. 
In the following, we will always work with generalized quantum surfaces parametrized by bounded domains. We say a loop \emph{surrounds} a point $z$ if it has nonzero winding number with respect to $z$, and in a collection of loops, the \emph{outermost loop surrounding $z$} is the loop surrounding all other loops that surround $z$.

\begin{definition}\label{def:GA}
     Let $(D, h, \Gamma, z)$ be an embedding of a sample from $\GQD_{1,0} \otimes \CLE_{\kappa'}$ and let $\eta \in \Gamma$ be the outermost loop surrounding $z$. Let $\eta^\mathrm{out}$ be the boundary of the unbounded connected component of $\C\backslash \eta$, and let $D_{\eta^\mathrm{out}}$ be the bounded connected component of $\C \backslash \eta^\mathrm{out}$. Choose $p \in \eta^\mathrm{out}$ such that  $(D \backslash D_{\eta^\mathrm{out}}, \Gamma,  p)/{\sim_\gamma}$ is measurable with respect to  $(D \backslash D_{\eta^\mathrm{out}}, \Gamma)/{\sim_\gamma}$. 
     Conditioned on 
    $(D \backslash D_{\eta^\mathrm{out}}, \Gamma, p)/{\sim_\gamma}$, sample a forested line segment $\cL$ conditioned on having quantum length (resp.\ generalized boundary length) equal to that of $\eta^\mathrm{out}$ (resp.\ $\eta$), and sample an independent $\CLE_{\kappa'}$ in each connected component of $\cL$. Glue $\cL$ to the boundary loop $\eta^\mathrm{out}$ of  $(D \backslash D_{\eta^\mathrm{out}}, \Gamma,  p)/{\sim_\gamma}$, identifying the endpoints of $\cL$ with $p$. Let $\wh \GA^{\rm d}$ be the law of the resulting generalized quantum surface, and let $\wh \GA^{\rm d}(a,b)$ be the disintegration of $\wh \GA^{\rm d}$ with respect to the outer and inner generalized boundary lengths. Let $\GA^{\rm d}(a,b) := (b\GQD_{1,0}(b))^{-1} \wh \GA^{\rm d}(a,b)$. Let $\GA(a,b)$ be the law of a sample from $\GA^{\rm d}(a,b)$ with the loops forgotten.
\end{definition}
Here is an imprecise but perhaps more intuitive description of $\wh{\GA}^\mathrm{d}$. In Definition~\ref{def:GA}, let $A \subset D$ be the set of points not surrounded by $\eta$. If one ``cuts'' some of the self-intersection points of $(A, h, \Gamma)/{\sim_\gamma}$, then 
$\wh{\GA}^\mathrm{d}$ is the law of the resulting surface. 
Definition~\ref{def:GA} sidesteps this ``cutting'' procedure: 
it defines $\wh{\GA}^\mathrm{d}$ by forgetting the forested boundary of $(A, h, \Gamma)/{\sim_\gamma}$ (which has extra self-intersection points) and resampling it so there are no extra self-intersection points. See \cite[Section 5.3]{ACSW24a} for further details.

\begin{proposition}[{\cite[Proposition 5.15]{ACSW24a}}]\label{prop-GA2-pt}
	For $a>0$, let $(D,h,\Gamma,z)$ be an embedding of a sample from $\GQD_{1,0}(a)\otimes\CLE_{\kappa'}$. Let $\eta$ be the outermost loop of $\Gamma$ surrounding $z$.
	Then the law of the decorated quantum surface $(D,h,\Gamma,\eta,z)/{\sim_\gamma}$ equals 
	\(\int_0^\infty  b \mathrm{Weld}(\GA^{\rm d}(a,b),\GQD_{1,0}(b)\otimes \CLE_{\kappa'}) \, db\).
\end{proposition}
We also know the total mass of $\GA(a,b)$ from {\cite[Definition 5.14]{ACSW24a}}.
\begin{proposition}\label{prop: ga total mass}
For $a,b>0$ and $\gamma^2=\frac{16}{\kappa'}\in (2,4)$
    \begin{equation}
    \left|\GA(a,b)\right|=  \frac{\cos(\pi (\frac{\gamma^2}4-1))}{\pi\sqrt{ab} (a+b)}.
\end{equation}
\end{proposition}
We now define the generalized quantum pair of pants in a manner parallel to Definition~\ref{def:GA}, using the pointed generalized quantum annulus $\GA^\mathrm{d}_1$ defined as follows. 
Let $(A,h,\Gamma_A)$ be an embedding of a sample from the weighted measure $\mu_h(A)\GA^\mathrm{d}$. Given $(A,h,\Gamma_A)$, sample $z \in A$ from the probability measure proportional to $\mu_h$, and let $\GA^\mathrm{d}_1$ denote the law of $(A, h, \Gamma_A, z)/{\sim_\gamma}$.

\begin{definition}\label{def: GP}
    Let $(A,h,\Gamma_A,z)$ be an embedding of a sample from $\GA^\mathrm{d}_1$. Fix a point $o\notin A$ in the bounded connected component of $\C \backslash A$.
    Let $\eta$ be the outermost loop surrounding $z$; denote its lift to the universal cover $\C\backslash\{o\}$ by $\wt\eta$. Then the unbounded connected component of the universal covering space minus $\wt\eta$ has a boundary; denote the projection of this boundary by $\eta^{\rm out}$.
    
    Let $D_{\eta^{\rm out}}$ be the interior of $\eta^{\rm out}$.
    Choose $p \in \eta^\mathrm{out}$ such that  $(A \backslash D_{\eta^{\rm out}}, \Gamma_A|_{A \backslash D_{\eta^{\rm out}}}, z,p)/{\sim_\gamma}$ is measurable with respect to  $(A \backslash D_{\eta^{\rm out}}, \Gamma_A|_{A \backslash D_{\eta^{\rm out}}}, z)/{\sim_\gamma}$. Conditioned on 
    $(A \backslash D_{\eta^{\rm out}}, \Gamma_A|_{A \backslash D_{\eta^{\rm out}}}, z,p)/{\sim_\gamma}$, sample a forested line segment $\cL$ conditioned on having quantum length (resp.\ generalized boundary length) equal to that of $\eta^\mathrm{out}$ (resp.\ $\eta$), and sample an independent $\CLE_{\kappa'}$ in each connected component of $\cL$. Glue $\cL$ to the boundary loop $\eta^\mathrm{out}$ of $(A \backslash D_{\eta^{\rm out}}, \Gamma_A|_{A\backslash D_{\eta^{\rm out}}}, z,p)/{\sim_\gamma}$, identifying the endpoints of $\cL$ with $p$.
    Let $\wh \GP^{\rm d}$ be the law of the resulting generalized quantum surface.
    
    Define $\wh \GP^{\rm d}(a,b,c)$ as the disintegration of $\wh \GP^{\rm d}$ according to all generalized boundary lengths. Let $\wt \GP^{\rm d}(a,b,c) := (c\GQD_{1,0}(c))^{-1} \wh \GP^{\rm d}(a,b,c)$. Let $\wt\GP(a,b,c)$ be the law of a sample from $\wt\GP^\mathrm{d}(a,b,c)$ with the loops forgotten.
\end{definition}

\begin{figure}[htb]
\centering
\includegraphics[width=0.32\linewidth]{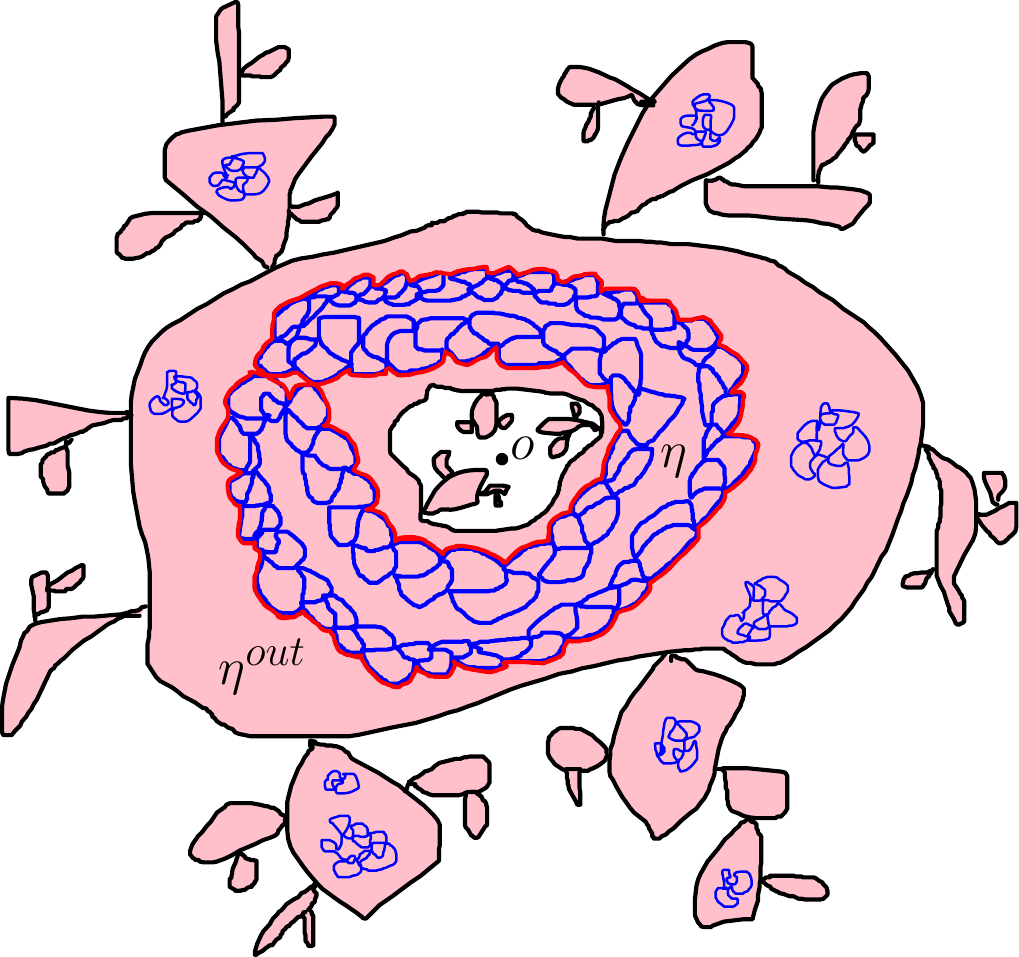}
\hspace{.02in}
\includegraphics[width=0.32\linewidth]{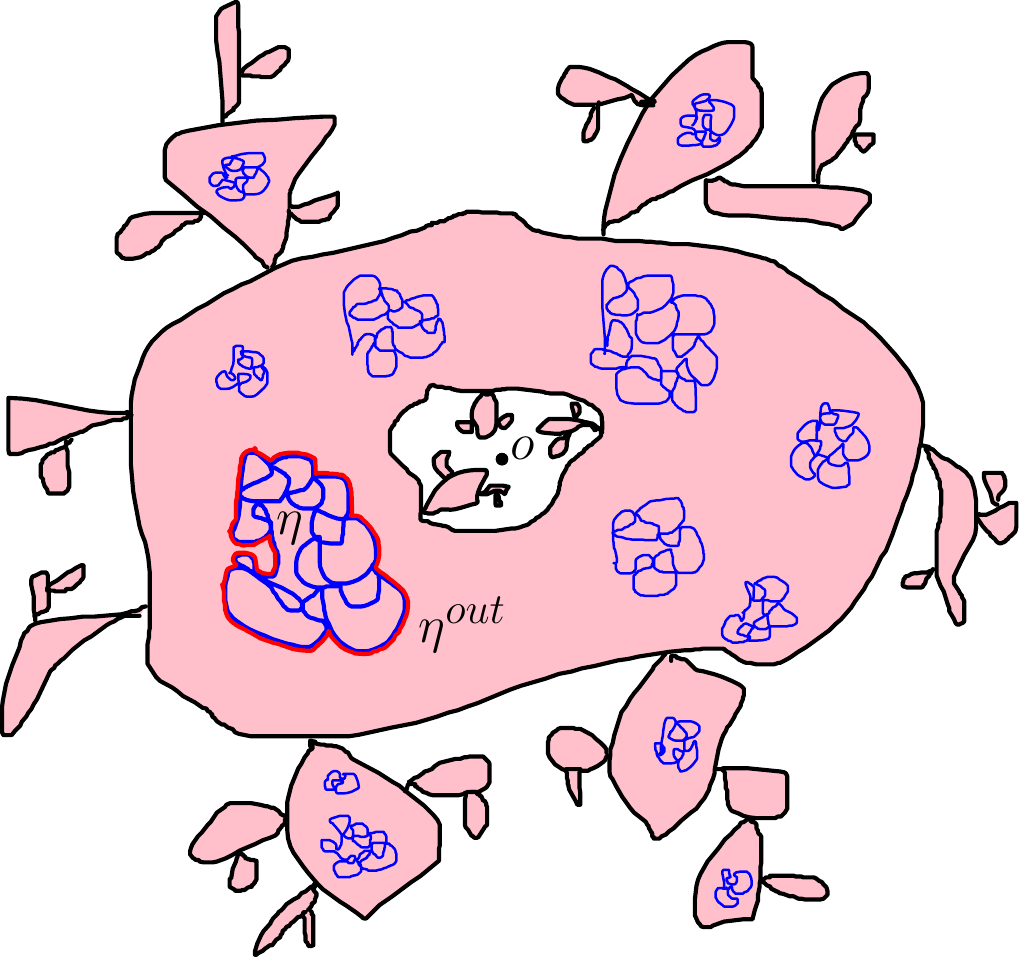}
\hspace{.02in}
\includegraphics[width=0.32\linewidth]{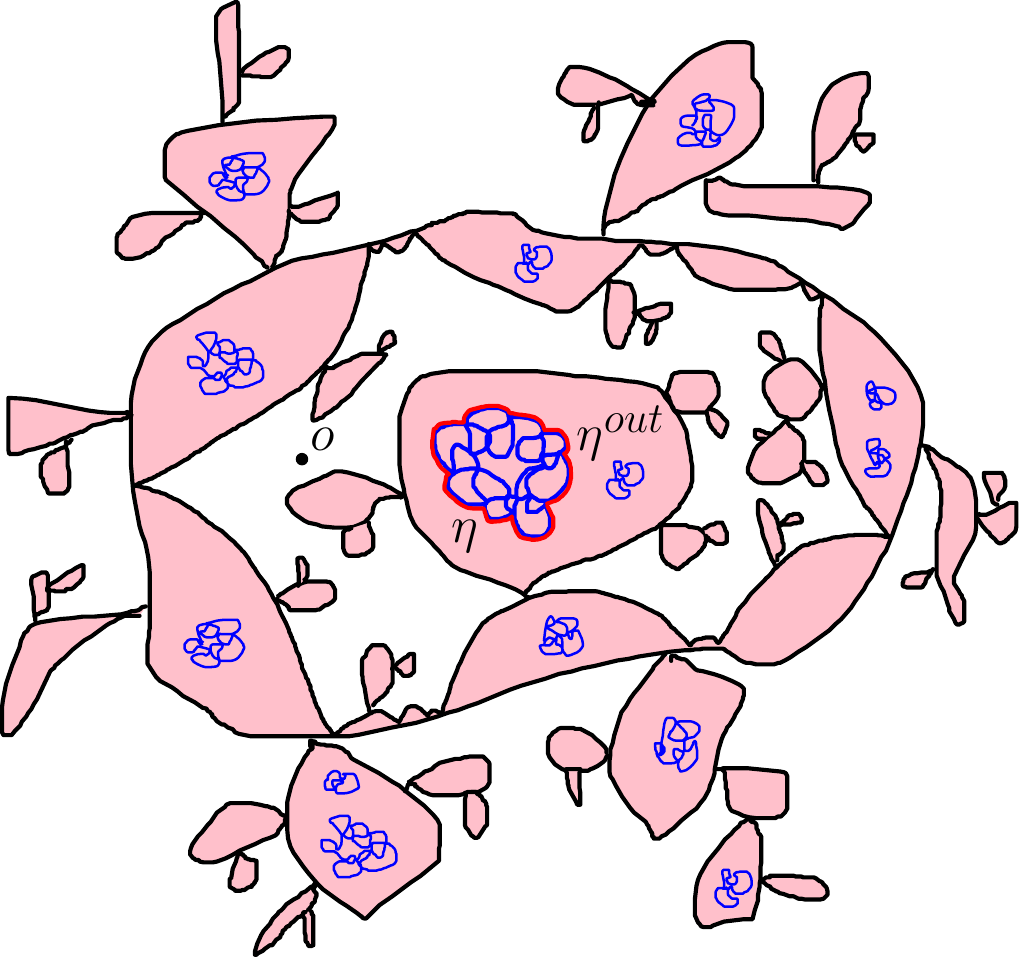}
\caption{Three cases of the generalized quantum pair of pants in Definition~\ref{def: GP}, where $\eta^{out}$ are colored red. \textbf{Left}: the outer boundary of $\eta$ is non-contractible, and $\eta^{out}$ is defined through lifting $\eta$ to the universal cover. \textbf{Middle/Right}: the outer boundary of $\eta$ is contractible and equals $\eta^{out}$.}
\label{fig:component}
\end{figure}

The existence of disintegration $\{\wt \GP^\mathrm{d}(a,b,c)\}_{a,b,c>0}$ comes from the L\'evy process description of the generalized quantum length of $\CLE_{\kappa'}$, see Proposition \ref{prop-ccm-ns}.

Similarly to Theorem~\ref{thm-weld-QA-3}, the generalized quantum pair of pants $\GP(a,b,c)$ conformally welded with three generalized quantum disks gives a quantum sphere decorated by $\CLE_{\kappa'}$. 
\begin{proposition}\label{thm-weld-GA-3}
There exists a constant $C_\gamma$ depending only on $\gamma$ such that, for $\GP(\ell_1,\ell_2,\ell_3):=C_\gamma \wt \GP(\ell_1,\ell_2,\ell_3)$, the following holds.     Sample $(\C, h, \Gamma, z_1, z_2, z_3)/{\sim_\gamma}$ from $\QS_3 \otimes \CLE_{\kappa'}$. For $i=1,2,3$,  let $\eta_i$ be the outermost loop of $\Gamma$ around $z_i$ separating it from the other two points. Then the law of the decorated generalized quantum surface $(\C, h, \eta_1, \eta_2, \eta_3,  z_1,z_2,z_3)/{\sim_\gamma}$ is	
	\eqb \label{eq-gpants-welding}
	\iiint_{\R^3_+} \ell_1\ell_2\ell_3\mathrm{Weld}\Big( \GQD_{1,0}(\ell_1), \GQD_{1,0}(\ell_2), 
	\GQD_{1,0}(\ell_3), \GP(\ell_1, \ell_2, \ell_3) \Big) \, d\ell_1 \, d\ell_2 \, d\ell_3.
	\eqe
\end{proposition}
\begin{proof}
For $i=1,2,3$, let ${D}_{\eta_i}$ be the (beaded) interior of $\eta_i$ that surround $z_i$, and ${D}_{\eta_i}^c=(\hat\C\setminus \eta_i)\setminus {D}_{\eta_i}$. Given $(h,\eta_1,\eta_2,\eta_3)$, let $p_3$ be a point sampled from the probability measure proportional to the generalized quantum length measure of $\eta_3$. By the same argument as in Proposition~\ref{prop-M2}, conditioned on $({D}_{\eta_3}^c,h,\eta_1,\eta_2,z_1,z_2,p_3)/\sim_\gamma$, the conditional law of $({D}_{\eta_3},h,z_3,p_3)/\sim_\gamma$ is $\GQD_{1,1}(\ell_3)^\#$, where $\ell_3$ is the generalized quantum length of $\eta_3$. Let $\frak{S}(\ell_3)$ be the law of the decorated generalized quantum surface $({D}_{\eta_3}^c,h,\Gamma\mid_{{D}_{\eta_3}^c},z_1,z_2,\eta_1,\eta_2)/{\sim_\gamma}$. By ~\cite[Theorem 7.4]{ACSW24a},
the law of $(\C, h, \eta_1, \eta_2, \eta_3,\Gamma,  z_1,z_2,z_3)/{\sim_\gamma}$ now equals
$$C_1\int_{\R_+}\ell_3\Weld(\GQD_{1,0}(\ell_3)\otimes \CLE_{\kappa'},\frak{S}(\ell_3))d\ell_3$$
where $C_1$ is a constant only depends on $\gamma$. 

Note that on an embedding $(D,h,\Gamma_D,z_1,z_2)$ of a sample from $\GQD_{2,0}(\ell)\otimes\CLE_{\kappa'}$, restricted on the event that there is no loop in $\Gamma_D$ surrounding both $z_1$ and $z_2$, let $\eta_1$ and $\eta_2$ be the outermost loop in $\Gamma_D$ surrounding $z_1$ and $z_2$ respectively, then the resulting law of $(D,h,\Gamma_D,z_1,z_2,\eta_1,\eta_2)/\sim_\gamma$ is $\frak{S}(\ell)$. Now suppose $(D,h,\Gamma_D,z_1)$ is an embedding of a sample from $\GQD_{1,0}(\ell)\otimes\CLE_{\kappa'}$ and $\eta_1$ be outermost loop surrounding $z_1$. Weight $(D,h,\Gamma_D,z_1,\eta_1)/\sim$ by $\mu_h(D\setminus D_{\eta_1})$, and sample $z_2\in D\setminus D_{\eta_1}$ from the probability measure proportional to the quantum area measure, and let $\eta_2$ be the outermost loop in $\Gamma_D$ surrounding $z_2$. On the one hand, the resulting law of $(D,h,\Gamma_D,z_1,z_2,\eta_1,\eta_2)/\sim$ is also $\frak{S}(\ell)$. On the other hand, from Proposition~\ref{prop-GA2-pt} and Definition~\ref{def: GP}, the resulting law of $(D,h,\Gamma_D,z_1,z_2,\eta_1,\eta_2)/\sim$ equals
$$C_2\iint_{\R^2_+} \ell_1\ell_2\mathrm{Weld}\Big( \GQD_{1,0}(\ell_1)\otimes\CLE_{\kappa'}, 
	\GQD_{1,0}(\ell_2)\otimes\CLE_{\kappa'}, \wt\GP^{\rm d}(\ell_1, \ell_2, \ell) \Big) \, d\ell_1 \, d\ell_2.$$
Combining the two indented equations and forgetting the loop decorations give the desired result (with $C_\gamma = C_1C_2$). 
\end{proof}

Now we give the quantitative description of the length and area distribution of $\GA$ and $\GP$, which is a counterpart of Theorems~\ref{thm-QA2-area} and~\ref{thm-QA3-laplace}.
\begin{theorem}\label{thm-laplace ns}
Let $\kappa'=\frac{16}{\gamma^2}\in (4,8)$ and $\ell_1,\ell_2,\ell_3>0$. Letting $A$ denote the quantum area of a sample from $\GA$ (resp.\ $\GP$) in the first (resp.\ second) equation below, we have 
\eqb\label{eq-GA-area}
\GA(\ell_1, \ell_2)[e^{-\mu A}] = \frac{\cos(\pi (\frac{\gamma^2}4-1))}\pi  \cdot \frac {e^{-2(\ell_1+\ell_2)\left(\mu/4\sin(\pi\gamma^2/4)\right)^{\kappa'/8}} }{\sqrt{\ell_1\ell_2} (\ell_1+\ell_2)} \quad \textrm{for }\mu\ge 0.
\eqe
\eqb\label{eq-GP-area}
\GP(\ell_1, \ell_2, \ell_3) [e^{-\mu A}]= C\mu^{\frac1{\gamma^2}-\frac{1}{2}}\frac{1}{\sqrt{\ell_1\ell_2\ell_3}} e^{-2(\ell_1+ \ell_2+ \ell_3)\left(\mu/4\sin(\pi \gamma^2/4)\right)^{\kappa'/8}} \quad \textrm{for }\mu>0.
\eqe
\end{theorem}
\begin{proof}
The proof uses the same L\'evy excursion computation as in the proof of  Theorems \ref{thm-QA2-area} and \ref{thm-QA3-laplace} for the simple case, so we only point out the minor changes. To obtain~\eqref{eq-GA-area}, let $\beta'=\frac{4}{\kappa'}+\frac{1}{2}$. Let $\zeta'$ be the $\beta'$-stable L\'evy process whose L\'evy measure is $1_{x>0} x^{-\beta'-1} \, dx$.
let 
$\P^{\beta'}$ denote its law, and  
for $y>0$ let $\tau_{-y}$ be the time $\zeta'$ first hits $-y$.
For $a,b>0$, let  $ (x_i)_{i \geq 1}$ be the collection of sizes of the jumps of $\zeta$ on $[0,\tau_{-a-b}]$ sorted in decreasing order. Independently of $(x_i)_{i\ge 1}$, 
	let $(A_i)_{i\ge 1}$ be independent copies of the quantum area of a sample from 
	$\GQD(1)^\#$. Then  the quantum area of a sample from $\GA(a,b)^{\#}$ has the law of 
	$\sum x_i^{8/\kappa'} A_i$.
The proof for this assertion is the same as in the simple case except that Proposition~\ref{prop-ccm} (used in the intermediate step Proposition~\ref{prop-QA2-Levy}) is replaced by Proposition~\ref{prop-ccm-ns}, and the quantum area scaling result is replaced by the following:
The quantum area of a sample from $\GQD(x)^\#$ agrees in law with $x^{8/\kappa'} A$ where $A$ is the quantum area of a sample from $\GQD(1)^\#$. This  scaling result follows from Proposition~\ref{prop:ns-area-law} instead of Proposition~\ref{prop-QD-area}. 
Now \eqref{eq-GA-area}   follows by combining Proposition~\ref{l-nonsimple} and Proposition~\ref{prop: ga total mass}.

Next, we prove~\eqref{eq-GP-area}. To that end, define $g(x)= \E[e^{-\mu x^{8/\kappa'}A_1}]$ where $A_1$ is the quantum area of a sample from $\GQD(1)^{\#}$. Let $\E^{\beta'}$ denote the expectation with respect to $\P^{\beta'}$. Using the argument of Lemma~\ref{lem-QP-translate}, there is a constant $C\in (0,\infty)$ such  that for $a,b,c>0$, we have
\begin{equation}\label{eq:GP-Levy}
\GP(a, b, c) [e^{-\mu A}]= 
C|\GA(a,b)|c^{-1/2} \E^{\beta'}[\tau_{-a-b} \prod_{i\ge 1}g(x_i)]
\end{equation}
where $(x_i)_{i\ge 1}$ is the set of jumps of $\zeta'$ that occur before time $\tau_{-a-b-c}$. Next, the proof of  Lemma~\ref{lem:Palm-Levy} carries over directly to the present setting to yield
\begin{equation}\label{eq:GP-excursion}
\E^{\beta'}[\tau_{-a-b} \prod_{i\ge 1}g(x_i)]=
C(\mu)(a+b)e^{-(a+b+c) 2(\mu/4\sin(\pi\gamma^2/4))^{\kappa'/8}}
\end{equation}
where $C(\mu)= \int T(e) \prod_{x\in J_e} g(x) \ul N (d e)$ and $\ul N (d e)$ is the excursion measure of $\zeta'$.
Combining \eqref{eq:GP-Levy} and \eqref{eq:GP-excursion} gives 
\begin{equation}\label{eq:GP-area-C}
\GP(\ell_1, \ell_2, \ell_3) [e^{-\mu A}] \propto_\gamma \frac{C(\mu)}{\sqrt{\ell_1\ell_2\ell_3}} e^{-(\ell_1+ \ell_2+ \ell_3)2\left(\mu/4\sin(\pi \gamma^2/4)\right)^{\kappa'/8}}.
\end{equation}
Finally, as in the proof of Theorem~\ref{thm-QA3-laplace}, taking the expectation of $e^{-\mu \mu_h (\C) }$ with respect to the left and right hand sides of~\eqref{eq-gpants-welding} gives $C(\mu) \propto_\gamma \mu^{\frac1{\gamma^2}-\frac{1}{2}}$. This finishes the proof of~\eqref{eq-GP-area}.
\end{proof}

\subsubsection{Proof of Theorem~\ref{thm:conformal radius} for $\kappa'\in (4,8)$.}
As in the simple case, let $(z_1, z_2, z_3) = (0,1,e^{i\pi/3})$ and $\eta_i$'s be the outermost loop in the full-plane $\CLE_{\kappa'}$ separating $z_i$ from the other two points. Denote the law of the triples of loops $(\eta_1, \eta_2, \eta_3)$ by $\sm_3$. Let $\CR(\eta_i,z_i)$ be the conformal radius of the connected component of $\C\backslash\eta_i$ containing $z_i$, viewed from $z_i$. For $(\alpha_i)_{1\le i\le3}\in \R$ and $\lambda_i = \alpha_i(Q -\frac{\alpha_i}2)  - 2$,
define $\sm_3^{\alpha_1, \alpha_2, \alpha_3}$ as
\begin{equation}\label{eq:loop-weight ns}
	\frac{d \sm_3^{\alpha_1, \alpha_2, \alpha_3}}{d \sm_3}(\eta_1,\eta_2,\eta_3) = \prod_{i=1}^3 \left(\frac12\CR(\eta_i, z_i)\right)^{\lambda_i}.
\end{equation}
Then as in Lemma \ref{lem-sphere-laplace}, for $\alpha_1,\alpha_2,\alpha_3$ satisfying the Seiberg bound~\eqref{eq-seiberg} in Proposition~\ref{prop-DOZZ}, we have
\begin{equation}\label{eq:prod-DOZZ-nonsimple}
\left(\LF_\C^{(\alpha_i, z_i)_i} \times \sm_3^{\alpha_1, \alpha_2, \alpha_3}\right) [e^{-\mu_\phi(\C)}] =2^{-\lambda_1-\lambda_2-\lambda_3-1}  C^\mathrm{DOZZ}_\gamma(\alpha_1,\alpha_2,\alpha_3) C_{\kappa'}^{\CLE} (\lambda_1,\lambda_2,\lambda_3).
\end{equation}

The measure $\LF_\C^{(\alpha_i, z_i)_i} \times \sm_3^{\alpha_1, \alpha_2, \alpha_3}$ can be obtained by conformal welding as in Theorem \ref{prop-weight-QA-3}. Recall the one-point marked generalized quantum disk $\cM_{1,0}^{\rm f.d.}(\alpha; \ell)$ in Definition \ref{def:forested-disk}. Its joint distribution of area and length is given by the following, which is essentially from \cite{holden2022liouville}.
\begin{lemma}[{\cite[Propositions 8.24, 8.25]{ACSW24a}}]\label{total mass gqd alpha}
For $\alpha>\frac{\gamma}{2}$, we have
\begin{equation*}
\begin{aligned}
|\mathcal{M}_{1,0}^{\rm f.d.}(\alpha;\ell)|=(2\pi)^{2-2\alpha/\gamma}2^{2-\alpha Q+\alpha^2/2}\frac{\Gamma\left(\frac{2\alpha}{\gamma}-\frac{4}{\gamma^2}\right)}{\Gamma\left(2-\frac{4}{\gamma^2}\right)}\Gamma\left(1-\frac{\gamma^2}{4}\right)^{2\alpha/\gamma-2}|\mathcal{M}_{1,0}^{\rm f.d.}(\gamma;1)|\ \ell^{\frac{\gamma}{2}(\alpha-Q)-1},\\
\mathcal{M}_{1,0}^{\rm f.d.}(\alpha;\ell)^{\#}[e^{-\mu A}]=\frac{2}{\Gamma(\frac{\gamma}{2}(Q-\alpha))}(\frac{M'\ell}{2})^{\frac{\gamma}{2}(Q-\alpha)}K_{\frac{\gamma}{2}(Q-\alpha)}(M'\ell)\text{ 
 with } M'=2\left(\frac{\mu}{4\sin\frac{\pi\gamma^2}{4}}\right)^{\frac{\kappa'}{8}}.
\end{aligned}
\end{equation*}
\end{lemma}

The following theorem gives the non-simple counterpart of Theorem~\ref{prop-weight-QA-3}.
\begin{theorem}\label{prop-weight-GA-3}
	Let $\alpha_i \in (Q - \frac1\gamma, Q)$ for $i=1,2,3$. 
	Let $(\phi,\eta_1,\eta_2,\eta_3)$ be sampled from $\LF_\C^{(\alpha_i, z_i)_i}\times \sm_3^{\alpha_1, \alpha_2, \alpha_3}$. 
	Then the law of the decorated quantum surface $(\C, \phi, \eta_1, \eta_2, \eta_3,  z_1,z_2,z_3)/{\sim_\gamma}$ is
	\alb
	\frac{\gamma^2}{4\pi^4(Q-\gamma)^4}\iiint_0^\infty \ell_1 \ell_2 \ell_3 \mathrm{Weld}\Big( \cM_{1,0}^{\rm f.d.}(\alpha_1;\ell_1), \cM_{1,0}^{\rm f.d.}(\alpha_2;\ell_2), 
	\cM_{1,0}^{\rm f.d.}(\alpha_3; \ell_3), \GP(\ell_1, \ell_2, \ell_3) \Big)\, d \ell_1 \, d\ell_2\, d\ell_3.
	\ale
\end{theorem}
\begin{proof}
  By Definition~\ref{def:forested-disk}, $\cM_{1,0}^\mathrm{f.d.}(\gamma)=\frac{\gamma}{2\pi(Q-\gamma)^2}\GQD_{1,0}$. 
  The proof Theorem~\ref{prop-weight-GA-3} follows the exact same steps as for Theorem~\ref{prop-weight-QA-3}, except that we use Proposition~\ref{thm-weld-GA-3} in place of Theorem~\ref{thm-weld-QA-3}, and the reweighting argument is applied to the generalized quantum disk with one interior point, which is treated in \cite[Proposition 8.26]{ACSW24a}.
\end{proof}
\begin{proof}[Proof of Theorem~\ref{thm:conformal radius}, $\kappa'\in (4,8)$.]
    Combining Theorem \ref{prop-weight-GA-3} with \eqref{eq:prod-DOZZ-nonsimple} we have
\begin{small}
\begin{equation*}
\begin{aligned}
    &\quad\quad 2^{-\lambda_1-\lambda_2-\lambda_3-1} C_\gamma^{\mathrm{DOZZ}}\left(\alpha_1, \alpha_2, \alpha_3\right) C_{\kappa'}^{\mathrm{CLE}}\left(\lambda_1, \lambda_2, \lambda_3\right)\\&=\frac{\gamma^2}{4 \pi^4(Q-\gamma)^4} \iiint_0^{\infty} \ell_1 \ell_2 \ell_3 \operatorname{Weld}\left(\cM_{1,0}^{\rm f.d.}\left(\alpha_1 ; \ell_1\right), \cM_{1,0}^{\rm f.d.}\left(\alpha_2 ; \ell_2\right),\cM_{1,0}^{\rm f.d.}\left(\alpha_3 ; \ell_3\right), \operatorname{GP}\left(\ell_1, \ell_2, \ell_3\right)\right) [e^{-\mu A}]d \ell_1 d \ell_2 d \ell_3.
\end{aligned}
\end{equation*}
\end{small}

Now the result follows by the similar calculations  as in Proposition \ref{prop-Z3}, with $\GP_3(\ell_1,\ell_2,\ell_3)[e^{-A}]$ given by Theorem \ref{thm-laplace ns} and  $\cM_{1,0}^{\rm f.d.}\left(\alpha ; \ell\right)[e^{-A}]$ given by Lemma~\ref{total mass gqd alpha}.
\end{proof}

\appendix

\section{Scaling limit of the cluster measure: proof of Proposition~\ref{prop:Rperc}}\label{percolation}
To prove Proposition \ref{prop:Rperc}, we first work on bounded domains where existing results from \cite{camia2019conformal,{cai2022natural}} can be used. 
Suppose $D\subset\C$ is a  Jordan domain with piecewise smooth boundary. Consider the critical Bernoulli percolation $\omega^\delta$ on the lattice $\delta\mathbb{T}\cap D$ with monochronic boundary condition, and let $\Gamma^\delta_D$ be the ensemble of interfaces of $\omega^\delta$. 
By \cite{camia-newman-sle6},  $\Gamma_D^\delta$ converges as  $\delta\to0$ to the $\CLE_6$ on $D$, which we denote by $\Gamma_D$.
For an open cluster $\cC^\delta$ in $\omega^\delta$, define  $\mu^\delta(\cdot;\cC^\delta)$ by the counting measure on its sites normalized by $\delta^2 \pi_\delta^{-2}(\delta,1)$.
The following  lemma, essentially from \cite{camia2019conformal,{cai2022natural}} gives the convergence of 
$\mu^\delta(\cdot;\cC^\delta)$ to the Miller-Schoug measure on the corresponding $\CLE_6$ cluster.

\begin{lemma}\label{lem:cme}
There is a constant $\frak C>0$ such that under a coupling where $\Gamma^\delta_D\to\Gamma_D$ a.s.,  the collection of measures $\{\mu^\delta(\cdot;K^\delta):  K^\delta \in \cC^\delta\}$ converges to 
the Miller-Schoug measures $\{\frak C\mu(\cdot;K): K\in \cC_D\}$, where $\cC^\delta_D$ enumerates open clusters in $\omega^\delta$ and $\cC_D$ enumerates the corresponding $\CLE_6$ clusters. Here the convergence is in the sense that for each $\varepsilon>0$, measures in the family whose clusters have diameter larger than $\varepsilon$ converge in the weak topology. 
\end{lemma}
\begin{proof}
\cite[Theorem 2.2]{camia2019conformal} gives the existence of the limit of $\{\mu^\delta(\cdot;K^\delta)\}$, and shows that the limit measures are supported on their corresponding CLE clusters.  For the outermost cluster, it was shown in \cite[Section 8.4]{cai2022natural} that $\lim_{\delta\to 0} \mu^\delta(\cdot;K^\delta_\partial)=\frak{C} \mu(\cdot;K_\partial)$ in probability for some constant $\frak C>0$ not depending on $D$, where $K^\delta_\partial$ and $K_\partial$ represent the discrete and continuum  the outermost cluster, respectively. It remains to treat the non-boundary touching clusters. Fix a Jordan domain $B$ inside $D$ with piecewise smooth boundary conditions. Then $\omega^\delta$ induces a percolation on $B\cap \delta \mathbb T$ with monochromatic boundary condition. Then the measure on its outermost cluster $K^\delta_{\partial B}$ converges in probability to $\frak C\mu(\cdot;K_{\partial B})$, where $\mu(\cdot;K_{\partial B})$ is the Miller-Schoug measure on the corresponding outermost cluster $K_{\partial B}$ for the CLE$_6$ on $B$.  
By~\cite[Theorem 1.10]{cai2022natural}, the Miller-Schoug measure satisfies the following locality property. Suppose $U$ is a measurable subset of a cluster $K\in \cC_D$, then
$\mu(U\cap B;K)=\mu(U\cap B;K_{\partial B})$. The discrete analogue of this locality is clearly true for percolation. By varying $B$, we see that  
$\lim_{\delta\to 0} \mu^\delta(\cdot;K^\delta)=\frak{C} \mu(\cdot;K)$ in probability.\qedhere
\end{proof}

Define $P_{n,D}^\delta(u_1,...,u_n)$ to be the probability that $u_1,...,u_n$ are in the same open cluster in $D$. Then by \cite[Theorem 1.1]{camia2023conformal}, $\pi_\delta^{-k}P_{n,D}^\delta(u_1,u_2,...,u_n)$ converges to a limit $P_{n,D}(u_1,...,u_n)$.  From Lemma~\ref{lem:cme} we can express $P_{n,D}(u_1,...,u_n)$ in terms of the Green function of 
$\{\mu(\cdot;K): K\in \cC_D\}$.

\begin{corollary}\label{cor:D-convergence}
Let $G^{\rm open}_{n, D}$ be such that
$G^{\rm open}_{n, D}(z_1,...,z_n)\prod_{i=1}^n dz_i= \int \sum_{K\in\cC_D}\prod_{i=1}^n\mu(dz_i;K){\rm CLE}_6^D(d\Gamma_D)$.
Then $P_{n,D}(u_1,...,u_n)=\frak{C}^n G^{\rm open}_{n, D}$, with the constant $\frak{C}$ as in Lemma~\ref{lem:cme}.
\end{corollary}
\begin{proof}
Note that
\begin{equation}\label{eq:discrete}
\pi_\delta^{-n}P_{n,D}^{\delta}(z_1^\delta,z_2^\delta,...,z_n^\delta)\prod_{i=1}^n \delta^2 dz_i^\delta=\E\left[\sum_{K^\delta\in\cC^\delta}\prod_{i=1}^n\mu^\delta(dz_i^\delta,K^\delta)\right].
\end{equation}
Suppose $U_i,1\le i\le k$ are non-intersecting open subsets in $\C$, and consider the integration of both two sides of \eqref{eq:discrete} on $U_1\times ...\times U_n$. The integration of the left hand sides converges to $\int_{U_1\times ... \times U_n} P_{n,D}(z_1,...,z_n)\prod_{i=1}^n dz_i$. By Lemma \ref{lem:cme}, the integration of the right hand side converges to $\frak{C}^nG^{\rm open}_{n, D}(z_1,...,z_n)\prod_{i=1}^n dz_i$.
\end{proof}

\begin{proof}[Proof of Proposition \ref{prop:Rperc}]
Taking  $D\nearrow \C$, we have  $P_{n,D} \nearrow P_{D} $. Moreover, $G^{\rm open}_{n, D}\nearrow \frac12G^{\rm cluster}_{n} $.
Here  the factor $\frac{1}{2}$ is due to the open-closed symmetry in percolation.
Therefore $2P_n(z_1,...,z_n)=\frak{C}^n G_n(z_1,...,z_n)$ as desired.\qedhere.
\end{proof}

\section{The continuity as \texorpdfstring{$\kappa\nearrow 4$}{g}}\label{kappato4}
In this appendix we provide detailed proof of various statements with $\kappa=4$ by taking the $\kappa\nearrow 4$ limit.
For $\kappa\in (0,4]$. Let $\Gamma_\kappa$ be the CLE$_\kappa$ on the unit disk $\D$ and $\mu_{\Gamma_\kappa}$ be the Miller-Schoug measure as defined in \cite{miller2023existence}.  Choose $c_\kappa$ such that  $c_\kappa\E[\mu_{\Gamma_\kappa}(\D)]=1$. Then by \cite[Proposition 4.5]{miller2023existence}, the measure $c_\kappa\mu_\kappa(\cdot;\Gamma_\kappa)$ weakly converges to a limit $\mu_{\Gamma_4}$, which we call 
the Miller-Schoug measure on the outermost cluster of CLE$_4$ on $\D$. When $\D$ is replaced by an arbitrary simply connected domain, 
the Miller-Schoug measure can be defined via its conformal covariance. This gives the definition of $ \cM^{\rm cluster}_i$ hence $G^{\rm cluster}$ and $R^{\rm cluster}(\kappa)$ for $\kappa=4$.
\begin{lemma}\label{lem:MS-cont}
Suppose $\kappa_n\nearrow 4$. Let $\Gamma_{\kappa_n}$ be a $\CLE_{\kappa_n}$ on a Jordan domain $D_n$, and $\Gamma_{4}$ is a $\CLE_4$ on a  Jordan domain $D$. Suppose  $D_n$ converge to $D$ in the Hausdorff topology, then the measure $\mu_{\Gamma_{\kappa_n}}$ converge in law to $\mu_{\Gamma_4}$. 
\end{lemma}
\begin{proof}
Since $D_n$ converge to $D$ in the Hausdorff topology, we can find conformal maps $f_n$ from $\D$ to $D_n$ and a conformal map $f$ from $\D$ to $D$ such that $f_n$ converge to $f$ uniformly on compact sets.
\end{proof}

\begin{lemma}\label{lem:Green-cont}
 $\lim_{\kappa\nearrow 4} R^{\rm cluster}(\kappa)= R^{\rm cluster}(4)$. 
\end{lemma}
\begin{proof}
Define $ \cM^{\rm cluster}_i$ and $G^{\rm cluster}$ as in~\eqref{eq:greendef}. In this proof we write them as $ \cM^{\rm cluster,\kappa}_{i}$ and $G^{\rm cluster,\kappa}$ to indicate the $\kappa$-dependence.  By Lemma~\ref{lem:MS-cont}, as $\kappa \nearrow 4$, the expectation $\E[\sum_{ i\ge 1 }  \int_{U_1\times \cdots U_n}   \prod_{k=1}^n c_\kappa^n  \cM^{\rm cluster,\kappa}_i(dx_i)]$ converge to  $\E[\sum_{ i\ge 1 }  \int_{U_1\times \cdots U_n} \prod_{k=1}^n \cM^{\rm cluster,4}_i(dx_i)]$, where $U_1,\cdots U_n$ are disjoint compact sets on the plane.  This is because for fixed $U_1,\cdots U_n$, there are  finitely many CLE loops intersect at least two of them, and all of them converge as $\kappa\nearrow 4$. 
Taking $n=3$, and setting $(u_1,u_2,u_3)=(0,1,e^{i\pi/3})$, we have $\int_{U_1\times U_2\times U_3}  c_\kappa^3 G^{\rm cluster, \kappa}(u_1,u_2,u_3)|z_1-z_2|^{-(2-d_\kappa)}  |z_2-z_3|^{-(2-d_\kappa)} |z_3-z_1|^{-(2-d_\kappa)} \prod_{i=1}^3dz_i$ converges to $\int_{U_1\times U_2\times U_3} G^{\rm cluster, 4}(u_1,u_2,u_3)|z_1-z_2|^{-(2-d_4)}  |z_2-z_3|^{-(2-d_4)} |z_3-z_1|^{-(2-d_4)} \prod_{i=1}^3dz_i$, where $d_\kappa=2-\frac{(3\kappa-8)(8-\kappa)}{32\kappa}$. Therefore, $\lim_{\kappa\nearrow 4} c_\kappa^3 G^{\rm cluster, \kappa}(u_1,u_2,u_3) = G^{\rm cluster, 4}(u_1,u_2,u_3)$.  Similarly,  $\lim_{\kappa\nearrow 4} c_\kappa^2 G^{\rm cluster, \kappa}_2(u_1,u_2) = G^{\rm cluster, 4}_2(u_1,u_2)$.
Since 
\[
R^{\rm cluster}(\kappa):=\frac{c^3_\kappa G^{\rm cluster}_3(u_1,u_2,u_3)}{\sqrt{c^2_\kappa G^{\rm cluster}_2(u_1,u_2) c^2_\kappa G^{\rm cluster}_2(u_1,u_3) c^2_\kappa G^{\rm cluster}_2(u_2,u_3)}},
\]we get $\lim_{\kappa\nearrow 4} R^{\rm cluster}(\kappa)= R^{\rm cluster}(4)$ as desired. 
\end{proof}
The Minkowski content measure on the CLE loops converge  as $\kappa\nearrow 4$ in a similar sense as for the Miller-Schoug measure. This was explained in \cite[Lemma A.6]{ACSW24a} using the formalism of the SLE loop measure. This gives:
\begin{lemma}\label{lem:Green-cont loop}
 $\lim_{\kappa\nearrow 4} R^{\rm loop}(\kappa)= R^{\rm loop}(4)$. 
\end{lemma}
\begin{proof}
This follows from the same proof as for Lemma~\ref{lem:Green-cont}.
\end{proof}

Next, we prove a continuity result for conformal radii used in the $\kappa=4$ case of Theorem~\ref{thm:conformal radius}.
\begin{lemma}\label{lem:cr-4}
For $z_1,z_2,z_3\in\C$ and $(\eta_i)$ be the outermost loop in the full-plane $\CLE_\kappa$ surrounding $z_i$ and separating other two $z_j$'s, the conformal radius $\{\CR(z_i, \eta_i)\}_{1\le i\le 3}$ for $\kappa<4$ converges in law as $\kappa\nearrow 4$ to $\{\CR(z_i, \eta_i)\}_{1\le i\le 3}$ with $\kappa=4$.
\end{lemma}
\begin{proof}
  Fix a small $\eps>0$.
  Let $S$ be the set of simple loops in $\wh \C$ separating $z_1$ from $\{z_2, z_3\}$. Let $S_\eps=\{ \eta\in S:  \CR(z_1,\eta)>\eps\}$. 
  For $\eta\in S$, let $D_\eta$ be the connected component of $\wh \C \backslash \eta$ containing $z_2, z_3$. 
  For $\kappa\in (8/3,4]$, the law of $(\eta_1,\Gamma)$ restricted to the event 
  $\eta_1\in S_\eps$ is the same as that of $(\eta, \Gamma)$ under $1_{\eta = \eta_1}\mathrm{Count}_{S_\eps\cap \Gamma}(d\eta)\CLE_\kappa(d\Gamma)$,	
  where $\mathrm{Count}_{S_\eps\cap \Gamma}(d\eta)$ is the counting measure on $S_\eps\cap \Gamma$.
  Given a sample of $(\eta, \Gamma)$  of $\mathrm{Count}_{S_\eps\cap \Gamma}(d\eta)\CLE_\kappa(d\Gamma)$,  let $\Gamma_+$ be the subset of $\Gamma$ in $D_\eta$ and $F_{\Gamma_+}$ be the event that no loop in $\Gamma_+$ surrounds both $z_2$ and $z_3$. Then the event $\eta = \eta_1$ is the same as  $F_{\Gamma_+}$. By the Markov property of the full-plane CLE \cite{werner-sphere-cle}, the law of $\eta$ under $1_{\eta = \eta_1}\mathrm{Count}_{S_\eps\cap \Gamma}(d\eta)\CLE_\kappa(d\Gamma)$ is 
  $\CLE_\kappa^{D_\eta}[F_{\Gamma_+}] 1_{\eta \in S_\eps}\cdot \SLE_\kappa^\mathrm{loop}(d\eta)$, where $\CLE_\kappa^{D_\eta}$ is the law of a $\CLE_\kappa$ in $D_\eta$.

  By \cite[Lemma A.3, Lemma A.5]{ACSW24a},   as $\kappa \nearrow 4$, the measure $\CLE_\kappa^{D_\eta}[F_{\Gamma_+}] 1_{\eta \in S_\eps}\cdot \SLE_\kappa^\mathrm{loop}(d\eta)$ on loops converges weakly with respect to the Hausdorff metric to $\CLE_4^{D_\eta}[F_{\Gamma_+}] 1_{\eta \in S_\eps}\cdot \SLE_4^\mathrm{loop}(d\eta)$.
  Therefore, the law of $\eta_1$ conditioned on the event $\eta_1\in S_\eps$ as $\kappa \nearrow 4$ converges weakly to the same law with $\kappa = 4$.
  Since $\lim_{\eps\to 0}\P[\Gamma\cap S_\eps \neq\emptyset]=\lim_{\eps \to 0}\P[\CR(z_1,\eta_1)>\eps]=1$ uniformly for $\kappa\in (\kappa_0,4]$ for a fixed $\kappa_0 \in (8/3,4)$, we can remove the restriction $\eta_1\in S_\eps$ and conclude that as $\kappa \nearrow 4$ the law of  $\eta_1$  converges weakly  to the same law when $\kappa=4$.
  
  The conditional law of $(\eta_2,\eta_3)$ given $\eta_1$ is the law of the  two outermost loops surrounding $z_2,z_3$ under the CLE measure $\CLE_\kappa^{D_{\eta_1}}$ conditioned on these two loops being distinct. By ~\cite[Lemma A.3]{ACSW24a},  the joint law of $(\eta_1,\eta_2,\eta_3)$ has the desired continuity as $\kappa \nearrow 4$.
\end{proof}

We also need the following continuity of special functions appeared in Theorems~\ref{thm: connectivity},~\ref{thm:nesting} and~\ref{thm:mag}.
\begin{lemma}\label{lem:cont4}
The right hand sides of \eqref{eq:carpet-formula}, \eqref{eq-main-idozz} and \eqref{eq:Rformula} are  continuous as $\kappa\nearrow 4$.
\end{lemma}
\begin{proof}
    By the shift equation for $\Upsilon_{\frac{\gamma}{2}}$ we get $\Upsilon_{\frac{\gamma}{2}}(\gamma)=\Upsilon_{\frac{\gamma}{2}}(\frac{\gamma}{2})\ell(\frac{\gamma^2}{4})(\frac{\gamma}{2})^{1-\gamma^2/2}$, hence $\Upsilon_{\frac{\gamma}{2}}(\gamma)\sim (1-\frac{\gamma^2}{4})$ $\gamma\nearrow2$. Plugging this into \eqref{eq:DOZZ} and by~\eqref{legfactor} we get the countinuity of the first two. The continuity of the right side of \eqref{eq:Rformula} follows from \eqref{eq:mag-exp} and that $C_\gamma^{\rm DOZZ}(\frac{\gamma}{2},\frac{\gamma}{2},\frac{\gamma}{2})\sim(1-\frac{\gamma^2}{4})^{1/2}$ as $\gamma\nearrow2$.
\end{proof}

\section{Integration formulas}\label{appendix:integration}
In this section, we prove some
 integral identities that are used in our paper.
\begin{proof}[Proof of Lemma~\ref{lem-int-KK}]
	\alb
	&\int_0^\infty \ell K_a(c \ell) K_{a'}(c\ell) \, d\ell = 2\iint_0^\infty \ell K_{a-a'}(2c\ell \cosh t) \cosh ((a+a')t)\, dt \, d\ell \\
	&= 2\iint_0^\infty x K_{a-a'}(x) \frac{\cosh ((a+a')t)}{(2c\cosh t)^2}\, dt \, dx = \frac1{2c^2}\int_0^\infty  x K_{a-a'}(x) \, dx \int_0^\infty  \frac{\cosh ((a+a')t)}{(\cosh t)^2}\, dt.
	\ale
	Here the first equality follows from \cite[(10.32.17)]{NIST:DLMF}, the second from the change of variables $x = 2c (\cosh t) \ell$. 
	By \cite[(10.43.19)]{NIST:DLMF} and some standard gamma function identities, we have
	\[\int_0^\infty  x K_{a-a'}(x) \, dx = \Gamma(1-\frac12(a-a'))\Gamma(1+\frac12(a-a')) = \frac12(a-a') \Gamma(1-\frac12(a-a'))\Gamma(\frac12(a-a')) =  \frac{\frac\pi2(a-a')}{\sin (\frac\pi2 (a-a'))}. \]
	By \cite[(4.40.9)]{NIST:DLMF} and the fact that $\cosh t$ is even, 
	\[\int_0^\infty \frac{\cosh((a+a')t)}{(\cosh t)^2} \, dt = \frac12\int_{-\infty}^\infty \frac{e^{(a+a')t}}{(\cosh t)^2} \, dt =  \frac{\frac\pi2(a+a')}{\sin (\frac\pi2(a+a'))}. \qedhere \]
\end{proof}

Now we prove  Lemma~\ref{lem:integration-fun} used in the proofs of Theorem~\ref{thm-QA3-laplace} as well as Proposition~\ref{eq:prop-3pt}.
\begin{lemma} \label{lem:integration-fun}
$\int_0^\infty \frac1{\sqrt x} e^{-cx} K_\nu(cx) \, dx = \frac{\pi^{3/2}}{\sqrt{2c} \cos (\pi \nu)}$ for $c>0$ and $\nu \in (-\frac12, \frac12)$.
\end{lemma}
	\begin{proof} 
	We will prove the $c=1$ case; the general case then follows from a change of variables. 
	Recall the integral definition of $K_\nu(x)$ in~\eqref{eq-Kv}.
	Using $1+\cosh t = 2(\cosh \frac{t}2)^2$ and the fact that $\cosh$ is even, we can express $\int_0^\infty \frac1{\sqrt x} e^{-x} K_\nu (x) \, dx$ as
	\[\iint_0^\infty \frac1{\sqrt x} e^{-x (1 + \cosh t)} \cosh(\nu t)\, dx \, dt = \sqrt \pi \int_0^\infty \frac{\cosh (\nu t)}{\sqrt{1+\cosh t}}\, dt = \frac12\sqrt {\pi/2} \int_{-\infty}^\infty \frac{\cosh (\nu t)}{\cosh(t/2)}\, dt.\]
Since $\int \frac{e^{\nu t}}{\cosh(t/2)} \, dt = \int \frac{e^{-\nu t}}{\cosh(t/2)} \, dt$, we have $\int \frac{\cosh (\nu t)}{\cosh(t/2)}\, dt = \int \frac{e^{\nu t}}{\cosh(t/2)}\, dt$.  
For a random variable $X$ following the \emph{hyperbolic secant distribution} $\frac{dx}{\pi \cosh x}$, it is known that $\E[e^{tX}]=\frac{1}{\cos(\pi t/2)}$ for $|t|<1$;  
see e.g.  \cite[Section 1.3]{Fischer-hypersec}. Therefore \(\int_{-\infty}^\infty \frac{e^{\nu t}}{\cosh(t/2)}\, dt = \frac{2\pi}{\cos (\pi \nu)}.\)
This concludes the proof.
\end{proof}

The hypergeometry function $_2F_1(\cdot,\cdot;\cdot;\cdot)$ is defined by
$$_2F_1(a,b;c;z) = \sum_{n=0}^\infty \frac{(a)_n (b)_n}{(c)_n} \frac{z^n}{n!},$$
where $(q)_n$ is the rising Pochhammer symbol defined by
$(q)_n = q(q+1) \cdots (q+n-1)$ for $n>0$ and $(q)_n=1$ for $n=0$. We need the following integration formula in the proofs of Lemmas~\ref{lem-6} and~\ref{lem:7}.
\begin{lemma}[{\cite[Volume 2, Section 7.7.3, Formula (26)]{Bateman1953HigherTF}}]\label{lem:bigint}
For ${\rm Re}\mu>|{\rm Re}\nu|$ and ${\rm Re}(\alpha+\beta)>0$,
\begin{equation}
    \int_0^\infty x^{\mu-1}e^{-\alpha x}K_\nu(\beta x)dx=\frac{\sqrt{\pi}(2\beta)^\nu}{(\alpha+\beta)^{\mu+\nu}}\frac{\Gamma(\mu+\nu)\Gamma(\mu-\nu)}{\Gamma(\mu+\frac{1}{2})}{}_2F_1\left(\mu+\nu,\nu+\frac{1}{2};\mu+\frac{1}{2};\frac{\alpha-\beta}{\alpha+\beta}\right).
\end{equation}
\end{lemma}

\bibliographystyle{alpha}
\footnotesize{
\def\cprime{$'$}

}

\end{document}